\documentclass[oneside,english]{amsart}
\usepackage[LGR,T1]{fontenc}
\usepackage[latin9]{inputenc}
\usepackage{geometry}
\usepackage{array}
\usepackage{booktabs}
\usepackage{units}
\usepackage{multirow}
\usepackage{amstext}
\usepackage{amsthm}
\usepackage{amssymb}

\makeatletter

\DeclareRobustCommand{\greektext}{%
  \fontencoding{LGR}\selectfont\def\encodingdefault{LGR}}
\DeclareRobustCommand{\textgreek}[1]{\leavevmode{\greektext #1}}
\ProvideTextCommand{\~}{LGR}[1]{\char126#1}

\providecommand{\tabularnewline}{\\}

\numberwithin{equation}{section}
\numberwithin{figure}{section}
\theoremstyle{plain}
\newtheorem{thm}{\protect\theoremname}
\theoremstyle{plain}
\newtheorem{prop}[thm]{\protect\propositionname}
\theoremstyle{definition}
\newtheorem{defn}[thm]{\protect\definitionname}
\theoremstyle{definition}
\newtheorem{example}[thm]{\protect\examplename}
\theoremstyle{plain}
\newtheorem{conjecture}[thm]{\protect\conjecturename}

\usepackage{nicefrac}
\usepackage{babel}
\@ifundefined{definecolor}
 {\usepackage{color}}{}
\@ifundefined{definecolor}{\usepackage{color}}{}
\usepackage{graphicx}

\usepackage[mathscr]{euscript}
 \let\mathscr\relax
\usepackage[scr]{rsfso}
\newcommand{\powerset}{\raisebox{.15\baselineskip}{\Large\ensuremath{\wp}}}

\usepackage{mathtools}
\usepackage{tikz}

\makeatother

\usepackage{babel}
\providecommand{\conjecturename}{Conjecture}
\providecommand{\definitionname}{Definition}
\providecommand{\examplename}{Example}
\providecommand{\propositionname}{Proposition}
\providecommand{\theoremname}{Theorem}

\begin{document}
\title{ON PARTIAL DIFFERENTIAL ENCODINGS OF BOOLEAN FUNCTIONS}
\author{Edinah K. Gnang, Rongyu Xu}
\begin{abstract}
We introduce partial differential encodings of Boolean functions as
a way of measuring the complexity of Boolean functions. These encodings
enable us to derive from group actions non-trivial bounds on the Chow-Rank
of polynomials used to specify partial differential encodings of Boolean
functions. We also introduce variants of partial differential encodings
called partial differential programs. We show that such programs optimally
describe important families of polynomials including determinants
and permanents. Partial differential programs also enables to quantitively
contrast these two families of polynomials. Finally we derive from
polynomial constructions inspired by partial differential programs
which exhibit an unconditional exponential separation between high
order hypergraph isomorhism instances and their sub-isomorphism counterparts.
\end{abstract}

\maketitle

\section{Introduction}

In a epoch-making book titled ``An Investigation of the Laws of Thought'',
George Boole \cite{Bool54} laid the foundations for the Boolean algebra.
This algebra serves as the first of two pillars of the computing revolution.
Interestingly, George Boole also initiated the branch of mathematics
known as invariant theory \cite{Wol08}. There is a recognition \cite{Garg19,Gro16,Aar16,Gro20}
that a rich interplay relates these seemingly separate branches both
pioneered by Boole. Invariant theory emphasizes consequences which
stem from symmetries or lack thereof. The importance of symmetries
in the analysis of Boolean functions was well known to pioneers of
the field, such as Shannon, Pólya and Redfield \cite{Sha49,Pol40,Pol37,Red27}.
We investigate in the present work, partial differential incarnations
of Turing machines. Turing machines were introduced by Alan Turing
\cite{Tur36} as a mathematical model of computation. Turing machines
are the second pillar of the computing revolution. The use of differential
operators in invariant theory is also very old. Their origin can be
traced back to the work of early pioneers of invariant theory. Most
notably to the work of George Boole, Arthur Cayley and James Joseph
Sylvester \cite{Cay45,Sylvester1852} who instigated the use of differential
operators to construct invariants of group actions. The framework
is also known as Cayley's \textgreek{W} process. In complexity theory,
differential operators were investigated in the context of arithmetic
complexity by Baur and Strassen \cite{BS83}. More recently, Cornelius
Brand and Kevin Pratt \cite{br2020algorithmic} were able to match
the runtime of the fastest known deterministic algorithm for detecting
subgraphs of bounded path-width using a method of partial derivatives.
We refer the reader to excellent recent surveys on partial differential
methods in arithmetic complexity \cite{TCS2011,TCS-039}. The importance
of investigating partial differential operators is reinforced by the
central role they play in physics and machine learning. The present
work formally ties, aspects of low depth arithmetic circuit complexity
to Boolean De Morgan circuit complexity. Recent depth reduction results
\cite{10.1007/3-540-10856-4_79,4690941,10.1145/2535928,doi:10.1137/140957123,doi:10.1137/0208010}
motivate our focus on low depth arithmetic circuits. In the present
work, we introduce partial differential encodings of Boolean functions
and their relaxations. These encodings enable us to determine the
fraction of optimal encodings. Our main result is a general method
for deriving Chow-Rank bounds of polynomial from group actions. Our
method is a discrete analog of representation theory methods which
devise bound on the border rank from Lie group actions\cite{landsberg_2017,JGrochow2015}.
We also introduce variants of partial differential encodings called
partial differential programs. We show that such programs optimally
describe important families of polynomials including determinants
and permanents. Partial differential programs also enables to quantitively
contrast these two families of polynomials. Finally we derive from
polynomial constructions inspired by partial differential programs
which exhibit an unconditional exponential separation between high
order hypergraph isomorhism instances and their sub-isomorphism counterparts.

\section{Partial Differential Encodings.}

Recall the ``needles in a haystack'' conundrum \cite{Sha49}. The
conundrum roughly translates into the observation that there are at
most $s^{3s}$ Boolean circuits (expressed in the De Morgan basis)
of size $s$ for as many Boolean functions among the $2^{\left(2^{n}\right)}$
possible Boolean functions on $n$ bits. Consequently, most circuits
have size at most $O\left(\frac{2^{n}}{n}\right)$. Unfortunately,
the argument does not bound to the size of circuit encodings of specific
Boolean functions. We circumvent this drawback by considering an algebraic
variant of the conundrum. The algebraic variant is based upon Boole's
correspondence
\begin{equation}
\begin{cases}
\begin{array}{ccc}
\text{True} & \rightarrow & 1,\\
\\
\text{False} & \rightarrow & 0,\\
\\
\neg\,x_{i} & \rightarrow & 1-x_{i},\\
\\
x_{i}\vee x_{j} & \rightarrow & x_{i}+x_{j}-x_{i}\,x_{j},\\
\\
x_{i}\wedge x_{j} & \rightarrow & x_{i}\,x_{j}.
\end{array}\end{cases}\label{Algebraic correspondence}
\end{equation}
and conversely
\begin{equation}
\begin{cases}
\begin{array}{ccc}
1 & \rightarrow & \text{True},\\
\\
0 & \rightarrow & \text{False},\\
\\
1-x_{i} & \rightarrow & \neg\,x_{i},\\
\\
x_{i}\,x_{j} & \rightarrow & x_{i}\wedge x_{j},\\
\\
\left(x_{i}+x_{j}\right)\text{mod }2 & \rightarrow & \left(x_{i}\vee x_{j}\right)\wedge\neg\left(x_{i}\wedge x_{j}\right):=\left(x_{i}\,\underline{\vee}\,x_{j}\right).
\end{array}\end{cases}\label{Boolean correspondence}
\end{equation}
To reflect the fact that the variables $\left\{ x_{i}:i\in\mathbb{Z}_{n}\right\} $\footnote{For notational convenience let $\mathbb{Z}_{n}\,:=\left[0,n\right)\cap\mathbb{Z}$.}
are Boolean, algebraic expression are often taken modulo the binary
algebraic relations
\begin{equation}
\left\{ \begin{array}{c}
\left(x_{i}\right)^{2}\equiv x_{i}\\
i\in\mathbb{Z}_{n}
\end{array}\right\} .\label{Binary Algebraic Relations}
\end{equation}

\begin{prop}
\label{Canonical_Boolean_Interpolant}An arbitrary Boolean function
\[
F:\left\{ 0,1\right\} ^{n}\rightarrow\left\{ 0,1\right\} 
\]
admits a canonical depth--3 $\sum\prod\sum$ arithmetic formula expression
over $\mathbb{Q}$ as well as a canonical depth--2 $\prod\sum$ arithmetic
formula over $\mathbb{C}$ prescribed modulo relations described in
Eq. (\ref{Binary Algebraic Relations}).
\end{prop}

\begin{proof}
By Lagrange's interpolation theorem $F\left(\mathbf{x}\right)$ admits
a unique multilinear interpolant given by
\begin{equation}
F\left(\mathbf{x}\right)=\sum_{\begin{array}{c}
\mathbf{b}\in\left\{ 0,1\right\} ^{n\times1}\\
\text{s.t. }F\left(\mathbf{b}\right)=1
\end{array}}\prod_{i\in\mathbb{Z}_{n}}\left(\prod_{d_{i}\in\left\{ 0,1\right\} \backslash\left\{ b_{i}\right\} }\left(\frac{x_{i}-d_{i}}{b_{i}-d_{i}}\right)\right)=\sum_{\begin{array}{c}
\mathbf{b}\in\left\{ 0,1\right\} ^{n\times1}\\
\text{s.t. }F\left(\mathbf{b}\right)=1
\end{array}}\prod_{i\in\mathbb{Z}_{n}}\left(\frac{x_{i}-\left(1-b_{i}\right)}{2b_{i}-1}\right).\label{first PDE}
\end{equation}
Thereby expressing the desired depth--3 $\sum\prod\sum$ formula.
Furthermore, Lagrange's interpolation construction yields the congruence
identity
\[
\sum_{\begin{array}{c}
\mathbf{b}\in\left\{ 0,1\right\} ^{n\times1}\\
\text{s.t. }F\left(\mathbf{b}\right)=1
\end{array}}\prod_{i\in\mathbb{Z}_{n}}\left(\frac{x_{i}-\left(1-b_{i}\right)}{2b_{i}-1}\right)\equiv\sum_{\begin{array}{c}
\mathbf{b}\in\left\{ 0,1\right\} ^{n\times1}\\
\text{s.t. }F\left(\mathbf{b}\right)=1
\end{array}}\prod_{\mathbf{d}\in\left\{ 0,1\right\} ^{n\times1}\backslash\left\{ \mathbf{b}\right\} }\left(\frac{\underset{j\in\mathbb{Z}_{n}}{\sum}2^{j}\,x_{j}-\underset{k\in\mathbb{Z}_{n}}{\sum}2^{k}\,d_{k}}{\underset{j\in\mathbb{Z}_{n}}{\sum}2^{j}\,b_{j}-\underset{k\in\mathbb{Z}_{n}}{\sum}2^{k}\,d_{k}}\right)\text{ mod}\left\{ \begin{array}{c}
\left(x_{i}\right)^{2}-x_{i}\\
i\in\mathbb{Z}_{n}
\end{array}\right\} 
\]
By the fundamental theorem of algebra there exist $\left\{ r_{i}:0\le i<2^{n}\right\} \subset\mathbb{C}$
such that 
\[
r_{0}\prod_{0<i<2^{n}}\left(r_{i}+\sum_{j\in\mathbb{Z}_{n}}2^{j}\,x_{j}\right)=\sum_{\begin{array}{c}
\mathbf{b}\in\left\{ 0,1\right\} ^{n\times1}\\
\text{s.t. }F\left(\mathbf{b}\right)=1
\end{array}}\prod_{\mathbf{d}\in\left\{ 0,1\right\} ^{n\times1}\backslash\left\{ \mathbf{b}\right\} }\left(\frac{\underset{j\in\mathbb{Z}_{n}}{\sum}2^{j}\,x_{j}-\underset{k\in\mathbb{Z}_{n}}{\sum}2^{k}\,d_{k}}{\underset{j\in\mathbb{Z}_{n}}{\sum}2^{j}\,b_{j}-\underset{k\in\mathbb{Z}_{n}}{\sum}2^{k}\,d_{k}}\right).
\]
Consequently 
\begin{equation}
F\left(\mathbf{x}\right)\equiv r_{0}\prod_{0<i<2^{n}}\left(r_{i}+\sum_{j\in\mathbb{Z}_{n}}2^{j}\,x_{j}\right)\mod\left\{ \begin{array}{c}
\left(x_{i}\right)^{2}-x_{i}\\
i\in\mathbb{Z}_{n}
\end{array}\right\} .\label{first_PDP}
\end{equation}
Thereby expressing the desired depth--2 $\prod\sum$ formula expressing
the Boolean function $F$ modulo relations described in Eq. (\ref{Binary Algebraic Relations}).
\end{proof}
Arithmetic formulas derived in the proof of Prop. (\ref{Canonical_Boolean_Interpolant})
feature expressions of the form 
\begin{equation}
\sum_{0\le u<\rho}\,\prod_{0\le v<d}\left(\mathbf{B}\left[u,v,0\right]+\underset{w\in\mathbb{Z}_{n}}{\sum}\mathbf{B}\left[u,v,w+1\right]x_{w}\right).\label{sum_of_product of linear form}
\end{equation}
We say that the hypermatrice $\mathbf{B}\in\mathbb{C}^{\rho\times d\times\left(n+1\right)}$
underlies the corresponding depth--3 $\sum\prod\sum$ arithmetic
formulas.

\subsection{Partial Differential Encoding of Boolean functions and their relaxations.}

For notational convenience, let $\mathbb{Z}_{n}$ denote the set formed
by the first $n$ consecutive non-negative integers i.e.
\[
\mathbb{Z}_{n}:=\left[0,n\right)\cap\mathbb{Z}.
\]
For simplicity take $n$ to be a perfect square. Edges of the complete
graph on $\sqrt{n}$ vertices allowing for loop edges are associated
with members of $\mathbb{Z}_{n}$ as prescribed by the following identification
:
\[
\text{the edge }\left(i,j\right)\in\mathbb{Z}_{\sqrt{n}}\times\mathbb{Z}_{\sqrt{n}}\text{ is associated with the integer }i\sqrt{n}+j\in\mathbb{Z}_{n}.
\]
Depth--3 arithmetic formulas expressing Boolean functions suggest
alternative partial differential encoding of Boolean functions.
\begin{defn}
A Partial Differential Encoding (or PDE for short) of a Boolean function
\[
F:\left\{ 0,1\right\} ^{n\times1}\rightarrow\left\{ 0,1\right\} 
\]
is one of two encodings of the Boolean function $F$. The first is
of the form :
\[
F\left(\mathbf{1}_{T}\right)=\left(\left.\left(\prod_{i\sqrt{n}+j\in T}\frac{\partial}{\partial a_{i,j}}\right)\sum_{0\le u<\rho}\,\prod_{0\le v<d}\left(\mathbf{B}\left[u,v,0\right]+\sum_{0\le i,j<\sqrt{n}}\mathbf{B}\left[u,v,1+i\sqrt{n}+j\right]a_{i,j}\right)\right\rfloor _{\mathbf{A}=\mathbf{0}_{\sqrt{n}\times\sqrt{n}}}\right)^{m},
\]
for all $T\subseteq\mathbb{Z}_{n}$ and where $\mathbf{1}_{T}$ denotes
the indicator vector of the edge subset $T$. Note that such a PDE
is specified via a mulitlinear polynomial in the $n$ variables $\left\{ a_{0,0},\cdots,a_{\sqrt{n}-1,\sqrt{n}-1}\right\} $.
In particular when $m=1$ the said multilinear polynomial is
\[
\sum_{0\le u<\rho}\,\prod_{0\le v<d}\left(\mathbf{B}\left[u,v,0\right]+\sum_{0\le i,j<\sqrt{n}}\mathbf{B}\left[u,v,1+i\sqrt{n}+j\right]a_{i,j}\right)=\sum_{\begin{array}{c}
\mathbf{b}\in\left\{ 0,1\right\} ^{n\times1}\\
\text{s.t. }F\left(\mathbf{b}\right)=1
\end{array}}\prod_{0\le i,j<\sqrt{n}}\left(a_{i,j}\right)^{b_{i\sqrt{n}+j}}.
\]
In its second form, a PDE of $F$ is specified by in polynomial in
the $\sqrt{n}$ variables $\left\{ x_{0},\cdots,x_{\sqrt{n}-1}\right\} $
not necessarily multilinear as follows
\[
F\left(\mathbf{1}_{T}\right)=\left(\left.\prod_{0\le i,j<\sqrt{n}}\left(\frac{\partial}{\sqrt[j]{j!}\,\partial x_{i}}\right)^{j\,\mathbf{1}_{T}\left[i\sqrt{n}+j\right]}\sum_{0\le u<\rho}\prod_{0\le v<d}\left(\mathbf{H}\left[u,v,0\right]+\sum_{0\le w<\sqrt{n}}\mathbf{H}\left[u,v,1+w\right]x_{w}\right)\right\rfloor _{\mathbf{x}=\mathbf{0}_{\sqrt{n}\times1}}\right)^{m}.
\]
In particular when $m=1$ the polynomial used to specify the PDE is
\[
\sum_{0\le u<\rho}\prod_{0\le v<d}\left(\mathbf{H}\left[u,v,0\right]+\sum_{0\le w<\sqrt{n}}\mathbf{H}\left[u,v,1+w\right]x_{w}\right)=\sum_{\begin{array}{c}
\mathbf{b}\in\left\{ 0,1\right\} ^{n\times1}\\
\text{s.t. }F\left(\mathbf{b}\right)=1
\end{array}}\prod_{0\le i,j<\sqrt{n}}\left(x_{i}\right)^{j\,b_{i\sqrt{n}+j}}.
\]
In both forms, the positive integer $m$ is called the exponent parameter
of the PDE. We see that hypermatrices $\mathbf{B}\in\mathbb{C}^{\rho\times d\times\left(1+n\right)}$
and $\mathbf{H}\in\mathbb{C}^{\rho\times d\times\left(1+\sqrt{n}\right)}$
completely specify the PDE. Similarly, a PDE relaxation is encodings
of one of the form 
\[
\left.\left(\prod_{i\sqrt{n}+j\in T}\frac{\partial}{\partial a_{i,j}}\right)\sum_{0\le u<\rho}\,\prod_{0\le v<d}\left(\mathbf{B}\left[u,v,0\right]+\sum_{0\le i,j<\sqrt{n}}\mathbf{B}\left[u,v,1+i\sqrt{n}+j\right]a_{i,j}\right)\right\rfloor _{\mathbf{A}=\mathbf{0}}\text{is }\begin{cases}
\begin{array}{cc}
\ne0 & \text{if }F\left(\mathbf{1}_{T}\right)=1\\
\\
0 & \text{otherwise}
\end{array},\end{cases}
\]
or alternatively as
\[
\left.\prod_{i\sqrt{n}+j\in T}\left(\frac{\partial}{\sqrt[j]{j!}\,\partial x_{i}}\right)^{j}\sum_{0\le u<\rho}\,\prod_{0\le v<d}\left(\mathbf{H}\left[u,v,0\right]+\sum_{0\le w<\sqrt{n}}\mathbf{H}\left[u,v,1+w\right]x_{w}\right)\right\rfloor _{\mathbf{x}=\mathbf{0}}\text{is }\begin{cases}
\begin{array}{cc}
\ne0 & \text{if }F\left(\mathbf{1}_{T}\right)=1\\
\\
0 & \text{otherwise}
\end{array}.\end{cases}
\]
\end{defn}

More generally, PDEs and their relaxations can be defined for Boolean
function on $m$--uniform hypergraphs. In that setting a PDE relaxation
is expressed as
\[
\left.\left(\prod_{\text{lex}\left(i_{0},\cdots,i_{m-1}\right)\in T}\frac{\partial}{\partial a_{i_{0},\cdots,i_{m-1}}}\right)\sum_{0\le u<\rho}\,\prod_{0\le v<d}\left(\mathbf{B}\left[u,v,0\right]+\sum_{0\le i_{0},\cdots,i_{m-1}<n}\mathbf{B}\left[u,v,1+\text{lex}\left(i_{0},\cdots,i_{m-1}\right)\right]a_{i_{0},\cdots,i_{m-1}}\right)\right\rfloor _{\mathbf{A}=\mathbf{0}}
\]
\[
\text{is }\begin{cases}
\begin{array}{cc}
\ne0 & \text{ if }F\left(\mathbf{1}_{T}\right)=1\\
\\
0 & \text{otherwise}
\end{array},\end{cases}
\]
where 
\[
\text{lex}\left(i_{0},\cdots,i_{m-1}\right)=\sum_{0\le k<m}i_{k}\,n^{\frac{k}{m}},\ \forall\,\left(i_{0},\cdots,i_{m-1}\right)\in\left(\mathbb{Z}_{\sqrt[k]{n}}\right)^{m}.
\]
PDEs exemplify our sought variant of the ``needles in a haystack''
conundrum. Consider Boolean functions specified in terms of a given
arbitrary subset $S\subseteq\mathbb{Z}_{n}$ such that
\begin{equation}
F_{\subseteq S}\left(\mathbf{1}_{T}\right)=\begin{cases}
\begin{array}{cc}
1 & \text{ if }T\subseteq S\\
\\
0 & \text{otherwise}
\end{array}, & F_{\supseteq S}\left(\mathbf{1}_{T}\right)=\end{cases}\begin{cases}
\begin{array}{cc}
1 & \text{ if }T\supseteq S\\
\\
0 & \text{otherwise}
\end{array} & \text{ and }F_{=S}\left(\mathbf{1}_{T}\right)=\end{cases}\begin{cases}
\begin{array}{cc}
1 & \text{ if }T=S\\
\\
0 & \text{otherwise}
\end{array}.\end{cases}\label{Subset Superset Boolean function specification}
\end{equation}
In other words these Boolean functions test wether or not the input
graph whose edges make up the subset $T\subset\mathbb{Z}_{n}$ is
a subgraph respectively supergraph or equal to of some fixed given
graph whose edge make up the subset $S\subseteq\mathbb{Z}_{n}$. PDEs
of $F_{\subseteq S}$, $F_{\supseteq S}$ and $F_{=S}$ are given
by
\[
\forall\,\mathbf{1}_{T}\in\left\{ 0,1\right\} ^{n\times1},\ F_{\subseteq S}\left(\mathbf{1}_{T}\right)=\left(\left.\left(\prod_{i\sqrt{n}+j\in T}\frac{\partial}{\partial a_{i,j}}\right)P_{\subseteq S}\left(\mathbf{A}\right)\right\rfloor _{\mathbf{A}=\mathbf{0}_{\sqrt{n}\times\sqrt{n}}}\right)^{m},
\]
\[
\forall\,\mathbf{1}_{T}\in\left\{ 0,1\right\} ^{n\times1},\ F_{\supseteq S}\left(\mathbf{1}_{T}\right)=\left(\left.\left(\prod_{i\sqrt{n}+j\in T}\frac{\partial}{\partial a_{i,j}}\right)P_{\supseteq S}\left(\mathbf{A}\right)\right\rfloor _{\mathbf{A}=\mathbf{0}_{\sqrt{n}\times\sqrt{n}}}\right)^{m},
\]
and
\[
\forall\,\mathbf{1}_{T}\in\left\{ 0,1\right\} ^{n\times1},\ F_{=S}\left(\mathbf{1}_{T}\right)=\left(\left.\left(\prod_{i\sqrt{n}+j\in T}\frac{\partial}{\partial a_{i,j}}\right)P_{=S}\left(\mathbf{A}\right)\right\rfloor _{\mathbf{A}=\mathbf{0}_{\sqrt{n}\times\sqrt{n}}}\right)^{m},
\]
where 
\[
P_{\subseteq S}\left(\mathbf{A}\right)\in\left\{ \sum_{R\subseteq S}\omega_{R}\prod_{i\sqrt{n}+j\in R}a_{i,j}\,:\,\begin{array}{c}
\left(\omega_{R}\right)^{m}=1\\
R\subseteq S
\end{array}\right\} ,
\]
\[
P_{\supseteq S}\left(\mathbf{A}\right)\in\left\{ \sum_{R\supseteq S}\omega_{R}\prod_{i\sqrt{n}+j\in R}a_{i,j}\,:\,\begin{array}{c}
\left(\omega_{R}\right)^{m}=1\\
R\supseteq S
\end{array}\right\} ,
\]
and 
\[
P_{=S}\left(\mathbf{A}\right)\in\left\{ \omega_{S}\prod_{i\sqrt{n}+j\in S}a_{i,j}\,:\,\left(\omega_{S}\right)^{m}=1\right\} .
\]

\begin{example}
For instance take $n=4$, and for simplicity take the exponent parameter
to be $m=1$. In that setting the edges of the complete graph allowing
for loop edges on $2$ vertices are identified with members of $\mathbb{Z}_{4}=\left\{ \text{lex}\left(0,0\right)=0,\,\text{lex}\left(0,1\right)=1,\,\text{lex}\left(1,0\right)=2,\,\text{lex}\left(1,1\right)=3\right\} $.
Further let the chosen subset of edges which make our chosen graph
be given by $S=\left\{ 0,1,2\right\} $ then
\[
\begin{array}{ccc}
P_{\subseteq S}\left(\mathbf{A}\right) & = & a_{00}a_{01}a_{10}+a_{00}a_{01}+a_{00}a_{10}+a_{01}a_{10}+a_{00}+a_{01}+a_{10}+1,\\
\\
P_{\supseteq S}\left(\mathbf{A}\right) & = & a_{00}a_{01}a_{10}a_{11}+a_{00}a_{01}a_{10},\\
\\
P_{=S}\left(\mathbf{A}\right) & = & a_{00}a_{01}a_{10}.
\end{array}
\]
In particular given $T=\left\{ 1,2\right\} $, 
\end{example}

\begin{figure}
\begin{tikzpicture}
	\node (0) at (-2,0) {0};
	\node (1) at ( 2,0) {1};

	\path [->] (0) edge  [bend left] node[above,sloped] {$\text{lex}(0,1)=1$} (1)
					  (1) edge [bend left] node[below,sloped] {$\text{lex}(1,0)=2$} (0);

\end{tikzpicture} \centering \caption{Graph $T$ and its edges.}
\end{figure}

\begin{example}
the corresponding indicator vector is
\[
\mathbf{1}_{T}=\left(\begin{array}{rrrr}
0 & 1 & 1 & 0\end{array}\right)^{\top}.
\]
We see that 
\[
F_{\subseteq S}\left(\mathbf{1}_{T}\right)=\left(\left.\left(\prod_{2i+j\in T}\frac{\partial}{\partial a_{i,j}}\right)P_{\subseteq S}\left(\mathbf{A}\right)\right\rfloor _{\mathbf{A}=\mathbf{0}_{2\times2}}\right)=1,
\]
\[
F_{\supseteq S}\left(\mathbf{1}_{T}\right)=\left(\left.\left(\prod_{2i+j\in T}\frac{\partial}{\partial a_{i,j}}\right)P_{\supseteq S}\left(\mathbf{A}\right)\right\rfloor _{\mathbf{A}=\mathbf{0}_{2\times2}}\right)=0,
\]
\[
F_{=S}\left(\mathbf{1}_{T}\right)=\left(\left.\left(\prod_{2i+j\in T}\frac{\partial}{\partial a_{i,j}}\right)P_{=S}\left(\mathbf{A}\right)\right\rfloor _{\mathbf{A}=\mathbf{0}_{2\times2}}\right)=0.
\]
\end{example}

For a fixed exponent parameter $m$, there are exactly $m^{\left(2^{\left|S\right|}\right)}$
distinct choices for $m$-th roots of unity which make up non-vanishing
coefficients of $P_{\subseteq S}\left(\mathbf{A}\right)$ and $m^{\left(2^{n-\left|S\right|}\right)}$
distinct choices for $m$-th roots of unity which make up non-vanishing
coefficients of $P_{\supseteq S}\left(\mathbf{A}\right)$. On the
one hand, such PDEs of $F_{\subseteq S}$ make up the ``haystack''.
On the other hand, the ``needles'' embedded in this haystack are
\emph{optimal} PDEs. A PDE of $F_{\subseteq S}$ is optimal if the
hypermatrix which underlies depth--3 $\sum\prod\sum$ arithmetic
formula used to specify the PDE is such that product of dimensions
$\rho\cdot d$ is the minimum possible. Let $\mathbf{B}\in\mathbb{C}^{\rho\times d\times\left(1+n\right)}$
underly the depth--3 $\sum\prod\sum$ arithmetic formula expressing
$P_{\subseteq S}$, such that $\rho$ is the smallest possible integer,
then recall that $\rho$ is the Chow-rank (over $\mathbb{C}$) of
the polynomial $P_{\subseteq S}$. For instance, recall that for a
multilinear polynomial of total degree two in the $n$ variables $a_{0,0},\cdots,a_{i,j},\cdots,a_{\sqrt{n},\sqrt{n}}$
given 
\[
\left(\begin{array}{c}
a_{0,0}\\
\vdots\\
a_{i,j}\\
\vdots\\
a_{\sqrt{n},\sqrt{n}}
\end{array}\right)^{\top}\left(\mathbf{M}\circ\left(\mathbf{1}_{n\times n}-\mathbf{I}_{n}\right)\right)\left(\begin{array}{c}
a_{0,0}\\
\vdots\\
a_{i,j}\\
\vdots\\
a_{\sqrt{n},\sqrt{n}}
\end{array}\right)\text{ where }\mathbf{M}\in\mathbb{C}^{n\times n},
\]
where $\circ$ denotes the entry-wise product
\[
\text{Chow-rank}\left\{ \left(\begin{array}{c}
a_{0,0}\\
\vdots\\
a_{i,j}\\
\vdots\\
a_{\sqrt{n},\sqrt{n}}
\end{array}\right)^{\top}\left(\mathbf{M}\circ\left(\mathbf{1}_{n\times n}-\mathbf{I}_{n}\right)\right)\left(\begin{array}{c}
a_{0,0}\\
\vdots\\
a_{i,j}\\
\vdots\\
a_{\sqrt{n},\sqrt{n}}
\end{array}\right)\right\} =\underset{\begin{array}{c}
\mathbf{H}\in\mathbb{C}^{n\times n}\\
\mathbf{H}^{\top}=-\mathbf{H}
\end{array}}{\text{Inf}}\text{tensor-rank}\left\{ \mathbf{H}+\mathbf{M}\circ\left(\mathbf{1}_{n\times n}-\mathbf{I}_{n}\right)\right\} .
\]
A similar definition extends to higher total degree multilinear polynomials.
Reading directly from the expanded forms on the left hand side of
equalities
\[
\left(\sum_{R\subseteq S}\omega_{R}\prod_{i\sqrt{n}+j\in R}a_{i,j}\right)=\sum_{0\le u<\rho}\,\prod_{0\le v<d}\left(\mathbf{B}\left[u,v,0\right]+\sum_{0\le i,j<\sqrt{n}}\mathbf{B}\left[u,v,1+i\sqrt{n}+j\right]a_{i,j}\right),
\]
\[
\left(\sum_{R\supseteq S}\omega_{R}\prod_{i\sqrt{n}+j\in R}a_{i,j}\right)=\sum_{0\le u<\rho^{\prime}}\,\prod_{0\le v<d^{\prime}}\left(\mathbf{B}^{\prime}\left[u,v,0\right]+\sum_{0\le i,j<\sqrt{n}}\mathbf{B}^{\prime}\left[u,v,1+i\sqrt{n}+j\right]a_{i,j}\right),
\]
yields respective Chow-rank and degree bounds $\rho\le2^{\left|S\right|}$,
$d\ge\left|S\right|$ and $\rho^{\prime}\le2^{n-\left|S\right|}$,
$d^{\prime}\ge n$. Multilinear polynomials
\begin{equation}
P_{\subseteq S}\left(\mathbf{A}\right)=\prod_{i\sqrt{n}+j\in S}\left(1+a_{i,j}\right)\text{ and }P_{\supseteq S}\left(\mathbf{A}\right)=\left(\prod_{i\sqrt{n}+j\in S}a_{i,j}\right)\prod_{i\sqrt{n}+j\in\overline{S}}\left(1+a_{i,j}\right),\label{Optimal Product Expression}
\end{equation}
yield optimal PDEs 
\[
F_{\subseteq S}\left(\mathbf{1}_{T}\right)=\left(\left.\left(\prod_{i\sqrt{n}+j\in T}\frac{\partial}{\partial a_{i,j}}\right)P_{\subseteq S}\left(\mathbf{A}\right)\right\rfloor _{\mathbf{A}=\mathbf{0}_{\sqrt{n}\times\sqrt{n}}}\right)^{m},
\]
and
\[
F_{\supseteq S}\left(\mathbf{1}_{T}\right)=\left(\left.\left(\prod_{i\sqrt{n}+j\in T}\frac{\partial}{\partial a_{i,j}}\right)P_{\supseteq S}\left(\mathbf{A}\right)\right\rfloor _{\mathbf{A}=\mathbf{0}_{\sqrt{n}\times\sqrt{n}}}\right)^{m},
\]
These PDEs are optimal in the sense that the both the total degree
and the Chow-rank of polynomials $P_{\subseteq S}$ and $P_{\supseteq S}$
used to specify PDEs for $F_{\subseteq S}$ and $F_{\supseteq S}$
are as small possible.\\

\begin{prop}
Optimal choices for $P_{\subseteq S}$ and $P_{\supseteq S}$ are
\[
P_{\subseteq S}\left(\mathbf{A}\right)\in\left\{ \omega_{S}\prod_{i\sqrt{n}+j\in S}\left(1+\omega_{i,j}\,a_{i,j}\right):\begin{array}{c}
\left(\omega_{i,j}\right)^{m}=1\\
\forall\,i\sqrt{n}+j\in S
\end{array}\right\} \subset\left\{ \sum_{R\subseteq S}\omega_{R}\prod_{i\sqrt{n}+j\in R}a_{i,j}:\begin{array}{c}
\left(\omega_{R}\right)^{m}=1\\
R\subseteq S
\end{array}\right\} ,
\]
and
\[
P_{\supseteq S}\left(\mathbf{A}\right)\in\left\{ \left(\omega_{S}\prod_{i\sqrt{n}+j\in S}a_{i,j}\right)\prod_{i\sqrt{n}+j\in\overline{S}}\left(1+\omega_{i,j}\,a_{i,j}\right):\begin{array}{c}
\left(\omega_{i,j}\right)^{m}=1\\
\forall\,i\sqrt{n}+j\in\overline{S}\\
\left(\omega_{S}\right)^{m}=1
\end{array}\right\} \subset\left\{ \sum_{R\supseteq S}\omega_{R}\prod_{i\sqrt{n}+j\in R}a_{i,j}:\begin{array}{c}
\left(\omega_{R}\right)^{m}=1\\
R\supseteq S
\end{array}\right\} .
\]
In which case the sparse and thin hypermatrices which underly the
respective depth--3 $\sum\prod\sum$ arithmetic formulas are of size
$1\times\left|S\right|\times\left(1+n\right)$ and $1\times n\times\left(1+n\right)$.
\end{prop}

\begin{proof}
Prime factors in the factorization of the integer count for the number
of non-vanishing monomial terms in the expanded form of $P_{\subseteq S}\left(\mathbf{x}\right)$
and $P_{\supseteq S}\left(\mathbf{x}\right)$ yield lower bounds for
the number of non-vanishing terms which make up each linear form.
There are $2^{\left|S\right|}$ non-vanishing monomial terms in the
expanded form of $P_{\subseteq S}\left(\mathbf{x}\right)$ and $2^{\left(n-\left|S\right|\right)}$
non-vanishing terms in the expanded form of $P_{\supseteq S}\left(\mathbf{x}\right)$.
The chosen expression for $P_{\subseteq S}\left(\mathbf{x}\right)$
and $P_{\supseteq S}\left(\mathbf{x}\right)$ have Chow rank one.
Consequently the Chow decomposition upper-bound matches the factorization
lower bound. Thus completing the proof.
\end{proof}
The fractions of optimal PDEs for $F_{\subseteq S}$ and $F_{\supseteq S}$
are respectively $m^{\left(\left|S\right|-2^{\left|S\right|}\right)}$
and $m^{\left(n-\left|S\right|\right)-2^{\left(n-\left|S\right|\right)}}$.
Optimal PDEs devised for Boolean functions $F_{\subseteq S}$ and
$F_{\supseteq S}$, epitomize their membership into the complexity
class \textbf{P/Poly.} Namely the class of Boolean functions which
admit efficient PDEs (i.e. PDEs whose underlying hypermatrices are
upper bounded in size by some polynomial in $n$). We conclude this
section by describing some PDEs as well as some PDE relaxations realizing
some important families of Boolean functions.
\begin{example}
Let
\[
F_{\text{func}}:\left\{ 0,1\right\} ^{\sqrt{n}\times\sqrt{n}}\rightarrow\left\{ 0,1\right\} ,
\]
The Boolean function $F_{\text{func}}$ takes as input the adjacency
matrix $\mathbf{M}\in\left\{ 0,1\right\} ^{\sqrt{n}\times\sqrt{n}}$
of a directed graph $G$ (allowing for loop edges) and tests whether
or not every vertex in the input graph has out-degree equal to one.
\[
F_{\text{funct}}\left(\mathbf{M}\right)=\begin{cases}
\begin{array}{cc}
1 & \text{ if }G\text{ is a functional directed graph}\\
\\
0 & \text{otherwise}
\end{array},\end{cases}
\]
where $\mathbf{M}\in\left\{ 0,1\right\} ^{\sqrt{n}\times\sqrt{n}}$
denotes the adjacency matrix of the input $\sqrt{n}$-vertex graph
$G$. The entries of $\mathbf{M}$ are such that 
\[
\mathbf{M}\left[u,v\right]=\begin{cases}
\begin{array}{cc}
1 & \text{ if }\left(u,v\right)\in E\left(G\right)\\
\\
0 & \text{otherwise}
\end{array}, & 0\le u,v<\sqrt{n}\end{cases}.
\]
PDEs of $F_{\text{Func}}$ with exponent parameter $m$ are of the
form
\[
F_{\text{funct}}\left(\mathbf{M}\right)=\left(\left.\left(\prod_{0\le i,j<\sqrt{n}}\left(\frac{\partial}{\partial a_{i,j}}\right)^{\mathbf{M}\left[i,j\right]}\right)P_{\text{funct}}\left(\mathbf{A}\right)\right\rfloor _{\mathbf{A}=\mathbf{0}_{\sqrt{n}\times\sqrt{n}}}\right)^{m},\text{ for all }\mathbf{A}_{G}\in\left\{ 0,1\right\} ^{\sqrt{n}\times\sqrt{n}},
\]
where 
\[
P_{\text{funct}}\left(\mathbf{A}\right)\in\left\{ \sum_{f\in\left(\mathbb{Z}_{\sqrt{n}}\right)^{\mathbb{Z}_{\sqrt{n}}}}\omega_{f}\,\prod_{i\in\mathbb{Z}_{n}}a_{i,f\left(i\right)}:\begin{array}{c}
\left(\omega_{f}\right)^{m}=1\\
f\in\left(\mathbb{Z}_{\sqrt{n}}\right)^{\mathbb{Z}_{\sqrt{n}}}
\end{array}\right\} .
\]
The integer factorization lower-bound argument used to prove Prop.
(3) can be applied to $P_{\text{funct}}$. The Boolean function $F_{\text{func}}$
also lies in the class \textbf{P/Poly} since $P_{\text{func}}$ can
be taken such that
\[
P_{\text{funct}}\left(\mathbf{A}\right)\in\left\{ \prod_{i\in\mathbb{Z}_{\sqrt{n}}}\sum_{j\in\mathbb{Z}_{\sqrt{n}}}\omega_{i,j}\,a_{i,j}:\begin{array}{c}
\left(\omega_{ij}\right)^{m}=1\\
0\le i,j<\sqrt{n}
\end{array}\right\} \subset\left\{ \sum_{f\in\mathbb{Z}_{n}^{\mathbb{Z}_{n}}}\omega_{f}\prod_{i\in\mathbb{Z}_{n}}\mathbf{A}\left[i,f\left(i\right)\right]:\begin{array}{c}
\left(\omega_{f}\right)^{m}=1\\
f\in\left(\mathbb{Z}_{\sqrt{n}}\right)^{\mathbb{Z}_{\sqrt{n}}}
\end{array}\right\} .
\]
When $\sqrt{n}$ is prime the count factorization lower-bound of $\sqrt{n}$
non-vanishing terms per linear functional matches the number of terms
in the irreducible factors of the Chow-rank one decomposition. Finally,
for a fixed exponent parameter $m$, the fraction of optimal PDEs
is $m^{\left(n-n^{\frac{\sqrt{n}}{2}}\right)}$. Over the transformation
monoid of functions whose domain and codomain are both $\mathbb{Z}_{\sqrt{n}}$
i.e. functions in $\left(\mathbb{Z}_{\sqrt{n}}\right)^{\mathbb{Z}_{\sqrt{n}}}$
we describe an additional families of Boolean functions 
\[
F_{\mathbb{Z}_{\sqrt{n}}^{\mathbb{Z}_{\sqrt{n}}}\circ h}:\left(\mathbb{Z}_{\sqrt{n}}\right)^{\mathbb{Z}_{\sqrt{n}}}\rightarrow\left\{ 0,1\right\} 
\]
defined such that 
\[
F_{\left(\mathbb{Z}_{\sqrt{n}}\right)^{\mathbb{Z}_{\sqrt{n}}}\circ h}\left(f\right)=\begin{cases}
\begin{array}{cc}
1 & \text{ if }\exists\,g\in\left(\mathbb{Z}_{\sqrt{n}}\right)^{\mathbb{Z}_{\sqrt{n}}}\text{s.t. }f=g\circ h\\
\\
0 & \text{otherwise}
\end{array}, & \forall\,f\in\left(\mathbb{Z}_{\sqrt{n}}\right)^{\mathbb{Z}_{\sqrt{n}}}.\end{cases}
\]
In other words, the Boolean function tests whether or not the input
function $f\in\left(\mathbb{Z}_{\sqrt{n}}\right)^{\mathbb{Z}_{\sqrt{n}}}$
lies in the $h$--right coset of the transformation monoid $\left(\mathbb{Z}_{\sqrt{n}}\right)^{\mathbb{Z}_{\sqrt{n}}}$.
The Boolean function $F_{\left(\mathbb{Z}_{\sqrt{n}}\right)^{\mathbb{Z}_{\sqrt{n}}}\circ h}$
admits a PDE relaxation with exponent parameter $m=1$ given by
\[
F_{\left(\mathbb{Z}_{\sqrt{n}}\right)^{\mathbb{Z}_{\sqrt{n}}}\circ h}\left(f\right)=\left(\left.\left(\prod_{i\in\mathbb{Z}_{\sqrt{n}}}\frac{\partial}{\partial a_{if\left(i\right)}}\right)P_{\left(\mathbb{Z}_{\sqrt{n}}\right)^{\mathbb{Z}_{\sqrt{n}}}\circ h}\left(\mathbf{A}\right)\right\rfloor _{\mathbf{A}=\mathbf{0}_{\sqrt{n}\times\sqrt{n}}}\right),
\]
expressed in terms of the multivariate polynomial
\[
P_{\left(\mathbb{Z}_{\sqrt{n}}\right)^{\mathbb{Z}_{\sqrt{n}}}\circ h}\left(\mathbf{A}\right)=\prod_{i\in\mathbb{Z}_{n}}\left(\sum_{j\in\mathbb{Z}_{n}}\prod_{k\in h^{\left(-1\right)}\left(\left\{ i\right\} \right)}a_{k,j}\right).
\]
\end{example}

\section{PDEs of cardinality variants of $F_{\subseteq S}$ and $F_{\supseteq S}$
and Group Orbitals.}

We discuss symmetric variants of Boolean functions $F_{\subseteq S}$,
$F_{\supseteq S}$ and $F_{=S}\in\left\{ 0,1\right\} ^{\left\{ 0,1\right\} ^{n}}$
defined such that 
\begin{equation}
F_{\le\left|S\right|}\left(\mathbf{1}_{T}\right)=\begin{cases}
\begin{array}{cc}
1 & \text{ if }\left|T\right|\le\left|S\right|\\
\\
0 & \text{otherwise}
\end{array}, & F_{\ge\left|S\right|}\left(\mathbf{1}_{T}\right)=\end{cases}\begin{cases}
\begin{array}{cc}
1 & \text{ if }\left|T\right|\ge\left|S\right|\\
\\
0 & \text{otherwise}
\end{array} & \text{ and }F_{=\left|S\right|}\left(\mathbf{1}_{T}\right)=\end{cases}\begin{cases}
\begin{array}{cc}
1 & \text{ if }\left|T\right|=\left|S\right|\\
\\
0 & \text{otherwise}
\end{array}.\end{cases}\label{Cardinality Variant Boolean function Specification}
\end{equation}
Prior to describing PDE constructions for the Boolean functions $F_{\le\left|S\right|}$,
$F_{\ge\left|S\right|}$ and $F_{=\left|S\right|}$ we start by defining
important notions used to devise various polynomial constructions.
\begin{defn}
Given a multivariate polynomial $P\in\mathbb{C}\left[y_{0},\cdots,y_{m-1}\right]$
the canonical representative of the congruence class 
\[
P\left(y_{0},\cdots,y_{m-1}\right)\mod\left(\prod_{i\in\mathbb{Z}_{n}}\left(y_{i}\right)^{f\left(i\right)}-\prod_{j\in\mathbb{Z}_{m}}\left(x_{j}\right)^{g\left(j\right)}\right)
\]
where $f\in\mathbb{Z}_{d}^{\mathbb{Z}_{n}}$ and $g\in\mathbb{Z}_{d^{\prime}}^{\mathbb{Z}_{m}}$
is the multivariate polynomial in $\mathbb{C}\left[x_{0},\cdots,x_{m-1},y_{0},\cdots,y_{n-1}\right]$
obtained by replacing in the expanded form of $P$ every occurrence
of the monomial $\underset{i\in\mathbb{Z}_{n}}{\prod}\left(y_{i}\right)^{f\left(i\right)}$
by the monomial $\underset{j\in\mathbb{Z}_{m}}{\prod}\left(x_{j}\right)^{g\left(j\right)}$.
\end{defn}

Note that the congruence relation relating $P$ to its canonical representative
is an immediate consequence of the binomial theorem. For instance,
given a polynomial interpolation of a Boolean function 
\[
F:\left\{ 0,1\right\} ^{n\times1}\rightarrow\left\{ 0,1\right\} 
\]
given by 
\[
F\left(y_{0},\cdots,y_{n-1}\right)=\sum_{\begin{array}{c}
\mathbf{b}\in\left\{ 0,1\right\} ^{n\times1}\\
\text{s.t. }F\left(\mathbf{b}\right)=1
\end{array}}\prod_{i\in\mathbb{Z}_{n}}\left(\frac{y_{i}-\left(1-b_{i}\right)}{2b_{i}-1}\right)
\]
The multilinear polynomial used to specify the corresponding PDE with
exponent parameter $m=1$ is the canonical representative of the congruence
class 
\[
F\left(y_{0},\cdots,y_{n-1}\right)\mod\left\{ \prod_{i\in\mathbb{Z}_{n}}\left(\frac{y_{i}-\left(1-b_{i}\right)}{2b_{i}-1}\right)-\prod_{i\in\mathbb{Z}_{n}}\left(x_{i}\right)^{b_{i}}:\begin{array}{c}
\mathbf{b}\in\left\{ 0,1\right\} ^{n\times1}\\
\text{s.t. }F\left(\mathbf{b}\right)=1
\end{array}\right\} .
\]

\begin{defn}
Let $\mathcal{G}$ denote an arbitrary subgroup of the symmetric group
S$_{n}$, let
\[
\text{lex}_{\mathcal{G}}:\mathcal{G}\to\mathbb{Z}_{\left|\mathcal{G}\right|}
\]
denote an arbitrary bijective/lexicographic map. Let $\mathbf{Z}$
denote a symbolic $n\times\left|\mathcal{G}\right|$ matrix. The $n\times1$
$\mathcal{G}$--orbital vector depicts orbits of the action of $\mathcal{G}$
on the vertex set of a complete directed graph allowing for loop edges
\[
\mathcal{O}_{\mathbf{Z},\mathcal{G}}\left[i\sqrt{n}+j\right]=\prod_{\sigma\in\mathcal{G}}\mathbf{Z}\left[\sigma\left(i\sqrt{n}+j\right),\text{lex}_{\mathcal{G}}\left(\sigma\right)\right],\ \forall\,\left(i,j\right)\in\mathbb{Z}_{\sqrt{n}}\times\mathbb{Z}_{\sqrt{n}}.
\]
\end{defn}

We illustrate an orbital construction used to devise PDEs for Boolean
functions $F_{\le\left|S\right|}$, $F_{\ge\left|S\right|}$ and $F_{=\left|S\right|}$
from PDEs for Boolean functions $F_{\subseteq S}$, $F_{\supseteq S}$
and $F_{=S}$. We take the group to be the whole symmetric group.
As a result, each entry of the orbital vector $\mathcal{O}_{\mathbf{Z}}$
depicts the action of the symmetric group S$_{n}\subset\mathbb{Z}_{n}^{\mathbb{Z}_{n}}$
on the corresponding edge. In particular, each entry of $\mathcal{O}_{\mathbf{Z}}$
is a monomial in the entries of a symbolic $n\times\left(n!\right)$
matrix $\mathbf{Z}$ such that
\begin{equation}
\mathcal{O}_{\mathbf{Z}}\left[i\sqrt{n}+j\right]=\prod_{\sigma\in\text{S}_{n}}\mathbf{Z}\left[\sigma\left(i\sqrt{n}+j\right),\text{lex}_{\text{S}_{n}}\left(\sigma\right)\right],\ \forall\,\left(i,j\right)\in\mathbb{Z}_{\sqrt{n}}\times\mathbb{Z}_{\sqrt{n}},\label{Orbital vector definition}
\end{equation}
where 
\[
\text{lex}_{\text{S}_{n}}\left(\sigma\right)=\sum_{k\in\mathbb{Z}_{n}}\left(n-1-k\right)!\,\left|\left\{ \sigma\left(i\right)>\sigma\left(k\right)\,:\,0\le i<k<n\right\} \right|,\text{ for all }\sigma\in\text{S}_{n}.
\]
For this choice lex$_{\text{S}_{n}}\left(\text{id}\right)=0$ and
lex$_{\text{S}_{n}}\left(n-1-\text{id}\right)=n!-1$. Similarly, let
\[
\text{lex}_{\powerset\left(\mathbb{Z}_{n}\right)}:\powerset\left(\mathbb{Z}_{n}\right)\to\left(\mathbb{Z}_{2^{n}}\right),
\]
map bijectively members of the power set $\powerset\left(\mathbb{Z}_{n}\right)$
to $\mathbb{Z}_{2^{n}}$ as follows 
\[
\text{lex}_{\powerset\left(\mathbb{Z}_{n}\right)}\left(R\right)=\sum_{j\in R}2^{j},\ \text{ for all }R\subseteq\mathbb{Z}_{n}.
\]
For simplicity we take the exponent parameter $m=1$. Let the canonical
representative of $P_{\subseteq S}\left(\mathcal{O}_{\mathbf{Z}}\right)$
modulo binomial relations 
\begin{equation}
P_{\subseteq S}\left(\mathcal{O}_{\mathbf{Z}}\right)\text{ mod }\left\{ \prod_{i\in R}\mathbf{Z}\left[i,\text{lex}_{\text{S}_{n}}\left(\sigma\right)\right]-\prod_{j\in R}\mathbf{Y}\left[j,\text{lex}_{\powerset\left(\mathbb{Z}_{n}\right)}\left(R\right)\right]:\begin{array}{c}
R\subseteq\mathbb{Z}_{n}\\
\left|R\right|\le\left|S\right|\\
\sigma\in\text{S}_{n}
\end{array}\right\} ,\label{Canonical}
\end{equation}
denote the unique polynomial depending upon entries of $\mathbf{Y}$
and crucially not depend upon any entry of $\mathbf{Z}$. Note that
the order with which we perform the reduction modulo prescribed binomial
relations matters. The canonical representative of the congruence
class in Eq. (\ref{Canonical}), is obtained by reducing $P_{\subseteq S}\left(\mathcal{O}_{\mathbf{Z}}\right)$
modulo relations taken in decreasing order of magnitude of the cardinality
of the set parameter $R$.
\begin{prop}
The canonical representative of the congruence class
\[
P_{\subseteq S}\left(\mathcal{O}_{\mathbf{Z}}\right)\text{ mod }\left\{ \prod_{i\in R}\mathbf{Z}\left[i,\text{lex}_{\text{S}_{n}}\left(\sigma\right)\right]-\prod_{j\in R}\mathbf{Y}\left[j,\text{lex}_{\powerset\left(\mathbb{Z}_{n}\right)}\left(R\right)\right]:\begin{array}{c}
R\subseteq\mathbb{Z}_{n}\\
\left|R\right|\le\left|S\right|\\
\sigma\in\text{S}_{n}
\end{array}\right\} ,
\]
is the orbit list generating polynomial
\[
\sum_{0\le t\le\left|S\right|}{\left|S\right| \choose t}\prod_{\begin{array}{c}
R\subseteq\mathbb{Z}_{n}\\
\left|R\right|=t
\end{array}}\left(\prod_{j\in R}\mathbf{Y}\left[j,\text{lex}_{\powerset\left(\mathbb{Z}_{n}\right)}\left(R\right)\right]\right)^{\left(n-\left|R\right|\right)!\,\left|R\right|!}.
\]
\end{prop}

\begin{proof}
The canonical representative is devised by substituting into the expanded
form of $P_{\subseteq S}\left(\mathcal{O}_{\mathbf{Z}}\right)$ each
monomial occurrence of the form
\[
\prod_{i\in R}\mathbf{Z}\left[i,\text{lex}_{\text{S}_{n}}\left(\sigma\right)\right],\quad\text{ for all }\:\begin{array}{c}
R\subseteq\mathbb{Z}_{n}\\
\sigma\in\text{S}_{n}
\end{array},
\]
with the corresponding monomial
\[
\prod_{j\in R}\mathbf{Y}\left[j,\text{lex}_{\powerset\left(\mathbb{Z}_{n}\right)}\left(R\right)\right].
\]
The canonical representative is thus given by 
\[
\begin{array}{cc}
 & \underset{\begin{array}{c}
R\subseteq\mathbb{Z}_{n}\\
\left|R\right|\le\left|S\right|
\end{array}}{\sum}\underset{\begin{array}{c}
T\subseteq\mathbb{Z}_{n}\\
\left|T\right|=\left|R\right|
\end{array}}{\prod}\left(\underset{j\in T}{\prod}\mathbf{Y}\left[j,\text{lex}_{\powerset\left(\mathbb{Z}_{n}\right)}\left(T\right)\right]\right)^{\left(n-\left|T\right|\right)!\,\left|T\right|!}\\
\\
= & \underset{0\le t\le\left|S\right|}{\sum}{\left|S\right| \choose t}\underset{\begin{array}{c}
R\subseteq\mathbb{Z}_{n}\\
\left|R\right|=t
\end{array}}{\prod}\left(\underset{j\in R}{\prod}\mathbf{Y}\left[j,\text{lex}_{\powerset\left(\mathbb{Z}_{n}\right)}\left(R\right)\right]\right)^{\left(n-\left|R\right|\right)!\,\left|R\right|!}.
\end{array}
\]
Similarly, the canonical representative of the congruence
\[
P_{\supseteq S}\left(\mathcal{O}_{\mathbf{Z}}\right)\text{ mod }\left\{ \prod_{i\in R}\mathbf{Z}\left[i,\text{lex}_{\text{S}_{n}}\left(\sigma\right)\right]-\prod_{j\in R}\mathbf{Y}\left[j,\text{lex}_{\powerset\left(\mathbb{Z}_{n}\right)}\left(R\right)\right]:\begin{array}{c}
R\subseteq\mathbb{Z}_{n},\\
\left|R\right|\ge\left|S\right|,\\
\sigma\in\text{S}_{n}
\end{array}\right\} ,
\]
is the polynomial
\[
\begin{array}{cc}
 & \underset{\begin{array}{c}
R\subseteq\mathbb{Z}_{n}\\
\left|R\right|\ge\left|S\right|
\end{array}}{\sum}\underset{\begin{array}{c}
T\subseteq\mathbb{Z}_{n}\\
\left|T\right|=\left|R\right|
\end{array}}{\prod}\left(\underset{j\in R}{\prod}\mathbf{Y}\left[j,\text{lex}_{\powerset\left(\mathbb{Z}_{n}\right)}\left(T\right)\right]\right)^{\left(n-\left|T\right|\right)!\,\left|T\right|!},\\
\\
= & \underset{\left|S\right|\le t\le n}{\sum}{n-\left|S\right| \choose t-\left|S\right|}\underset{\begin{array}{c}
R\subseteq\mathbb{Z}_{n}\\
\left|R\right|=t
\end{array}}{\prod}\left(\underset{j\in R}{\prod}\mathbf{Y}\left[j,\text{lex}_{\powerset\left(\mathbb{Z}_{n}\right)}\left(R\right)\right]\right)^{\left(n-\left|R\right|\right)!\,\left|R\right|!}.
\end{array}
\]
and the canonical representative of the congruence class
\[
P_{=S}\left(\mathcal{O}_{\mathbf{Z}}\right)\text{ mod }\left\{ \prod_{i\in R}\mathbf{Z}\left[i,\text{lex}_{\text{S}_{n}}\left(\sigma\right)\right]-\prod_{j\in R}\mathbf{Y}\left[j,\text{lex}_{\powerset\left(\mathbb{Z}_{n}\right)}\left(R\right)\right]:\begin{array}{c}
R\subseteq\mathbb{Z}_{n},\\
\left|R\right|\ge\left|S\right|,\\
\sigma\in\text{S}_{n}
\end{array}\right\} ,
\]
is the polynomial
\[
\prod_{\begin{array}{c}
T\subseteq\mathbb{Z}_{n}\\
\left|T\right|=\left|S\right|
\end{array}}\left(\prod_{j\in R}\mathbf{Y}\left[j,\text{lex}_{\powerset\left(\mathbb{Z}_{n}\right)}\left(T\right)\right]\right)^{\left(n-\left|S\right|\right)!\,\left|S\right|!},
\]
Let $P_{\le\left|S\right|}$ and $P_{\ge\left|S\right|}$ 
\begin{equation}
P_{\le\left|S\right|}=\left(P_{\subseteq S}\left(\mathcal{O}_{\mathbf{Z}}\right)\text{mod}\left\{ \prod_{i\in R}\mathbf{Z}\left[i,\text{lex}_{\text{S}_{n}}\left(\sigma\right)\right]-\left(\begin{array}{cc}
1 & \frac{\underset{i\sqrt{n}+j\in R}{\prod}a_{i,j}}{\left(n-\left|R\right|\right)!\,\left|R\right|!\,{\left|S\right| \choose \left|R\right|}}\\
0 & 1
\end{array}\right):\begin{array}{c}
R\subseteq\mathbb{Z}_{n},\\
\left|R\right|\le\left|S\right|,\\
\sigma\in\text{S}_{n}
\end{array}\right\} \right)\left[0,1\right],\label{ImplicitA}
\end{equation}
and 
\begin{equation}
P_{\ge\left|S\right|}=\left(P_{\supseteq S}\left(\mathcal{O}_{\mathbf{Z}}\right)\text{mod}\left\{ \prod_{i\in R}\mathbf{Z}\left[i,\text{lex}_{\text{S}_{n}}\left(\sigma\right)\right]-\left(\begin{array}{cc}
1 & \frac{\underset{i\sqrt{n}+j\in R}{\prod}a_{i,j}}{\left(n-\left|R\right|\right)!\,\left|R\right|!\,{n-\left|S\right| \choose \left|R\right|-\left|S\right|}}\\
0 & 1
\end{array}\right):\begin{array}{c}
R\subseteq\mathbb{Z}_{n},\\
\left|R\right|\ge\left|S\right|,\\
\sigma\in\text{S}_{n}
\end{array}\right\} \right)\left[0,1\right].\label{ImplicitB}
\end{equation}
The respective representative of the congruence classes are polynomials
in the class which depend only upon entries of $\mathbf{A}$ and do
not depend upon entries of $\mathbf{Z}$.
\end{proof}
\begin{prop}
Polynomials $P_{\le\left|S\right|}$ and $P_{\ge\left|S\right|}$
are used to specify PDEs
\begin{equation}
F_{\le\left|S\right|}\left(\mathbf{1}_{T}\right)=\left(\left.\left(\underset{i\sqrt{n}+j\in T}{\prod}\frac{\partial}{\partial a_{ij}}\right)P_{\le\left|S\right|}\left(\mathbf{A}\right)\right\rfloor _{\mathbf{A}=\mathbf{0}_{\sqrt{n}\times\sqrt{n}}}\right)^{m}\text{and }F_{\ge\left|S\right|}\left(\mathbf{1}_{T}\right)=\left(\left.\left(\underset{i\sqrt{n}+j\in T}{\prod}\frac{\partial}{\partial a_{ij}}\right)P_{\ge\left|S\right|}\left(\mathbf{A}\right)\right\rfloor _{\mathbf{A}=\mathbf{0}_{\sqrt{n}\times\sqrt{n}}}\right)^{m}.\label{Implicit PDE}
\end{equation}
\end{prop}

\begin{proof}
Similarly to the argument used to prove Prop. (3), the canonical representative
for the first of these congruence classes is obtained by successively
replacing into the expanded form of $P_{\subseteq S}\left(\mathcal{O}_{\mathbf{Z}}\right)$
every occurrence of monomials of the form
\[
\prod_{i\in R}\mathbf{Z}\left[i,\text{lex}_{\text{S}_{n}}\left(\sigma\right)\right],\;\forall\,\begin{array}{c}
R\subseteq\mathbb{Z}_{n}\\
\sigma\in\text{S}_{n}
\end{array},
\]
with the corresponding upper triangular $2\times2$ matrix 
\[
\left(\begin{array}{cc}
1 & \frac{\underset{i\sqrt{n}+j\in R}{\prod}a_{i,j}}{\left(n-\left|R\right|\right)!\,\left|R\right|!\,{\left|S\right| \choose \left|R\right|}}\\
0 & 1
\end{array}\right)
\]
followed taking the $\left[0,1\right]$ entry of the $2\times2$ matrix
resulting from the said substitutions. Similarly, the canonical representatives
for the second of the two congruence classes is obtained by successively
replacing into the expanded form of $P_{\supseteq S}\left(\mathcal{O}_{\mathbf{Z}}\right)$
every occurrence of monomials in the entries of $\mathbf{Z}$ given
by
\[
\prod_{i\in R}\mathbf{Z}\left[i,\text{lex}_{\text{S}_{n}}\left(\sigma\right)\right],\;\forall\,\begin{array}{c}
R\subseteq\mathbb{Z}_{n^{k}}\\
\sigma\in\text{S}_{n^{k}}
\end{array},
\]
with the corresponding upper triangular $2\times2$ matrix 
\[
\left(\begin{array}{cc}
1 & \frac{\underset{i\sqrt{n}+j\in R}{\prod}a_{i,j}}{\left(n-\left|R\right|\right)!\,\left|R\right|!\,{\left|S\right| \choose \left|R\right|}}\\
0 & 1
\end{array}\right)
\]
followed by taking the $\left[0,1\right]$ entry of the $2\times2$
matrix resulting from the said substitutions. 
\end{proof}
When $\left|S\right|$ is not a fixed constant independent of $n$
say $\left|S\right|=O\left(\sqrt{n}\right)$, then the PDE construction
above do not certify membership of $F_{\le\left|S\right|}$ and $F_{\ge\left|S\right|}$
into the complexity class \textbf{P/Poly}. In the setting where $\left|S\right|$
depends on $n$, Ben Or \cite{WN96} devises an optimal Chow decompositions
for $P_{\le\left|S\right|}\left(\mathbf{x}\right)$ and $P_{\ge\left|S\right|}\left(\mathbf{x}\right)$
via Cramer's rule as follows
\[
P_{\le\left|S\right|}\left(\mathbf{A}\right)=\sum_{0\le i<\left|S\right|}\frac{\det\mathbf{V}_{i}}{\underset{0\le u<v<n}{\prod}\left(\exp\left\{ \frac{2\pi\,v\,\sqrt{-1}}{n}\right\} -\exp\left\{ \frac{2\pi\,u\,\sqrt{-1}}{n}\right\} \right)},
\]
and 
\[
P_{\ge\left|S\right|}\left(\mathbf{A}\right)=\sum_{\left|S\right|\le i\le n}\frac{\det\mathbf{V}_{i}}{\underset{0\le u<v<n}{\prod}\left(\exp\left\{ \frac{2\pi\,v\,\sqrt{-1}}{n}\right\} -\exp\left\{ \frac{2\pi\,u\,\sqrt{-1}}{n}\right\} \right)},
\]
The $n\times n$ matrix $\mathbf{V}_{i}$ has entries given by 
\[
\mathbf{V}_{k}\left[u,v\right]=\begin{cases}
\begin{array}{ccc}
\underset{i\sqrt{n}+j\in\mathbb{Z}_{n}}{\prod}\left(1+\exp\left\{ \frac{2\pi\,u\,v\,\sqrt{-1}}{n}\right\} \,a_{i,j}\right) &  & \text{if }k=v\\
\exp\left\{ \frac{2\pi\,u\,v\,\sqrt{-1}}{n}\right\}  &  & \text{otherwise}
\end{array}.\end{cases}
\]
Such expansions describe depth--3 $\sum\prod\sum$ arithmetic formula
whose underlying hypermatrix is of size $n\times n\times\left(n+1\right)$.
\begin{example}
Let us illustrate the orbital construction in the case $n=4$, and
$S=\left\{ 0,1,3\right\} $. It follows from the setup that
\[
P_{\subseteq S}\left(\mathbf{A}\right)=\left(1+a_{00}\right)\left(1+a_{01}\right)\left(1+a_{11}\right).
\]
\[
\Rightarrow P_{\subseteq S}\left(\mathcal{O}_{\mathbf{Z}}\right)=\left(1+\prod_{\sigma\in\text{S}_{4}}Z\left[\sigma\left(2\cdot0+0\right),\text{lex}\left(\sigma\right)\right]\right)\left(1+\prod_{\sigma\in\text{S}_{4}}Z\left[\sigma\left(2\cdot0+1\right),\text{lex}\left(\sigma\right)\right]\right)\left(1+\prod_{\sigma\in\text{S}_{4}}Z\left[\sigma\left(2\cdot1+1\right),\text{lex}\left(\sigma\right)\right]\right).
\]
Hence 
\[
\left(P_{\subseteq S}\left(\mathcal{O}_{\mathbf{Z}}\right)\text{ mod }\left\{ \prod_{i\in R}\mathbf{Z}\left[i,\text{lex}_{\text{S}_{n}}\left(\sigma\right)\right]-\prod_{j\in R}\mathbf{Y}\left[j,\text{lex}_{\powerset\left(\mathbb{Z}_{n}\right)}\left(R\right)\right]:\begin{array}{c}
R\subseteq\mathbb{Z}_{n}\\
\left|R\right|\le\left|S\right|\\
\sigma\in\text{S}_{n}
\end{array}\right\} \right)\equiv
\]
\[
{3 \choose 0}+{3 \choose 1}\underset{\begin{array}{c}
R\subseteq\mathbb{Z}_{4}\\
\left|R\right|=1
\end{array}}{\prod}\left(\underset{j\in R}{\prod}\mathbf{Y}\left[j,\text{lex}_{\powerset\left(\mathbb{Z}_{4}\right)}\left(R\right)\right]\right)^{\left(4-\left|R\right|\right)!\,\left|R\right|!}+
\]
\[
{3 \choose 2}\underset{\begin{array}{c}
R\subseteq\mathbb{Z}_{4}\\
\left|R\right|=1
\end{array}}{\prod}\left(\underset{j\in R}{\prod}\mathbf{Y}\left[j,\text{lex}_{\powerset\left(\mathbb{Z}_{4}\right)}\left(R\right)\right]\right)^{\left(4-\left|R\right|\right)!\,\left|R\right|!}+{3 \choose 3}\underset{\begin{array}{c}
R\subseteq\mathbb{Z}_{4}\\
\left|R\right|=3
\end{array}}{\prod}\left(\underset{j\in R}{\prod}\mathbf{Y}\left[j,\text{lex}_{\powerset\left(\mathbb{Z}_{4}\right)}\left(R\right)\right]\right)^{\left(4-\left|R\right|\right)!\,\left|R\right|!}.
\]
Finally 
\[
\left(P_{\subseteq S}\left(\mathcal{O}_{\mathbf{Z}}\right)\text{ mod }\left\{ \prod_{i\in R}\mathbf{Z}\left[i,\text{lex}_{\text{S}_{n}}\left(\sigma\right)\right]-\left(\begin{array}{cc}
1 & \frac{\underset{i\sqrt{n}+j\in R}{\prod}a_{i,j}}{\left(n-\left|R\right|\right)!\,\left|R\right|!\,{\left|S\right| \choose \left|R\right|}}\\
0 & 1
\end{array}\right):\begin{array}{c}
R\subseteq\mathbb{Z}_{n}\\
\left|R\right|\le\left|S\right|\\
\sigma\in\text{S}_{n}
\end{array}\right\} \right)\equiv
\]
\[
\left(\begin{array}{ccc}
1 &  & 1+\underset{0\le\text{lex}\left(i_{0},j_{0}\right)<4}{\sum}a_{i_{0}j_{0}}+\underset{0\le\text{lex}\left(i_{0},j_{0}\right)<\text{lex}\left(i_{1},j_{1}\right)<4}{\sum}a_{i_{0}j_{0}}a_{i_{1}j_{1}}+\underset{0\le\text{lex}\left(i_{0},j_{0}\right)<\text{lex}\left(i_{1},j_{1}\right)<\text{lex}\left(i_{2},j_{2}\right)<4}{\sum}a_{i_{0}j_{0}}a_{i_{1}j_{1}}a_{i_{2}j_{2}}\\
\\
0 &  & 1
\end{array}\right)
\]
\end{example}

\section{Partial Differential Programs.}

We introduce here a variant of PDEs called \emph{Partial Differential
Programs} ( or PDPs for short). A PDP differs from a PDE in the fact
that the multilinear polynomial used to specify a PDE is implicitly
specified up to a polynomial size set of algebraic relations presented
in their expanded form. In fact the interpolation construction described
in Eq. (\ref{first_PDP}), illustrates such an implicit description.
PDPs are specified via smaller $\sum\prod\sum$ arithmetic formulas
compared to their PDE counterparts. PDPs also broaden the scope of
our proposed model of computation. This broadening hinges upon the
fact that in PDPs, polynomials used to specify PDEs are implicitly
prescribed by supplying a member of their congruence class. We refer
to such implicit descriptions of polynomials as \emph{programs}. For
a concrete example, consider a PDE for the Boolean function specified
by the truth table :
\begin{center}
\begin{tabular}{ccc}
\toprule 
\multirow{1}{*}{$x_{0}$} & $x_{1}$ & $F\left(\mathbf{x}\right)$\tabularnewline
\midrule
$0$ & $0$ & $1$\tabularnewline
\midrule
$0$ & $1$ & \multicolumn{1}{c}{$0$}\tabularnewline
\midrule
$1$ & $0$ & $1$\tabularnewline
\midrule
$1$ & $1$ & $1$\tabularnewline
\bottomrule
\end{tabular}
\par\end{center}

\[
\implies F\left(\mathbf{1}_{T}\right)=\left(\left.\left(\frac{\partial}{\partial x_{0}}\right)^{\mathbf{1}_{T}\left[0\right]}\left(\frac{\partial}{\partial x_{1}}\right)^{\mathbf{1}_{T}\left[1\right]}P_{F}\left(\mathbf{x}\right)\right\rfloor _{\mathbf{x}=\mathbf{0}_{2\times1}}\right).
\]
The multilinear polynomial $P_{F}\left(\mathbf{x}\right)$ used to
specify a PDE for $F$ with exponent parameter equal to one is given
by
\[
P_{F}\left(\mathbf{x}\right)=1+x_{0}+x_{0}x_{1}.
\]
Trivially, a $3\times2\times3$ hypermatrix underlies the depth--3
$\sum\prod\sum$ arithmetic formula which expresses the expanded form
of $P_{F}\left(\mathbf{x}\right)$. However, hypermatrices which underly
optimal depth--3 $\sum\prod\sum$ arithmetic formula for $P_{F}\left(\mathbf{x}\right)$
are of size $2\times2\times3$ as seen from the equality
\[
\begin{array}{cccc}
P_{F}\left(\mathbf{x}\right) & = & \left(\mathbf{B}\left[0,0,0\right]+\mathbf{B}\left[0,0,1\right]x_{0}+\mathbf{B}\left[0,0,2\right]x_{1}\right)\times\\
 &  & \left(\mathbf{B}\left[0,1,0\right]+\mathbf{B}\left[0,1,1\right]x_{0}+\mathbf{B}\left[0,1,2\right]x_{1}\right)\\
 &  & +\\
 &  & \left(\mathbf{B}\left[1,0,0\right]+\mathbf{B}\left[1,0,1\right]x_{0}+\mathbf{B}\left[1,0,2\right]x_{1}\right)\times\\
 &  & \left(\mathbf{B}\left[1,1,0\right]+\mathbf{B}\left[1,1,1\right]x_{0}+\mathbf{B}\left[1,1,2\right]x_{1}\right)
\end{array}
\]
where for instance we take non zero entries of $\mathbf{B}\in\left\{ 0,1\right\} ^{2\times2\times3}$
to be
\[
\mathbf{B}\left[0,0,0\right]=\mathbf{B}\left[0,1,0\right]=\mathbf{B}\left[1,0,1\right]=\mathbf{B}\left[1,1,0\right]=\mathbf{B}\left[1,1,2\right]=1.
\]
Hence taking
\[
P_{F}\left(\mathbf{x}\right)\in\left\{ \mu+u_{0}\,x_{0}\left(1+u_{1}x_{1}\right)\,:\,\mu^{m}=\left(u_{0}\right)^{m}=\left(u_{1}\right)^{m}=1\right\} .
\]
yields optimal PDEs for $F$ with exponent parameter $m$ of the form
\[
F\left(\mathbf{1}_{T}\right)=\left(\left.\left(\frac{\partial}{\partial x_{0}}\right)^{\mathbf{1}_{T}\left[0\right]}\left(\frac{\partial}{\partial x_{1}}\right)^{\mathbf{1}_{T}\left[1\right]}\mu+u_{0}x_{0}\left(1+u_{1}x_{1}\right)\right\rfloor _{\mathbf{x}=\mathbf{0}_{2\times1}}\right)^{m},
\]
Alternatively, we prescribe $P_{F}\left(\mathbf{x}\right)$ up to
congruence modulo the Boolean relations
\[
\left(x_{0}\right)^{2}\equiv x_{0}\:\text{ and }\:\left(x_{1}\right)^{2}\equiv x_{1}.
\]
In which case a PDP for $F$ is specified by a polynomial $Q_{F}\left(\mathbf{x}\right)$
member of the congruence class of $P_{F}\left(\mathbf{x}\right)$.
We write 
\[
F\left(\mathbf{1}_{T}\right)=\left(\left.\left(\frac{\partial}{\partial x_{0}}\right)^{\mathbf{1}_{T}\left[0\right]}\left(\frac{\partial}{\partial x_{1}}\right)^{\mathbf{1}_{T}\left[1\right]}Q_{F}\left(\mathbf{x}\right)\text{ mod}\left\{ \begin{array}{c}
\left(x_{0}\right)^{2}-x_{0}\\
\left(x_{1}\right)^{2}-x_{1}
\end{array}\right\} \right\rfloor _{\mathbf{x}=\mathbf{0}_{2\times1}}\right).
\]
Choices for a hypermatrix $\mathbf{B}^{\prime}$ which underlies optimal
depth--3 $\sum\prod\sum$ arithmetic formula for $Q_{F}\left(\mathbf{x}\right)$
are of size $1\times2\times3$ as seen from the expression
\[
Q_{F}\left(\mathbf{x}\right)=\left(\mathbf{B}^{\prime}\left[0,0,0\right]+\mathbf{B}^{\prime}\left[0,0,1\right]x_{0}+\mathbf{B}^{\prime}\left[0,0,2\right]x_{1}\right)\left(\mathbf{B}^{\prime}\left[0,1,0\right]+\mathbf{B}^{\prime}\left[0,1,1\right]x_{0}+\mathbf{B}^{\prime}\left[0,1,2\right]x_{1}\right),
\]
where for instance an optimal PDP for $F$ is completely determined
by the taking $\mathbf{B}^{\prime}$ such that
\[
\mathbf{B}^{\prime}\left[0,0,0\right]=\mathbf{B}^{\prime}\left[0,0,1\right]=1,\mathbf{B}^{\prime}\left[0,1,0\right]=\mathbf{B}^{\prime}\left[0,1,2\right]=1,\mathbf{B}^{\prime}\left[0,0,2\right]=-\frac{1}{2}.
\]
Our example illustrates an instance in where PDPs are smaller than
optimal PDEs for the same Boolean function. Reduction in size is achieved
at the expense of introducing some non-trivial algebraic relations.
By definition, PDEs form a proper subset of PDPs. There are finitely
many PDEs for any given Boolean function (assuming a fixed exponent
parameter $m$). By contrast there are infinitely many PDPs for a
given Boolean functions (assuming a fixed exponent parameter $m$).

\subsection{Orbital Chow-rank bound.}

We describe here the simplest illustration of a general method for
devising Chow rank bounds from group actions. As a concrete illustration
for the orbital bound argument, we derive bounds on the size of a
hypermatrix $\mathbf{H}$ which underlies an optimal Chow-decompositions
over $\mathbb{C}$ for $Q_{\le\left|S\right|}$, $Q_{\ge\left|S\right|}$
and $Q_{=\left|S\right|}\left(\mathbf{x}\right)$ used to specify
PDPs for Boolean functions
\[
F_{\le\left|S\right|}\left(\mathbf{1}_{T}\right)=\begin{cases}
\begin{array}{cc}
1 & \text{ if }\left|T\right|\le\left|S\right|\\
\\
0 & \text{otherwise}
\end{array} & ,\quad\end{cases}F_{\ge\left|S\right|}\left(\mathbf{1}_{T}\right)=\begin{cases}
\begin{array}{cc}
1 & \text{ if }\left|T\right|\ge\left|S\right|\\
\\
0 & \text{otherwise}
\end{array},\end{cases}
\]
\[
F_{=\left|S\right|}\left(\mathbf{1}_{T}\right)=\begin{cases}
\begin{array}{cc}
1 & \text{ if }\left|T\right|=\left|S\right|\\
\\
0 & \text{otherwise}
\end{array}.\end{cases}
\]
PDPs for such Boolean $F_{\le\left|S\right|}$ and $F_{\ge\left|S\right|}$
are respectively of the form 
\[
F_{\le\left|S\right|}\left(\mathbf{1}_{T}\right)=\left(\left.\left(\underset{i\in T}{\prod}\frac{\partial}{\partial x_{i}}\right)Q_{\le\left|S\right|}\left(x_{0},\cdots,x_{n-1}\right)\text{mod}\left\{ \begin{array}{c}
\left(x_{i}\right)^{2}-x_{i}\\
i\in\mathbb{Z}_{n}
\end{array}\right\} \right\rfloor _{\mathbf{x}=\mathbf{0}_{n\times1}}\right)^{m},
\]
\[
F_{\ge\left|S\right|}\left(\mathbf{1}_{T}\right)=\left(\left.\left(\underset{i\in T}{\prod}\frac{\partial}{\partial x_{i}}\right)Q_{\ge\left|S\right|}\left(x_{0},\cdots,x_{n-1}\right)\text{mod}\left\{ \begin{array}{c}
\left(x_{i}\right)^{2}-x_{i}\\
i\in\mathbb{Z}_{n}
\end{array}\right\} \right\rfloor _{\mathbf{x}=\mathbf{0}_{n\times1}}\right)^{m}.
\]

\begin{thm}
\label{Orbital-bound}A hypermatrix $\mathbf{H}\in\mathbb{C}^{\rho\times d\times\left(1+n\right)}$
can be chosen to underly a depth--3 $\sum\prod\sum$ arithmetic formula
for $Q_{\le\left|S\right|}$, $Q_{\ge\left|S\right|}$ and $Q_{=\left|S\right|}$
respectively used to specify PDPs for Boolean functions $F_{\le\left|S\right|}$,
$F_{\ge\left|S\right|}$ and $F_{=\left|S\right|}$ such that $\rho=1$.
\end{thm}

\begin{proof}
For each Boolean function $F_{\le\left|S\right|}$, $F_{\ge\left|S\right|}$
and $F_{=\left|S\right|}$, $\mathbf{H}$ is determined by congruence
identities of the form 
\[
\sum_{\left|R\right|\le\left|S\right|}\prod_{i\in R}x_{i}\equiv\sum_{0\le u<\rho}\prod_{0\le v<d}\left(\mathbf{H}\left[u,v,0\right]+\sum_{w\in\mathbb{Z}_{n}}\mathbf{H}\left[u,v,1+w\right]x_{w}\right)\text{mod}\left\{ \begin{array}{c}
\left(x_{i}\right)^{2}-x_{i}\\
i\in\mathbb{Z}_{n}
\end{array}\right\} ,
\]
\[
\sum_{\left|R\right|\ge\left|S\right|}\prod_{i\in R}x_{i}\equiv\sum_{0\le u<\rho}\prod_{0\le v<d}\left(\mathbf{H}\left[u,v,0\right]+\sum_{w\in\mathbb{Z}_{n}}\mathbf{H}\left[u,v,1+w\right]x_{w}\right)\text{mod}\left\{ \begin{array}{c}
\left(x_{i}\right)^{2}-x_{i}\\
i\in\mathbb{Z}_{n}
\end{array}\right\} ,
\]
\[
\sum_{\left|R\right|=\left|S\right|}\prod_{i\in R}x_{i}\equiv\sum_{0\le u<\rho}\prod_{0\le v<d}\left(\mathbf{H}\left[u,v,0\right]+\sum_{w\in\mathbb{Z}_{n}}\mathbf{H}\left[u,v,1+w\right]x_{w}\right)\text{mod}\left\{ \begin{array}{c}
\left(x_{i}\right)^{2}-x_{i}\\
i\in\mathbb{Z}_{n}
\end{array}\right\} .
\]
Expanding the right hand side yields
\[
\sum_{\left|R\right|\le\left|S\right|}\prod_{i\in R}x_{i}=\sum_{R\subseteq\mathbb{Z}_{n}}K_{R}\left(\mathbf{H}\right)\,\prod_{i\in R}x_{i},\ \sum_{\left|R\right|\ge\left|S\right|}\prod_{i\in R}x_{i}=\sum_{R\subseteq\mathbb{Z}_{n}}K_{R}\left(\mathbf{H}\right)\,\prod_{i\in R}x_{i},
\]
\[
\sum_{\left|R\right|=\left|S\right|}\prod_{i\in R}x_{i}=\sum_{R\subseteq\mathbb{Z}_{n}}K_{R}\left(\mathbf{H}\right)\,\prod_{i\in R}x_{i}.
\]
For each one of the congruence identities we have that for all $\forall\,R\subseteq\mathbb{Z}_{n}$,
the polynomial in the entries of the unknown matrix $\mathbf{H}$
given by $K_{R}\left(\mathbf{H}\right)$ is given by
\[
K_{R}\left(\mathbf{H}\right)=\left(\left.\sum_{\left\{ 1\le d_{i}<\rho\,:i\in R\right\} }\prod_{i\in R}\left(\frac{\partial}{\sqrt[d_{i}]{d_{i}!}\,\partial x_{i}}\right)^{d_{i}}\sum_{0\le u<\rho}\prod_{0\le v<d}\left(\mathbf{H}\left[u,v,0\right]+\sum_{w\in\mathbb{Z}_{n}}\mathbf{H}\left[u,v,1+w\right]x_{w}\right)\right\rfloor _{\mathbf{x}=\mathbf{0}_{n\times1}}\right),
\]
substituting each entry of $\mathbf{x}$ for the corresponding entry
of the orbital vector $\mathcal{O}_{\mathbf{Z}}$ yields constraints
\[
\sum_{\left|R\right|\le\left|S\right|}\prod_{i\in R}\mathcal{O}_{\mathbf{Z}}\left[i\right]=\sum_{R\subseteq\mathbb{Z}_{n}}K_{R}\left(\mathbf{H}\right)\,\prod_{i\in R}\mathcal{O}_{\mathbf{Z}}\left[i\right],
\]
\[
\sum_{\left|R\right|\ge\left|S\right|}\prod_{i\in R}\mathcal{O}_{\mathbf{Z}}\left[i\right]=\sum_{R\subseteq\mathbb{Z}_{n}}K_{R}\left(\mathbf{H}\right)\,\prod_{i\in R}\mathcal{O}_{\mathbf{Z}}\left[i\right],
\]
\[
\sum_{\left|R\right|=\left|S\right|}\prod_{i\in R}\mathcal{O}_{\mathbf{Z}}\left[i\right]=\sum_{R\subseteq\mathbb{Z}_{n}}K_{R}\left(\mathbf{H}\right)\,\prod_{i\in R}\mathcal{O}_{\mathbf{Z}}\left[i\right].
\]
Equating corresponding coefficients on both sides of the equal sign
in the respective canonical representatives of
\[
\sum_{\left|R\right|\le\left|S\right|}\prod_{i\in R}\mathcal{O}_{\mathbf{Z}}\left[i\right]\text{ mod}\left\{ \prod_{i\in R}\mathbf{Z}\left[i,\text{lex}\left(\sigma\right)\right]-\sqrt[\left(n-\left|R\right|\right)!\cdot\left|R\right|!]{\frac{\underset{j\in R}{\prod}\mathbf{Y}\left[j,\text{lex}\left(R\right)\right]}{{n \choose \left|R\right|}}}:\begin{array}{c}
\left|R\right|\le\left|S\right|\\
\sigma\in\text{S}_{n}
\end{array}\right\} ,
\]
\[
\sum_{\left|R\right|\ge\left|S\right|}\prod_{i\in R}\mathcal{O}_{\mathbf{Z}}\left[i\right]\text{ mod}\left\{ \prod_{i\in R}\mathbf{Z}\left[i,\text{lex}\left(\sigma\right)\right]-\sqrt[\left(n-\left|R\right|\right)!\cdot\left|R\right|!]{\frac{\underset{j\in R}{\prod}\mathbf{Y}\left[j,\text{lex}\left(R\right)\right]}{{n \choose \left|R\right|}}}:\begin{array}{c}
\left|R\right|\ge\left|S\right|\\
\sigma\in\text{S}_{n}
\end{array}\right\} ,
\]
\[
\sum_{\left|R\right|=\left|S\right|}\prod_{i\in R}\mathcal{O}_{\mathbf{Z}}\left[i\right]\text{ mod}\left\{ \prod_{i\in R}\mathbf{Z}\left[i,\text{lex}\left(\sigma\right)\right]-\sqrt[\left(n-\left|R\right|\right)!\cdot\left|R\right|!]{\frac{\underset{j\in R}{\prod}\mathbf{Y}\left[j,\text{lex}\left(R\right)\right]}{{n \choose \left|R\right|}}}:\begin{array}{c}
\left|R\right|=\left|S\right|\\
\sigma\in\text{S}_{n}
\end{array}\right\} .
\]
yields respectively 
\[
\sum_{0\le t\le\left|S\right|}\prod_{\begin{array}{c}
R\subseteq\mathbb{Z}_{n}\\
\left|R\right|=t
\end{array}}\left(\prod_{i\in R}\mathbf{Y}\left[i,\text{lex}\left(R\right)\right]\right),\,\sum_{\left|S\right|\le t\le n}\prod_{\begin{array}{c}
R\subseteq\mathbb{Z}_{n}\\
\left|R\right|=t
\end{array}}\left(\prod_{i\in R}\mathbf{Y}\left[i,\text{lex}\left(R\right)\right]\right),
\]
\[
\prod_{\begin{array}{c}
R\subseteq\mathbb{Z}_{n}\\
\left|R\right|=\left|S\right|
\end{array}}\left(\prod_{i\in R}\mathbf{Y}\left[i,\text{lex}\left(R\right)\right]\right).
\]
with the corresponding coefficients in the canonical representative
of the congruence class 
\[
\sum_{R\subseteq\mathbb{Z}_{n}}K_{R}\left(\mathbf{H}\right)\,\prod_{i\in R}\mathcal{O}_{\mathbf{Z}}\left[i\right]\text{ mod}\left\{ \prod_{i\in R}\mathbf{Z}\left[i,\text{lex}\left(\sigma\right)\right]-\sqrt[\left(n-\left|R\right|\right)!\cdot\left|R\right|!]{\frac{\underset{j\in R}{\prod}\mathbf{Y}\left[j,\text{lex}\left(R\right)\right]}{{n \choose \left|R\right|}}}:\begin{array}{c}
\left|R\right|\le\left|S\right|\\
\sigma\in\text{S}_{n}
\end{array}\right\} ,
\]
yields for each of the the three constraints a different systems of
$\left(n+1\right)$ equations in the $\rho\cdot d\cdot\left(1+n\right)$
unknown entries for $\mathbf{H}\in\mathbb{C}^{\rho\times d\times\left(1+n\right)}$
respectively of the form
\[
\left\{ 1=\sum_{\begin{array}{c}
R\subseteq\mathbb{Z}_{n}\\
\left|R\right|=t
\end{array}}K_{R}\left(\mathbf{H}\right)\,:\,0\le t\le\left|S\right|\right\} \cup\left\{ 0=\sum_{\begin{array}{c}
R\subseteq\mathbb{Z}_{n}\\
\left|R\right|=t
\end{array}}K_{R}\left(\mathbf{H}\right)\,:\,\left|S\right|<t\le n\right\} .
\]
\[
\left\{ 1=\sum_{\begin{array}{c}
R\subseteq\mathbb{Z}_{n}\\
\left|R\right|=t
\end{array}}K_{R}\left(\mathbf{H}\right)\,:\,\left|S\right|\le t\le n\right\} \cup\left\{ 0=\sum_{\begin{array}{c}
R\subseteq\mathbb{Z}_{n}\\
\left|R\right|=t
\end{array}}K_{R}\left(\mathbf{H}\right)\,:\,0\le t<\left|S\right|\right\} .
\]
\[
\left\{ 1=\sum_{\begin{array}{c}
R\subseteq\mathbb{Z}_{n}\\
\left|R\right|=\left|S\right|
\end{array}}K_{R}\left(\mathbf{H}\right)\,:\,\left|R\right|=\left|S\right|\right\} \cup\left\{ 0=\sum_{\begin{array}{c}
R\subseteq\mathbb{Z}_{n}\\
\left|R\right|=t
\end{array}}K_{R}\left(\mathbf{H}\right)\,:\,0\le t\ne\left|S\right|\le n\right\} .
\]
We know that by eliminating variables via the method of resultants,
the latter system of equation necessarily admits a solution whenever
the number unknowns namely $\rho\,d\,\left(1+n\right)$ matches or
exceeds the number of algebraically independent constraints $\le\left(1+n\right)$.
Hence when $\rho=\left\lceil \frac{1+n}{\left(1+n\right)\,d}\right\rceil $,
the number of variables matches or exceeds the number of algebraically
independent constraints. It follows from the degree lower bounds $d\ge\left|S\right|$
that the desired claim holds.
\end{proof}
We now proceed to devise the PDPs for $F_{\le\left|S\right|}$, $F_{\ge\left|S\right|}$
specified in terms of Chow rank one polynomials $Q_{\le\left|S\right|}$
and $Q_{\ge\left|S\right|}$ as suggested by Thrm. (\ref{Orbital-bound}).
Note that Boolean functions $F_{\le\left|S\right|}$, $F_{\ge\left|S\right|}$
are both symmetric with respect to permutations of their input variables.
So too are polynomials $P_{\le\left|S\right|}$ and $P_{\ge\left|S\right|}$
used to specify their PDEs. So we can express them in a way that the
fundamental theorem of symmetric polynomials tells us. Recall the
well known Newton--Girard identities. These identities relate the
densest (in their monomial support) set of generators for the ring
of symmetric polynomials given by
\[
e_{t}\left(\mathbf{x}\right)=\sum_{\begin{array}{c}
R\subseteq\mathbb{Z}_{k}\\
\left|R\right|=t
\end{array}}\prod_{j\in R}x_{j},\quad\forall\:t\in\mathbb{Z}_{n+1}\backslash\left\{ 0\right\} ,
\]
to the sparsest set of generators ( for the same polynomial ring )
given by
\[
p_{t}\left(\mathbf{x}\right)=\sum_{0\le i<n}\left(x_{i}\right)^{t},\quad\forall\:t\in\mathbb{Z}_{n+1}\backslash\left\{ 0\right\} 
\]

\begin{prop}
\label{Newton-Girard}For all integer $0<t\le n$, we have
\[
e_{t}\left(\mathbf{x}\right)=\left(-1\right)^{t}\sum_{{m_{1}+2m_{2}+\cdots+tm_{t}=t\atop m_{1}\ge0,\ldots,m_{t}\ge0}}\prod_{1\le i\le t}\frac{\left(-p_{i}\left(\mathbf{x}\right)\right)^{m_{i}}}{m_{i}!\,i^{m_{i}}}.
\]
\end{prop}

\begin{proof}
Consider the polynomial
\[
\frac{1}{\left(n-1\right)!}\sum_{\sigma\in\text{S}_{n}}\prod_{\gamma\in\text{S}_{n}}\left(x_{\sigma\left(0\right)}-y_{\gamma\left(1\right)}\right)^{\frac{1}{\left(n-1\right)!}}=\sum_{i\in\mathbb{Z}_{n}}\prod_{j\in\mathbb{Z}_{n}}\left(x_{i}-y_{j}\right)
\]
 
\[
\implies0=\sum_{0\le t\le n}p_{t}\left(\mathbf{x}\right)e_{n-t}\left(\mathbf{x}\right).
\]
Solving via back-substitution the resulting triangular system of linear
equations in the unknowns $e_{t}\left(\mathbf{x}\right)$, for all
integer $0<t\le n$ yields the Newton-Girard identity
\[
e_{t}\left(\mathbf{x}\right)=\left(-1\right)^{t}\sum_{{m_{1}+2m_{2}+\cdots+tm_{t}=t\atop m_{1}\ge0,\ldots,m_{t}\ge0}}\prod_{1\le i\le t}\frac{\left(-p_{i}\left(\mathbf{x}\right)\right)^{m_{i}}}{m_{i}!\,i^{m_{i}}}.
\]
\end{proof}
Using the Newton--Girard identity and exploiting the binary algebraic
relations in Eq. (\ref{Binary Algebraic Relations}), we eliminate
cross terms from multilinear polynomials $P_{\le\left|S\right|}$
and $P_{\ge\left|S\right|}$.
\begin{thm}
\label{Optimal PDPs}Boolean functions $F_{\le\left|S\right|}$ and
$F_{\ge\left|S\right|}$ admit PDPs respectively expressed in terms
of polynomials 
\[
Q_{\le\left|S\right|}\left(\mathbf{x}\right)\in\left\{ \sum_{\begin{array}{c}
R\subseteq\mathbb{Z}_{n}\\
\left|R\right|\le\left|S\right|
\end{array}}\omega_{R}\prod_{j\in R}x_{j}\,:\,\left(\omega_{R}\right)^{m}=1\right\} 
\]
and 
\[
Q_{\ge\left|S\right|}\left(\mathbf{x}\right)\in\left\{ \sum_{\begin{array}{c}
R\subseteq\mathbb{Z}_{n}\\
\left|R\right|\ge\left|S\right|
\end{array}}\omega_{R}\prod_{j\in R}x_{j}\,:\,\left(\omega_{R}\right)^{m}=1\right\} ,
\]
which can both be chosen to have Chow-rank $1$.
\end{thm}

\begin{proof}
We describe an elimination procedure which exploits congruence identities
\[
p_{t}\left(\mathbf{x}\right)\equiv p_{1}\left(\mathbf{x}\right)\text{ mod}\left\{ \begin{array}{c}
\left(x_{i}\right)^{2}-x_{i}\\
i\in\mathbb{Z}_{n}
\end{array}\right\} ,\quad\forall\:t\in\mathbb{Z}_{n+1}\backslash\left\{ 1,0\right\} .
\]
to reduce the number of terms. By Prop. (\ref{Newton-Girard}) taken
modulo binary algebraic relation described in Eq. (\ref{Binary Algebraic Relations}),
we have
\[
e_{t}\left(\mathbf{x}\right)\equiv\left(-1\right)^{t}\sum_{{m_{1}+2m_{2}+\cdots+tm_{t}=t\atop m_{1}\ge0,\ldots,m_{t}\ge0}}\prod_{1\le i\le t}\frac{\left(-p_{1}\left(\mathbf{x}\right)\right)^{m_{i}}}{m_{i}!\,i^{m_{i}}}\text{ mod}\left\{ \begin{array}{c}
\left(x_{i}\right)^{2}-x_{i}\\
i\in\mathbb{Z}_{n}
\end{array}\right\} .
\]
It follows that within respective congruence classes
\[
\left(P_{\le\left|S\right|}\left(\mathbf{x}\right)\text{ mod}\left\{ \begin{array}{c}
\left(x_{i}\right)^{2}-x_{i}\\
i\in\mathbb{Z}_{n}
\end{array}\right\} \right)\;\text{ and }\;\left(P_{\ge\left|S\right|}\left(\mathbf{x}\right)\text{ mod}\left\{ \begin{array}{c}
\left(x_{i}\right)^{2}-x_{i}\\
i\in\mathbb{Z}_{n}
\end{array}\right\} \right)
\]
lies univariate polynomials in the linear functional $\underset{i\in\mathbb{Z}_{n}}{\sum}x_{i}$
of degree $\left|S\right|$ and $n$ respectively. By the Fundamental
Theorem of Algebra, there exists
\[
\left\{ \alpha,\beta\right\} \cup\left\{ \alpha_{i},\beta_{j}:\begin{array}{c}
0\le i<\left|S\right|\\
0\le j<n
\end{array}\right\} \subset\mathbb{C},
\]
with which we express optimal PDPs for $F_{\le\left|S\right|}$, $F_{\ge\left|S\right|}$
and $F_{=\left|S\right|}$ as follows
\begin{equation}
F_{\le\left|S\right|}\left(\mathbf{1}_{T}\right)=\left(\left.\alpha\prod_{i\in\mathbb{Z}_{n}}\left(\frac{\partial}{\partial x_{i}}\right)^{\mathbf{1}_{T}\left[i\right]}\prod_{0\le j<\left|S\right|}\left(\alpha_{j}+\sum_{i\in\mathbb{Z}_{n}}x_{i}\right)\text{mod}\left\{ \begin{array}{c}
\left(x_{i}\right)^{2}-x_{i}\\
i\in\mathbb{Z}_{n}
\end{array}\right\} \right\rfloor _{\mathbf{x}=\mathbf{0}_{n\times1}}\right).\label{Opt_PDP_A}
\end{equation}
The hypermatrix underlying the polynomial used to specify the PDP
above is of size $1\times\left|S\right|\times\left(1+n\right)$.
\begin{equation}
F_{\ge\left|S\right|}\left(\mathbf{1}_{T}\right)=\left(\left.\beta\prod_{i\in\mathbb{Z}_{n}}\left(\frac{\partial}{\partial x_{i}}\right)^{\mathbf{1}_{T}\left[i\right]}\prod_{0\le j<n}\left(\beta_{j}+\sum_{i\in\mathbb{Z}_{n}}x_{i}\right)\text{mod}\left\{ \begin{array}{c}
\left(x_{i}\right)^{2}-x_{i}\\
i\in\mathbb{Z}_{n}
\end{array}\right\} \right\rfloor _{\mathbf{x}=\mathbf{0}_{n\times1}}\right).\label{Opt_PDP_B}
\end{equation}
The hypermatrix underlying the depth--3 $\sum\prod\sum$ arithmetic
formula used to specify the PDP is of size $1\times n\times\left(1+n\right)$. 
\end{proof}
We see that Ben Or \cite{WN96} constructions yield optimal PDEs for
Boolean functions $F_{\le\left|S\right|}$, $F_{\ge\left|S\right|}$
specified respectively via polynomials $P_{\le\left|S\right|}\left(\mathbf{x}\right)$
and $P_{\ge\left|S\right|}\left(\mathbf{x}\right)$ whose Chow rank
are at most $O\left(n\right)$ however Thrm. (\ref{Optimal PDPs})
establishes that PDPs for the same Boolean function are specified
respectively via polynomials $Q_{\le\left|S\right|}\left(\mathbf{x}\right)$
and $Q_{\ge\left|S\right|}\left(\mathbf{x}\right)$ whose Chow rank
is equal to one. We further remark that the distinction between PDEs
and PDPs is akin to the distinctions between time complexity \cite{Wig19}
and Kolmogorov complexity \cite{Kolmogorovcomplexity}.

\section{PDEs/PDPs over the transformation monoid $\mathbb{Z}_{n}^{\mathbb{Z}_{n}}$.}

We emphasize salient features of PDEs/PDPs by focusing on Boolean
functions whose domain are members of the transformation monoid $\mathbb{Z}_{n}^{\mathbb{Z}_{n}}$
in other words functions whose domain and codomain is $\mathbb{Z}_{n}$.
Given an arbitrary $T\subseteq\mathbb{Z}_{n}^{\mathbb{Z}_{n}}$, let
$F_{T}\,:\mathbb{Z}_{n}^{\mathbb{Z}_{n}}\rightarrow\left\{ 0,1\right\} $
be such that
\begin{equation}
F_{T}\left(f\right)=\begin{cases}
\begin{array}{cc}
1 & \text{ if }f\in T\\
0 & \text{otherwise}
\end{array} & \text{for all }f\in\mathbb{Z}_{n}^{\mathbb{Z}_{n}}\end{cases}.\label{Monoid-Boolean-Function}
\end{equation}
The Boolean function $F_{T}$ therefore tests for membership of an
input function $f\in\mathbb{Z}_{n}^{\mathbb{Z}_{n}}$ into the subset
$T$. PDEs of $F_{T}$ with exponent parameter $m$ are of the form
\[
F_{T}\left(f\right)=\left(\frac{\partial^{n}\,P_{T}\left(\mathbf{A}\right)}{\underset{i\in\mathbb{Z}_{n}}{\prod}\partial a_{i,f\left(i\right)}}\right)^{m},\:\text{ for all }\:f\in\mathbb{Z}_{n}^{\mathbb{Z}_{n}},
\]
where
\[
P_{T}\left(\mathbf{A}\right)\in\left\{ \sum_{g\in T}\omega_{g}\,\prod_{i\in\mathbb{Z}_{n}}a_{i,g\left(i\right)}\,:\,\begin{array}{c}
\left(\omega_{g}\right)^{m}=1,\\
g\in T
\end{array}\right\} .
\]
Note when expressing such PDEs, evaluations of each entries of $\mathbf{A}$
at $0$ are no longer needed. Furthermore it is easy to see that for
all $S,T\subset\mathbb{Z}_{n}^{\mathbb{Z}_{n}}$,
\[
\neg\,F_{T}\left(f\right)=\left(\frac{\partial^{n}\,P_{\overline{T}}\left(\mathbf{A}\right)}{\underset{i\in\mathbb{Z}_{n}}{\prod}\partial a_{i,f\left(i\right)}}\right)^{m},\quad\text{ for all }\:f\in\mathbb{Z}_{n}^{\mathbb{Z}_{n}},
\]
where $\overline{T}:=\mathbb{Z}_{n}^{\mathbb{Z}_{n}}\backslash T$.
\[
F_{S}\left(f\right)\vee F_{T}\left(f\right)=\left(\frac{\partial^{n}\,P_{S\cup T}\left(\mathbf{A}\right)}{\underset{i\in\mathbb{Z}_{n}}{\prod}\partial a_{i,f\left(i\right)}}\right)^{m},\quad\text{ for all }\:f\in\mathbb{Z}_{n}^{\mathbb{Z}_{n}},
\]
and
\[
F_{S}\left(f\right)\wedge F_{T}\left(f\right)=\left(\frac{\partial^{n}\,P_{S\cap T}\left(\mathbf{A}\right)}{\underset{i\in\mathbb{Z}_{n}}{\prod}\partial a_{i,f\left(i\right)}}\right)^{m},\quad\text{ for all }\:f\in\mathbb{Z}_{n}^{\mathbb{Z}_{n}}.
\]
There are $m^{\left|T\right|}$ distinct choices for $P_{T}\left(\mathbf{A}\right)$
with Chow-rank trivially upper-bounded by $\left|T\right|$. 
\begin{example}
Take 
\[
T=\left\{ f\in\mathbb{Z}_{n}^{\mathbb{Z}_{n}}\begin{array}{c}
f\left(0\right)=0\\
f\left(i\right)<i,\,\forall\,i\in\mathbb{Z}_{n}\backslash\left\{ 0\right\} 
\end{array}\right\} .
\]
the an optimal PDE of $F_{T}$ with exponent parameter $m=1$ is given
by 
\[
F_{T}\left(f\right)=\left(\frac{\partial^{n}\,P_{T}\left(\mathbf{A}\right)}{\underset{i\in\mathbb{Z}_{n}}{\prod}\partial a_{i,f\left(i\right)}}\right)^{m},\:\text{ for all }\:f\in\mathbb{Z}_{n}^{\mathbb{Z}_{n}},
\]
where 
\[
P_{T}\left(\mathbf{A}\right)=\mathbf{A}\left[0,0\right]\prod_{i\in\mathbb{Z}_{n}\backslash\left\{ 0\right\} }\sum_{i<j<n}\mathbf{A}\left[i,j\right]=\sum_{\begin{array}{c}
f\in\mathbb{Z}_{n}^{\mathbb{Z}_{n}}\\
f\left(0\right)=0\\
f\left(i\right)<i\,\forall\,i\in\mathbb{Z}_{n}\backslash\left\{ 0\right\} 
\end{array}}\prod_{i\in\mathbb{Z}_{n}}\mathbf{A}\left[i,f\left(i\right)\right].
\]
Consider for example the Boolean function
\end{example}

\[
F_{\text{S}_{n}}:\left(\mathbb{Z}_{n}\right)^{\mathbb{Z}_{n}}\rightarrow\left\{ 0,1\right\} 
\]
defined such that 
\[
F_{\text{S}_{n}}\left(f\right)=\begin{cases}
\begin{array}{cc}
1 & \text{ if }f\in\text{S}_{n}\\
0 & \text{otherwise}
\end{array} & \text{for all }f\in\mathbb{Z}_{n}^{\mathbb{Z}_{n}}.\end{cases}
\]
The corresponding PDE with exponent parameter $m$ specified via a
multilinear polynomial is given by 
\[
F_{\text{S}_{n}}\left(f\right)=\left(\left.\left(\prod_{i\in\mathbb{Z}_{\sqrt{n}}}\frac{\partial}{\partial a_{i,f\left(i\right)}}\right)P_{\text{S}_{n}}\left(\mathbf{A}\right)\right\rfloor _{\mathbf{A}=\mathbf{0}_{n\times n}}\right)^{m},
\]
where 
\[
P_{\text{S}_{n}}\left(\mathbf{A}\right)\in\left\{ \sum_{\sigma\in\text{S}_{n}}\omega_{\sigma}\prod_{i\in\mathbb{Z}_{n}}a_{i,\sigma\left(i\right)}:\begin{array}{c}
\left(\omega_{\sigma}\right)^{m}=1\\
\sigma\in\text{S}_{n}
\end{array}\right\} .
\]
Alternatively an optimal PDE with exponent parameter $m=2$ is devised
 for $F_{\text{S}_{n}}$ by specifying it using non-multilinear polynomial
as follows
\[
F_{\text{S}_{n}}\left(f\right)=\left(\left.\prod_{i\in\mathbb{Z}_{\sqrt{n}}}\left(\frac{\partial}{\sqrt[f\left(i\right)]{f\left(i\right)!}\,\partial x_{i}}\right)^{f\left(i\right)}p_{\text{S}_{n}}\left(x_{0},\cdots,x_{n-1}\right)\right\rfloor _{\mathbf{x}=\mathbf{0}_{n\times1}}\right)^{2},
\]
where 
\[
p_{\text{S}_{n}}\left(x_{0},\cdots,x_{n-1}\right)\in\left\{ \prod_{0\le i<j<n}\left(\omega_{v}\,x_{v}-\omega_{u}\,x_{u}\right):\begin{array}{c}
\left(\omega_{u}\right)^{2}=1\\
u\in\mathbb{Z}_{n}
\end{array}\right\} .
\]

\begin{prop}
Let $\mathcal{O}_{\mathbf{Z}}$ denote the orbital vector $n\times1$
vector with entries 
\[
\mathcal{O}_{\mathbf{Z}}\left[i\right]=\prod_{\gamma\in\text{S}_{n}}\mathbf{Z}\left[\gamma\left(i\right),\text{lex}\left(\gamma\right)\right],\quad\forall\:i\in\mathbb{Z}_{n}.
\]
 The $\left[0,1\right]$ entry of the canonical representative of
the congruence class
\[
\prod_{i\in\mathbb{Z}_{n}}\left(\mathcal{O}_{\mathbf{Z}}\left[i\right]\right)^{i}\text{ mod}\left\{ \begin{array}{c}
\underset{i\in\mathbb{Z}_{n}}{\prod}\left(\mathbf{Z}\left[\gamma\left(i\right),\text{lex}\left(\gamma\right)\right]\right)^{i}-\left(\begin{array}{cc}
1 & \underset{i\in\mathbb{Z}_{n}}{\prod}\left(x_{\gamma\left(i\right)}\right)^{i}\\
0 & 1
\end{array}\right)\\
\gamma\in\text{S}_{n}
\end{array}\right\} 
\]
yields the polynomial $p_{\text{S}_{n}}\left(x_{0},\cdots,x_{n-1}\right)$
used to specify the PDE 
\[
F_{\text{S}_{n}}\left(f\right)=\left(\left.\prod_{i\in\mathbb{Z}_{\sqrt{n}}}\left(\frac{\partial}{\sqrt[f\left(i\right)]{f\left(i\right)!}\,\partial x_{i}}\right)^{f\left(i\right)}p_{\text{S}_{n}}\left(x_{0},\cdots,x_{n-1}\right)\right\rfloor _{\mathbf{x}=\mathbf{0}_{n\times1}}\right)^{m},
\]
with exponent parameter $m=1$.
\end{prop}

\begin{proof}
When the exponent parameter $m=1$, 
\[
p_{\text{S}_{n}}\left(x_{0},\cdots,x_{n-1}\right)=\text{Per}\left\{ \text{Vandermonde}\left(\mathbf{x}\right)\right\} .
\]
Recall that 
\[
\left(\text{Vandermonde}\left(\mathbf{x}\right)\right)\left[i,j\right]=\left(x_{i}\right)^{j},\ \forall\,0\le i,j<n.
\]
We see that 
\[
\left(\prod_{i\in\mathbb{Z}_{n}}\left(\mathcal{O}_{\mathbf{Z}}\left[i\right]\right)^{i}\text{ mod}\left\{ \begin{array}{c}
\underset{i\in\mathbb{Z}_{n}}{\prod}\left(\mathbf{Z}\left[\gamma\left(i\right),\text{lex}\left(\gamma\right)\right]\right)^{i}-\left(\begin{array}{cc}
1 & \underset{i\in\mathbb{Z}_{n}}{\prod}\left(x_{\gamma\left(i\right)}\right)^{i}\\
0 & 1
\end{array}\right)\\
\gamma\in\text{S}_{n}
\end{array}\right\} \right)\left[0,1\right]\equiv\text{Per}\left\{ \text{Vandermonde}\left(\mathbf{x}\right)\right\} .
\]
From which the desired claim follows.
\end{proof}
\begin{prop}
An optimal PDE for the Boolean function $F_{\text{S}_{n}}$ with exponent
parameter $m=2$ is 
\[
F_{\text{S}_{n}}\left(f\right)=\left(\left.\prod_{i\in\mathbb{Z}_{\sqrt{n}}}\left(\frac{\partial}{\sqrt[f\left(i\right)]{f\left(i\right)!}\,\partial x_{i}}\right)^{f\left(i\right)}p_{\text{inv}}\left(x_{0},\cdots,x_{n-1}\right)\right\rfloor _{\mathbf{x}=\mathbf{0}_{n\times1}}\right)^{m},
\]
where 
\[
p_{\text{S}_{n}}\in\left\{ \det\left(\text{Vandermonde}\left(\text{diag}\left(\mathbf{s}\right)\mathbf{x}\right)\right):\mathbf{I}_{n}=\text{diag}\left(\mathbf{s}\right)^{m}\right\} \subset\left\{ \sum_{\sigma\in\text{S}_{n}}\omega_{\sigma}\prod_{i\in\mathbb{Z}_{n}}\left(x_{\sigma\left(i\right)}\right)^{i}:\begin{array}{c}
\left(\omega_{\sigma}\right)^{m}=1\\
\sigma\in\text{S}_{n}
\end{array}\right\} .
\]
\end{prop}

\begin{proof}
The proof immediately follows from the Chow-rank one decomposition
of the well known determinant of the Vandermonde matrix given by 
\[
\prod_{0\le i<j<n}\left(x_{j}-x_{i}\right)=\det\left(\text{Vandermonde}\left(\mathbf{x}\right)\right).
\]
\end{proof}
\begin{defn}
The subset $T$ is a \emph{normal} subset of $\mathbb{Z}_{n}^{\mathbb{Z}_{n}}$
if
\[
T=\sigma T\sigma^{\left(-1\right)}:=\left\{ \sigma f\sigma^{\left(-1\right)}:f\in T\right\} ,\quad\text{for all}\:\sigma\in\text{S}_{n}.
\]
For instance $\mathbb{Z}_{n}^{\mathbb{Z}_{n}}$, S$_{n}$, S$_{n}$$\backslash\mathbb{Z}_{n}^{\mathbb{Z}_{n}}$
and $\left\{ f\in\mathbb{Z}_{n}^{\mathbb{Z}_{n}}:1=\left|f^{\left(n-1\right)}\left(\mathbb{Z}_{n}\right)\right|\right\} $
all form normal subsets of the transformation monoid $\mathbb{Z}_{n}^{\mathbb{Z}_{n}}$.
Consider linear transformations prescribed with respect to the standard
Euclidean basis over the finite field with $p$ elements ($\mathbb{F}_{p}$)
where $p$ is prime. With respect to the standard Euclidean basis,
linear transformations are represented by matrices in $\mathbb{F}_{p}^{\left\lfloor \text{log}_{p}\left(n\right)\right\rfloor \times\left\lfloor \text{log}_{p}\left(n\right)\right\rfloor }$.
We abuse of notation slightly and view such linear transformations
as members of the transformation monoid $\left(\mathbb{F}_{p}^{\left\lfloor \text{log}_{p}\left(n\right)\right\rfloor \times1}\right)^{\left(\mathbb{F}_{p}\right)^{\left\lfloor \text{log}_{p}\left(n\right)\right\rfloor \times1}}$.
Assume for simplicity throughout the subsequent discussion that $n$
is a power of $p$ and let the canonical embedding of the vector space
$\mathbb{F}_{p}^{\text{log}_{p}\left(n\right)\times1}$ into $\mathbb{Z}_{n}$
be prescribed by the map
\[
\eta\,:\,\mathbb{F}_{p}^{\left\lfloor \text{log}_{p}\left(n\right)\right\rfloor \times1}\rightarrow\mathbb{Z}_{n},\:\text{ such that }\:\eta\left(\mathbf{b}\right)=\sum_{0\le i<\text{lg}n}b_{i}\,p^{i},\text{ for all }\mathbf{b}\in\left(\mathbb{F}_{p}\right)^{\text{log}_{p}\left(n\right)\times1}.
\]
Via this embedding, the monoid of linear transformations from $\mathbb{F}_{p}^{\text{log}_{p}\left(n\right)\times1}$
to $\mathbb{F}_{p}^{\text{log}_{p}\left(n\right)\times1}$ is isomorphic
to a sub-monoid of $\mathbb{Z}_{n}^{\mathbb{Z}_{n}}$ of order $n^{\text{lg}n}$.
We call this particular monoid the monoid of endomorphism and is denoted
for notational convenience $\text{End}_{\text{log}_{p}\left(n\right)}\left(\mathbb{F}_{p}\right)$.
The largest group in $\text{End}_{\text{log}_{p}\left(n\right)}\left(\mathbb{F}_{p}\right)$
is isomorphic to GL$_{\text{log}_{p}\left(n\right)}\left(\mathbb{F}_{p}\right)$
thus isomorphic to a subgroup of the permutation group S$_{n}\subseteq\mathbb{Z}_{n}^{\mathbb{Z}_{n}}$
of order
\[
\prod_{0\le k<\text{log}_{p}\left(n\right)}\left(n-p^{k}\right).
\]
We abuse notion and identify GL$_{\text{log}_{p}\left(n\right)}\left(\mathbb{F}_{2}\right)$
with its isomorphic image in S$_{n}\subseteq\mathbb{Z}_{n}^{\mathbb{Z}_{n}}$.
Consider the $\text{log}_{p}\left(n\right)\times\text{log}_{p}\left(n\right)$
orbital matrix $\mathcal{O}_{\mathbf{Z}}$ depicts orbits of the action
of the Abelian group $\mathbb{F}_{p}^{\text{log}_{p}\left(n\right)\times1}$
\[
\mathcal{O}_{\mathbf{Z}}\left[i,j\right]=\prod_{\mathbf{b}\in\mathbb{F}_{p}^{\text{log}_{p}\left(n\right)}}\left(\mathbf{Z}\left[i,\eta\left(\mathbf{b}\right)\right]\right)^{b_{j}},\ \forall\,\left(i,j\right)\in\mathbb{Z}_{\text{log}_{p}\left(n\right)}\times\mathbb{Z}_{\text{log}_{p}\left(n\right)}.
\]
Using the orbital matrix we devise a listing of $\text{End}_{\text{log}_{p}\left(n\right)}$
as follows
\end{defn}

\begin{prop}
The $\left[0,1\right]$ entry of the $2\times2$ canonical representative
of the congruence class
\[
\prod_{0\le i,j<\text{log}_{p}\left(n\right)}\sum_{k\in\mathbb{F}_{p}}\left(\mathcal{O}_{\mathbf{Z}}\left[i,j\right]\right)^{k}
\]
\[
\text{ mod }\left\{ \prod_{\begin{array}{c}
0\le u<\text{log}_{p}\left(n\right)\\
\mathbf{b}\in\mathbb{F}_{p}^{\text{log}_{p}\left(n\right)}
\end{array}}\left(\mathbf{Z}\left[u,\eta\left(\mathbf{b}\right)\right]\right)^{\underset{0\le v<\text{log}_{p}\left(n\right)}{\sum}\mathbf{M}\left[u,v\right]b_{v}}-\left(\begin{array}{cc}
1 & \underset{i\in\mathbb{Z}_{n}}{\prod}\mathbf{A}\left[i,\eta\left(\mathbf{M}\eta^{-1}\left(i\right)\right)\right]\\
0 & 1
\end{array}\right):\mathbf{M}\in\mathbb{F}_{p}^{\text{log}_{p}\left(n\right)\times\text{log}_{p}\left(n\right)}\right\} ,
\]
yields the listing
\[
P_{\text{End}_{\text{log}_{p}\left(n\right)}}\left(\mathbf{A}\right)=\sum_{f\in\text{End}_{\text{log}_{p}\left(n\right)}\left(\mathbb{F}_{p}\right)}\prod_{i\in\mathbb{Z}_{n}}\mathbf{A}\left[i,f\left(i\right)\right].
\]
\end{prop}

\begin{proof}
The proof stems from the matrix identity
\[
\prod_{0\le i,j<\text{log}_{p}\left(n\right)}\sum_{0\le k<p}\left(\mathbf{X}\left[i,j\right]\right)^{k}=\sum_{\mathbf{M}\in\mathbb{F}_{p}^{\text{log}_{p}\left(n\right)\times\text{log}_{p}\left(n\right)}}\prod_{0\le i,j<\text{log}_{p}\left(n\right)}\left(\mathbf{X}\left[i,j\right]\right)^{\mathbf{M}\left[i,j\right]}.
\]
Substituting into the polynomial identity the orbital matrix $\mathcal{O}_{\mathbf{Z}}$
for the matrix $\mathbf{X}$ yields the equality
\[
\prod_{0\le i,j<n}\sum_{0\le k<p}\left(\prod_{\mathbf{b}\in\mathbb{F}_{p}^{\text{log}_{p}\left(n\right)}}z_{i,\eta\left(\mathbf{b}\right)}\right)^{k\,b_{j}}=\sum_{\mathbf{M}\in\mathbb{F}_{p}^{\text{log}_{p}\left(n\right)\times\text{log}_{p}\left(n\right)}}\prod_{\mathbf{b}\in\mathbb{F}_{p}^{\text{log}_{p}\left(n\right)}}\left(\prod_{0\le i<\text{log}_{p}\left(n\right)}\left(z_{i,\eta\left(\mathbf{b}\right)}\right)^{\underset{0\le j<\text{log}_{p}\left(n\right)}{\sum}\mathbf{M}\left[i,j\right]b_{j}}\right).
\]
The canonical representative for the first of these congruence classes
is obtained by successively replacing into the expanded form of $P_{\subseteq S}\left(\mathcal{O}_{\mathbf{Z}}\right)$
every occurrence of monomials of the form
\[
\prod_{\begin{array}{c}
0\le u,v<\text{log}_{p}\left(n\right)\\
\mathbf{b}\in\mathbb{F}_{p}^{\text{log}_{p}\left(n\right)}
\end{array}}\left(\mathbf{Z}\left[u,\eta\left(\mathbf{b}\right)\right]\right)^{\mathbf{M}\left[u,v\right]b_{v}},
\]
with the $2\times2$ upper triangular matrix
\[
\left(\begin{array}{cc}
1 & \underset{i\in\mathbb{Z}_{n}}{\prod}\mathbf{A}\left[i,\eta\left(\mathbf{M}\eta^{-1}\left(i\right)\right)\right]\\
0 & 1
\end{array}\right),
\]
followed by taking the $\left[0,1\right]$ entry of the $2\times2$
matrix resulting from the said substitutions.
\end{proof}
The ensuing PDE with exponent parameter $m=1$ is 
\[
F_{\text{End}_{\text{log}_{p}\left(n\right)}}\left(f\right)=\left(\left.\left(\prod_{i\in\mathbb{Z}_{n}}\frac{\partial}{\partial a_{i,f\left(i\right)}}\right)P_{\text{End}_{\text{log}_{p}\left(n\right)}}\left(\mathbf{A}\right)\right\rfloor _{\mathbf{A}=\mathbf{0}_{n\times n}}\right).
\]
For simplicity take $p=2$ and let lg denote the logarithm base 2.
The canonical embedding of the vector space $\left(\mathbb{F}_{2}\right)^{\text{lg}n\times1}$
into $\mathbb{Z}_{n}$ is prescribed by the map
\[
\eta\,:\,\left(\mathbb{F}_{2}\right)^{\text{lg}n\times1}\rightarrow\mathbb{Z}_{n},\:\text{ such that }\:\eta\left(\mathbf{b}\right)=\sum_{0\le i<\text{lg}n}b_{i}\,2^{i},\text{ for all }\mathbf{b}\in\left(\mathbb{F}_{2}\right)^{\text{lg}n\times1}.
\]
Via this embedding, the monoid of linear transformations from $\left(\mathbb{F}_{2}\right)^{\text{lg}n\times1}$
to $\left(\mathbb{F}_{2}\right)^{\text{lg}n\times1}$ is isomorphic
to a sub-monoid of $\left(\mathbb{Z}_{n}\right)^{\mathbb{Z}_{n}}$
of order $n^{\text{lg}n}$. We call this particular monoid the monoid
of endomorphism and is denoted for notational convenience $\text{End}_{\text{lg}n}\left(\mathbb{F}_{2}\right)$.
Boolean functions of fundamental importance are Boolean function which
test for membership in a subset $T\subset\text{End}_{\text{lg}n}\left(\mathbb{F}_{2}\right)\subset\mathbb{Z}_{n}^{\mathbb{Z}_{n}}$
subject to the invariance
\[
\sigma T\gamma=T,\quad\forall\,\left(\sigma,\gamma\right)\in\text{GL}_{\text{lg}n}\left(\mathbb{F}_{2}\right)\times\text{GL}_{\text{lg}n}\left(\mathbb{F}_{2}\right).
\]
Their importance stem from their invariance to coordinate change.
PDPs of $F_{T}$ are of the form 
\[
F_{T}\left(f\right)=\left(\frac{\partial^{n}}{\underset{i\in\mathbb{Z}_{n}}{\prod}\partial a_{i,f\left(i\right)}}\left(Q_{T}\left(\mathbf{A}\right)\text{mod}\left\{ \begin{array}{c}
\left(a_{ij}\right)^{2}-a_{ij}\\
0\le i,j<n
\end{array}\right\} \right)\right)^{m},\;\text{ for all}\:f\in\mathbb{Z}_{n}^{\mathbb{Z}_{n}}.
\]
Let the rank of $f\in\text{End}_{\text{lg}n}\left(\mathbb{F}_{2}\right)$
denote the rank of the corresponding $\text{lg}n\times\text{lg}n$
matrix.
\begin{thm}
\label{PDP-Chow-Bound}Let $S\subset\mathbb{Z}_{1+\text{lg}n}$ and
$T\subset\text{End}_{\text{lg}n}\left(\mathbb{F}_{2}\right)\subset\mathbb{Z}_{n}^{\mathbb{Z}_{n}}$
be such that 
\[
T=\left\{ f\in\text{End}_{\text{lg}n}\left(\mathbb{F}_{2}\right):\text{Rank}\left(f\right)\in S\right\} 
\]
then there exist a PDP for $F_{T}$ specified via a polynomial $Q_{T}\left(\mathbf{A}\right)$
up to binary algebraic relations of Chow rank at most $\left\lceil \frac{1+\text{lg}n}{\left(1+n^{2}\right)n}\right\rceil =1$
whereby 
\[
\left(Q_{T}\left(\mathbf{A}\right)\text{mod}\left\{ \begin{array}{c}
\left(a_{ij}\right)^{2}-a_{ij}\\
0\le i,j<n
\end{array}\right\} \right)\equiv P_{T}\left(\mathbf{A}\right)\in\left\{ \sum_{g\in T}\omega_{g}\,\prod_{i\in\mathbb{Z}_{n}}a_{i\,g\left(i\right)}\,:\,\begin{array}{c}
\left(\omega_{g}\right)^{m}=1,\\
g\in T
\end{array}\right\} 
\]
\end{thm}

\begin{proof}
For simplicity we take the exponent parameter $m=1$. Consider the
symbolic listing of $\text{End}_{\text{lg}n}\left(\mathbb{F}_{2}\right)$
given by 
\[
\sum_{\mathbf{M}\in\left(\mathbb{F}_{2}\right)^{\text{lg}n\times\text{lg}n}}\prod_{i\in\mathbb{Z}_{n}}\mathbf{A}\left[i,\,\eta\left(\mathbf{M}\eta^{\left(-1\right)}\left(i\right)\right)\right]=\sum_{f\in\text{End}_{\text{lg}n}\left(\mathbb{F}_{2}\right)}\prod_{i\in\mathbb{Z}_{n}}a_{i\,f\left(i\right)}.
\]
Given the prescribed invariance 
\[
\sigma T\gamma=T,\:\text{ for all }\:\left(\sigma,\gamma\right)\in\text{GL}_{\text{lg}n}\left(\mathbb{F}_{2}\right)\times\text{GL}_{\text{lg}n}\left(\mathbb{F}_{2}\right),
\]
We consider the $n\times n$ orbital matrix associated with the corresponding
group action as follows
\[
\mathcal{O}_{\mathbf{Z}}\left[i,j\right]=\prod_{\left(\sigma,\gamma\right)\in\text{GL}_{\text{lg}n}\left(\mathbb{F}_{2}\right)\times\text{GL}_{\text{lg}n}\left(\mathbb{F}_{2}\right)}\mathbf{Z}\left[\sigma^{\left(-1\right)}\left(i\right),\gamma\left(j\right),\text{lex}_{\text{GL}}\left(\sigma,\gamma\right)\right],\text{ for all }0\le i,j<n.
\]
The lexicographic map above is an arbitrary bijection from $\text{GL}_{\text{lg}n}\left(\mathbb{F}_{2}\right)\times\text{GL}_{\text{lg}n}\left(\mathbb{F}_{2}\right)$
to $\mathbb{Z}_{\left|\text{GL}_{\text{lg}n}\left(\mathbb{F}_{2}\right)\right|^{2}}$.
Let $\text{lex}_{\text{End}}$ denote an arbitrary bijection from
$\text{End}_{\text{lg}n}\left(\mathbb{F}_{2}\right)$ to $\mathbb{Z}_{n^{\lg n}}$
then the corresponding orbit list generating polynomial is
\[
\sum_{f\in\text{End}_{\text{lg}n}\left(\mathbb{F}_{2}\right)}\prod_{i\in\mathbb{Z}_{n}}\mathcal{O}_{\mathbf{Z}}\left[i,f\left(i\right)\right]\mod\left\{ \begin{array}{c}
\underset{i\in\mathbb{Z}_{2^{n}}}{\prod}\mathbf{Z}\left[i,f\left(i\right),\text{lex}_{\text{GL}}\left(\sigma,\gamma\right)\right]-\underset{j\in\mathbb{Z}_{n}}{\prod}\mathbf{Y}\left[j,f\left(j\right),\text{lex}_{\text{End}}\left(f\right)\right]\\
\left(\sigma,\gamma\right)\in\text{GL}_{\text{lg}n}\left(\mathbb{F}_{2}\right)\times\text{GL}_{\text{lg}n}\left(\mathbb{F}_{2}\right)\\
f\in\text{End}_{\text{lg}n}\left(\mathbb{F}_{2}\right)
\end{array}\right\} ,
\]
and is given by
\[
\sum_{0\le t\le\lg n}\left|\nicefrac{\left(\text{GL}_{\text{lg}n}\left(\mathbb{F}_{2}\right)\times\text{GL}_{\text{lg}n}\left(\mathbb{F}_{2}\right)\right)}{\mathcal{A}\left(f\right)}\right|\prod_{\begin{array}{c}
f\in\text{End}_{\text{lg}n}\left(\mathbb{F}_{2}\right)\\
\text{rank}\left(f\right)=t
\end{array}}\left(\prod_{i\in\mathbb{Z}_{n}}\mathbf{Y}\left[i,f\left(i\right),\text{lex}\left(f\right)\right]\right)^{\left|\mathcal{A}\left(f\right)\right|}
\]
where $\mathcal{A}\left(f\right)$ denotes the $f$--Left-Right invariant
subgroup of $\text{GL}_{\text{lg}n}\left(\mathbb{F}_{2}\right)\times\text{GL}_{\text{lg}n}\left(\mathbb{F}_{2}\right)$
in other words
\[
\mathcal{A}\left(f\right):=\left\{ \left(\sigma,\gamma\right)\in\text{GL}_{\text{lg}n}\left(\mathbb{F}_{2}\right)\times\text{GL}_{\text{lg}n}\left(\mathbb{F}_{2}\right):\sigma f\gamma=f\right\} .
\]
We bound the Chow-rank of $Q_{T}$ via the orbital argument. Consider
the equality 
\[
\sum_{f\in\text{End}_{\text{lg}n}\left(\mathbb{F}_{2}\right)}\prod_{i\in\mathbb{Z}_{n}}a_{i,f\left(i\right)}=\sum_{0\le u<\rho}\prod_{0\le v<d}\left(\mathbf{H}\left[u,v,0\right]+\sum_{0\le i,j<n}\mathbf{H}\left[u,v,1+n\,i+j\right]\,a_{i\,j}\right)
\]
Substituting on both side of the equal sign above entries of $\mathbf{A}$
by corresponding entries of the orbital matrix yields 
\[
\sum_{f\in\text{End}_{\lg n}\left(\mathbb{F}_{2}\right)}\prod_{i\in\mathbb{Z}_{n}}\mathcal{O}_{\mathbf{Z}}\left[i,f\left(i\right)\right]\equiv\sum_{f\in\text{End}_{\lg n}\left(\mathbb{F}_{2}\right)}K_{f}\left(\mathbf{H}\right)\,\prod_{i\in\mathbb{Z}_{n}}\mathcal{O}_{\mathbf{Z}}\left[i,f\left(i\right)\right]\mod\left\{ \begin{array}{c}
\left(a_{ij}\right)^{2}-a_{ij}\\
0\le i,j<n
\end{array}\right\} ,
\]
where the polynomial $K_{f}\left(\mathbf{H}\right)$ is given by
\[
K_{f}\left(\mathbf{H}\right)=\left.\sum_{\left\{ 1\le d_{k}<\rho\,:\,k\right\} }\prod_{k\in\mathbb{Z}_{n}}\left(\frac{\partial}{\sqrt[d_{k}]{d_{k}!}\,\partial a_{if\left(i\right)}}\right)^{d_{k}}\sum_{0\le u<\rho}\,\prod_{0\le v<d}\left(\mathbf{H}\left[u,v,0\right]+\sum_{0\le i,j<n}\mathbf{H}\left[u,v,1+n\,i+j\right]\,a_{ij}\right)\right\rfloor _{\mathbf{A}=\mathbf{0}_{n\times n}}
\]
with the corresponding coefficients in the canonical representative
of the congruence class
\[
\sum_{f\in\text{End}_{\lg n}\left(\mathbb{F}_{2}\right)}\prod_{i\in\mathbb{Z}_{n}}\mathcal{O}_{\mathbf{Z}}\left[i,f\left(i\right)\right]\text{mod}\left\{ \begin{array}{c}
\underset{i\in\mathbb{Z}_{n}}{\prod}\mathbf{Z}\left[i,f\left(i\right),\text{lex}\left(\sigma,\gamma\right)\right]-\sqrt[\left|\mathcal{A}\left(f\right)\right|]{\frac{\underset{j\in\mathbb{Z}_{n}}{\prod}\mathbf{Y}\left[j,f\left(j\right),\text{lex}\left(f\right)\right]}{\nicefrac{\left(\text{GL}_{\text{lg}n}\left(\mathbb{F}_{2}\right)\times\text{GL}_{\text{lg}n}\left(\mathbb{F}_{2}\right)\right)}{\mathcal{A}\left(f\right)}}}\\
\left(\sigma,\gamma\right)\in\text{GL}_{\text{lg}n}\left(\mathbb{F}_{2}\right)\times\text{GL}_{\text{lg}n}\left(\mathbb{F}_{2}\right)\\
f\in\text{End}_{\text{lg}n}\left(\mathbb{F}_{2}\right)
\end{array}\right\} ,
\]
yields a systems of at most $\left(1+\lg n\right)$ equations in the
$\rho\cdot d\cdot\left(1+n^{2}\right)$ unknown entries for $\mathbf{H}\in\mathbb{C}^{\rho\times d\times n^{2}}$
respectively of the form
\[
\left\{ 1=\prod_{\begin{array}{c}
f\in\text{End}_{\lg n}\left(\mathbb{F}_{2}\right)\\
t=\text{Rank}\left(f\right)
\end{array}}K_{f}\left(\mathbf{H}\right)\,:\,t\in S\right\} \cup\left\{ 0=\prod_{\begin{array}{c}
f\in\text{End}_{\lg n}\left(\mathbb{F}_{2}\right)\\
t=\text{Rank}\left(f\right)
\end{array}}K_{f}\left(\mathbf{H}\right)\,:\,t\notin S\right\} .
\]
We know from the method of elimination via resultants that the corresponding
system necessarily admits a solution whenever the number unknowns
$\rho\cdot d\cdot\left(1+n^{2}\right)$ matches or exceeds the number
of algebraically independent constraints $\le1+\lg n$. Hence, when
\[
\rho=\left\lceil \frac{1+\lg n}{\left(1+n^{2}\right)\,d}\right\rceil =1,
\]
the number of variables matches or exceeds the number of algebraically
independent constraints. It follows from the degree lower bound $d\ge n$
that the desired claim holds. 
\end{proof}
In particular if we take $S=\left\{ \lg n\right\} $ then the corresponding
Boolean function is 
\[
F_{\text{GL}_{\text{lg}n}\left(\mathbb{F}_{2}\right)}\left(f\right)=\begin{cases}
\begin{array}{cc}
1 & \text{ if }f\in\text{GL}_{\text{lg}n}\left(\mathbb{F}_{2}\right)\\
\\
0 & \text{otherwise}
\end{array} & \text{for all }f\in\text{End}_{\text{lg}n}\left(\mathbb{F}_{2}\right).\end{cases}
\]
By Thrm. (10) $F_{\text{GL}_{\text{lg}n}\left(\mathbb{F}_{2}\right)}$
admits PDPs which can be specified via a polynomial $Q_{\text{GL}_{\text{lg}n}\left(\mathbb{F}_{2}\right)}\left(\mathbf{A}\right)$
whose Chow-rank is at most $\left(1+\lg n\right)$. If the orbital
argument used in the proof of Thrm. (\ref{PDP-Chow-Bound}) is carried
out for a Boolean function $F_{T}$ defined in Eq. (\ref{Monoid-Boolean-Function})
via the $n\times n$ orbital matrix
\[
\mathcal{O}_{\mathbf{Z}}\left[i,j\right]=\prod_{\left(\sigma,\gamma\right)\in\text{S}_{n}\times\text{S}_{n}}\mathbf{Z}\left[\sigma^{\left(-1\right)}\left(i\right),\gamma\left(j\right),\text{lex}_{\text{S}_{n}\times\text{S}_{n}}\left(\sigma,\gamma\right)\right],\text{ for all }0\le i,j<n,
\]
then it means that $F_{T}$ admits a PDP specified via a polynomial
$Q_{T}\left(\mathbf{A}\right)$ (up to binary algebraic relations)
of Chow-rank
\[
\rho\le\left\lceil \frac{\text{Pa}\left(n\right)}{\left(1+n^{2}\right)n}\right\rceil 
\]
where Pa$\left(n\right)$ denotes the integer partition function.

\section{Cauchy relations and PDP relaxations.}

We describe two optimal PDPs differing in their exponent parameter
for the Boolean function $F_{\text{S}_{n}}$ defined such that 
\[
F_{\text{S}_{n}}\left(f\right)=\begin{cases}
\begin{array}{cc}
1 & \text{ if }f\in\text{S}_{n}\\
0 & \text{otherwise}
\end{array} & \text{for all }f\in\mathbb{Z}_{n}^{\mathbb{Z}_{n}}.\end{cases}
\]
Our proposed PDPs will have exponent parameter $2$ and $1$ respectively
and will be specified to modulo a new set of polynomial size algebraic
relations presented in their expanded form.
\begin{thm}
\label{Optimal-PDP-F-inv-m2}There exist optimal PDPs for $F_{\text{S}_{n}}$
with exponent parameters $2$ specified via $\sum\prod\sum$ depth--3
arithmetic formulas whose underlying hypermatrices are of size $1\times{n+1 \choose 2}\times n^{2}$
\end{thm}

\begin{proof}
The proof follows from the observation that PDEs of $F_{\text{S}_{n}}$
where the exponent parameter is $m=2$, include expressions of the
form
\[
F_{\text{S}_{n}}\left(f\right)=\left(\frac{\partial^{n}\,\det\left(\mathbf{D}_{0}\mathbf{A}\mathbf{D}_{1}\right)}{\underset{i\in\mathbb{Z}_{n}}{\prod}\partial a_{i,f\left(i\right)}}\right)^{2},
\]
where $\mathbf{D}_{0}$ and $\mathbf{D}_{1}$ denote arbitrary diagonal
matrices whose diagonal entries are either $1$ or $-1$. By reducing
modulo Cauchy's algebraic relations 
\[
\left\{ \begin{array}{c}
a_{i,0}\left(a_{i,1}\right)^{j}\equiv a_{i,j}\\
\text{s.t.}\\
\begin{array}{c}
i\in\mathbb{Z}_{n}\\
0<j<n
\end{array}
\end{array}\right\} ,
\]
we optimally express det$\left(\mathbf{A}\right)$ via the Vandermonde
determinant identity
\begin{equation}
\det\left(\mathbf{A}\right)\equiv\left(\prod_{k\in\mathbb{Z}_{n}}a_{k,0}\right)\left(\prod_{0<i<j<n}\left(a_{j,1}-a_{i,1}\right)\right)\text{mod}\left\{ \begin{array}{c}
a_{i,0}\left(a_{i,1}\right)^{j}-a_{i,j}\\
\text{s.t.}\\
\begin{array}{c}
i\in\mathbb{Z}_{n}\\
0<j<n
\end{array}
\end{array}\right\} .\label{Determinant}
\end{equation}
Crucially, reductions modulo Cauchy's algebraic relations must be
performed in decreasing order of degrees of variables $\left\{ a_{i1}:i\in\mathbb{Z}_{n}\right\} .$
Thus completing the proof.
\end{proof}
\begin{thm}
\label{Optimal-PDP-F-inv-m1}There is an optimal PDPs for $F_{\text{S}_{n}}$
with exponent parameters $1$ specified via $\sum\prod\sum$ depth--3
arithmetic formulas whose underlying hypermatrices are of size $1\times n\times n^{2}$.
\end{thm}

\begin{proof}
The proof follows from the observation the PDE of $F_{\text{S}_{n}}$
where the exponent parameter is $m=1$, is given by
\[
F_{\text{S}_{n}}\left(f\right)=\frac{\partial^{n}\,\text{Per}\left(\mathbf{A}\right)}{\underset{i\in\mathbb{Z}_{n}}{\prod}\partial a_{i,f\left(i\right)}}.
\]
By reducing modulo algebraic relations
\[
\left\{ \begin{array}{c}
a_{i,k}\,a_{j,k}\equiv0\\
\text{s.t.}\\
\begin{array}{c}
0\le i<j<n\\
0\le k<n\\
\text{and}\\
a_{i,0}\left(a_{i,1}\right)^{j}\equiv a_{i,j}\\
\text{s.t.}\\
\begin{array}{c}
i\in\mathbb{Z}_{n}\\
0<j<n
\end{array}
\end{array}
\end{array}\right\} ,
\]
we optimally express Per$\left(\mathbf{A}\right)$ as follows
\begin{equation}
\text{Per}\left(\mathbf{A}\right)\equiv\left(\prod_{k\in\mathbb{Z}_{n}}a_{k,0}\right)\prod_{\begin{array}{c}
i\in\mathbb{Z}_{n}\\
0<j<n
\end{array}}\left(a_{i1}-\exp\left\{ \frac{2\pi\,j\sqrt{-1}}{n}\right\} \right)\text{ mod}\left\{ \begin{array}{c}
\left(a_{u,1}a_{v,1}\right)^{w}\\
\text{s.t.}\\
\begin{array}{c}
0\le u<v<n\\
0<w<n
\end{array}\\
\text{and}\\
a_{i,0}\left(a_{i,1}\right)^{j}-a_{i,j}\\
\text{s.t.}\\
\begin{array}{c}
i\in\mathbb{Z}_{n}\\
0<j<n
\end{array}
\end{array}\right\} .\label{Permanent}
\end{equation}
Thus completing the proof.
\end{proof}
The argument used to prove Thrm. (\ref{Optimal-PDP-F-inv-m2}) and
Thrm. (\ref{Optimal-PDP-F-inv-m1}) captures an important difference
separating the permanent from the determinant. Namely, the determinant
is obtained by reducing a Chow-rank one polynomial of total degree
${n+1 \choose 2}$ depending only on $2n$ variables taken from $\left\{ a_{ij}:0\le i,j<n\right\} $
modulo $2{n \choose 2}$ polynomial size algebraic relations presented
in their expanded form. Whereas the permanent is obtained by reducing
a Chow-rank one polynomial of total degree ${n+1 \choose 2}$ in $2n$
variables taken from $\left\{ a_{i,j}:0\le i,j<n\right\} $ modulo
$\left(n+1\right)\,{n \choose 2}$ algebraic relations presented in
their expanded form.
\begin{conjecture}
\label{Det-vs-Per-Conjecure}There exists no Chow--rank one polynomial
$Q_{\text{S}_{n}}\left(\mathbf{A}\right)$ of total degree $O\left(n^{2}\right)$
depending asymptotically only on $O\left(n\right)$ variables taken
from $\left\{ a_{i,j}:0\le i,j<n\right\} $ which reduces to Per$\left(\mathbf{A}\right)$
modulo $O\left(n^{2}\right)$ polynomial size algebraic relations
presented in their expanded form.
\end{conjecture}

By analogy to the determinant case, natural candidates for refuting
Conj. (\ref{Det-vs-Per-Conjecure}) are polynomial constructions devised
from the permanent of $\left(n-1\right)\times\left(n-1\right)$ Vandermonde
matrix given explicit as
\[
\sum_{\sigma\in\text{S}_{n-1}\subset\left(\mathbb{Z}_{n}\backslash\left\{ 0\right\} \right)^{\mathbb{Z}_{n}\backslash\left\{ 0\right\} }}\,\prod_{i\in\mathbb{Z}_{n}\backslash\left\{ 0\right\} }\left(a_{\sigma\left(i\right),1}\right)^{i-1}.
\]
Unfortunately when $n>3$, in contrast to the determinant setting,
the permanent of $\left(n-1\right)\times\left(n-1\right)$ Vandermonde
matrix has a non trivial Galois group over the field of fraction $\mathbb{Q}\left(a_{1,1},\cdots,a_{\left(k-1\right),1},a_{\left(k+1\right),1},\cdots,a_{\left(n-1\right),1}\right)$
when viewed as a univariate polynomial in the variable $a_{k,1}$
for all $0<k<n$. Consequently it does not split into linear factors
over $\mathbb{Q}\left(a_{1,1},\cdots,a_{\left(k-1\right),1},a_{\left(k+1\right),1},\cdots,a_{\left(n-1\right),1}\right)$.
Alternatively, we may consider the Chow-rank one polynomial construction
\[
\text{Per}\left(\mathbf{A}\right)\equiv\left(\prod_{k\in\mathbb{Z}_{n}}a_{k,0}\right)\left(\prod_{0<i<j<n}\left(a_{j,1}+a_{i,1}\right)\right)\text{mod}\left\{ \begin{array}{c}
\left(a_{u,1}a_{v,1}\right)^{w}\\
\text{s.t.}\\
\begin{array}{c}
0\le u<v<n\\
0<w<n
\end{array}\\
\text{and}\\
a_{i,0}\left(a_{i,1}\right)^{j}-a_{i,j}\\
\text{s.t.}\\
\begin{array}{c}
i\in\mathbb{Z}_{n}\\
0<j<n
\end{array}
\end{array}\right\} .
\]
Unfortunately we see that the construction above requires the same
number of algebraic relations.

Having obtained an optimal implicit description of the determinant
polynomial, we used it to devise other efficient PDPs. 
\begin{thm}
Let $T\subset\mathbb{Z}_{n}^{\mathbb{Z}_{n}}$ be defined such that
\[
T=\left\{ f\in\mathbb{Z}_{n}^{\mathbb{Z}_{n}}:\left|f^{\left(n-1\right)}\left(\mathbb{Z}_{n}\right)\right|=1\right\} ,
\]
then there exist a PDP with exponent parameter $m=1$, for the Boolean
function 
\[
F_{T}\left(f\right)=\begin{cases}
\begin{array}{cc}
1 & \text{ if }\left|f^{\left(n-1\right)}\left(\mathbb{Z}_{n}\right)\right|=1\\
0 & \text{otherwise}
\end{array}, & \text{for all}\:f\in\mathbb{Z}_{n}^{\mathbb{Z}_{n}}\end{cases},
\]
specified via a polynomial which admits a Chow decomposition of rank
at most $n$.
\end{thm}

\begin{proof}
The proof of the claim follows from Tutte's Directed Matrix Theorem
\cite{ZEILBERGER198561,Tutte1948} which asserts that 
\[
P_{T}\left(\mathbf{A}\right)=\left(\sum_{f\in T}\prod_{i\in\mathbb{Z}_{n}}a_{i\,f\left(i\right)}\right)=
\]
\[
\sum_{i\in\mathbb{Z}_{n}}\mathbf{A}\left[i,i\right]\det\left\{ \left(\text{diag}\left(\mathbf{A}\mathbf{1}_{n\times1}\right)-\mathbf{A}\right)\left[\begin{array}{cccccc}
0 & \cdots & i-1 & i+1 & \cdots & n-1\\
0 & \cdots & i-1 & i+1 & \cdots & n-1
\end{array}\right]\right\} 
\]
using the optimal implicit description for the determinant in Eq.
(\ref{Determinant}), the desired PDP stems from the identity 
\[
F_{T}\left(f\right)=\left(\frac{\partial^{n}}{\underset{i\in\mathbb{Z}_{n}}{\prod}\partial a_{i,f\left(i\right)}}\right)\sum_{i\in\mathbb{Z}_{n}}\mathbf{A}\left[i,i\right]\det\left\{ \left(\text{diag}\left(\mathbf{A}\mathbf{1}_{n\times1}\right)-\mathbf{A}\right)\left[\begin{array}{cccccc}
0 & \cdots & i-1 & i+1 & \cdots & n-1\\
0 & \cdots & i-1 & i+1 & \cdots & n-1
\end{array}\right]\right\} .
\]
\end{proof}
\begin{thm}
Let $T\subset\mathbb{Z}_{n-1}^{\mathbb{Z}_{n-1}}$ defined such that
\[
T=\left\{ f\in\mathbb{Z}_{n-1}^{\mathbb{Z}_{n-1}}:f^{\left(n-2\right)}\left(\mathbb{Z}_{n-1}\right)=\left\{ i:\begin{array}{c}
i\in\mathbb{Z}_{n-1}\\
f\left(i\right)=i
\end{array}\right\} \right\} ,
\]
then there exists a PDP with exponent parameter $1$ for the Boolean
function 
\[
F_{T}\left(f\right)=\begin{cases}
\begin{array}{cc}
1 & \text{ if }G_{f}\:\text{ contains no cycle of length}>1\\
0 & \text{otherwise}
\end{array}, & \text{for all}\:f\in\mathbb{Z}_{n-1}^{\mathbb{Z}_{n-1}}\end{cases},
\]
specified via a polynomial which admits a Chow decomposition of rank
at most $n$.
\end{thm}

\begin{proof}
The proof of the claim follows from Tutte's Directed Matrix Theorem
\cite{ZEILBERGER198561,Tutte1948} from which we get 
\[
P_{T}\left(\mathbf{A}\right)=\left(\sum_{f\in T}\prod_{i\in\mathbb{Z}_{n-1}}a_{i,f\left(i\right)}\right)=
\]
\[
\det\left\{ \left(\text{diag}\left(\mathbf{A}\mathbf{1}_{n\times1}\right)-\mathbf{A}\right)\left[:n-2,:n-2\right]\right\} \mod\left\{ \begin{array}{c}
\mathbf{A}\left[j,n-1\right]-\mathbf{A}\left[j,j\right]\\
j\in\mathbb{Z}_{n-1}
\end{array}\right\} 
\]
using the optimal implicit description for the determinant in Eq.
(\ref{Determinant}), the desired PDP stem from the identity 
\[
F_{T}\left(f\right)=\left(\frac{\partial^{n}\,P_{T}\left(\mathbf{A}\right)}{\underset{i\in\mathbb{Z}_{n}}{\prod}\partial a_{i,f\left(i\right)}}\right).
\]
\end{proof}
Note that the sets S$_{n}$, 
\[
\left\{ f\in\mathbb{Z}_{n}^{\mathbb{Z}_{n}}:\left|f^{\left(n-1\right)}\left(\mathbb{Z}_{n}\right)\right|=1\right\} ,
\]
as well as 
\[
\left\{ f\in\mathbb{Z}_{n}^{\mathbb{Z}_{n}}:f^{\left(n-1\right)}\left(\mathbb{Z}_{n}\right)=\left\{ i:\begin{array}{c}
i\in\mathbb{Z}_{n}\\
f\left(i\right)=i
\end{array}\right\} \right\} ,
\]
are normal subsets $\mathbb{Z}_{n}^{\mathbb{Z}_{n}}$ of size $n^{n-1}$
and $\left(n+1\right)^{\left(n-1\right)}$ respectively.

\subsection{PDP relaxations}

Recall that if $T$ denotes some arbitrary subset of the transformation
monoid $\mathbb{Z}_{n}^{\mathbb{Z}_{n}}$, then the corresponding
relaxation $F_{T}\,:\mathbb{Z}_{n}^{\mathbb{Z}_{n}}\rightarrow\mathbb{N}$
is such that
\[
F_{T}\left(f\right)\text{ is }\begin{cases}
\begin{array}{cc}
\ne0 & \text{ if }f\in T\\
0 & \text{otherwise}
\end{array} & \forall\:f\in\mathbb{Z}_{n}^{\mathbb{Z}_{n}}\end{cases}.
\]
We see that the non-vanishing support of $F_{T}$ implicitly tests
for membership of $f$ into $T$. Therefore relaxed PDPs of $F_{T}$
are of the form 
\[
F_{T}\left(f\right)=\left(\frac{\partial^{n}}{\underset{i\in\mathbb{Z}_{n}}{\prod}\partial a_{i,f\left(i\right)}}\left(Q_{T}\left(\mathbf{A}\right)\text{mod}\left\{ \begin{array}{c}
h_{u}\left(\mathbf{A}\right)\\
0\le u<k
\end{array}\right\} \right)\right)^{m},\quad\forall\:f\in\mathbb{Z}_{n}^{\mathbb{Z}_{n}},
\]
where $\left\{ \begin{array}{c}
h_{u}\left(\mathbf{A}\right)\\
0\le u<k
\end{array}\right\} $ denotes polynomial size set of algebraic relations presented in their
expanded form. By construction it must be the case that
\[
\left(Q_{T}\left(\mathbf{A}\right)\text{mod}\left\{ \begin{array}{c}
h_{u}\left(\mathbf{A}\right)\\
0\le u<k
\end{array}\right\} \right)\equiv P_{T}\left(\mathbf{A}\right)\in\left\{ \sum_{g\in T}\omega_{g}\,c_{g}\,\prod_{i\in\mathbb{Z}_{n}}a_{i,g\left(i\right)}\,:\,\begin{array}{c}
c_{g}\in\mathbb{N}\\
\left(\omega_{g}\right)^{m}=1\\
g\in T
\end{array}\right\} .
\]

\begin{thm}
Let T denotes the largest subset of permutations in $S_{n}\subset\mathbb{Z}_{n}^{\mathbb{Z}_{n}}$
whose graphs are a Spanning Union of Directed Even Cycles (SUDEC for
short) then
\[
F_{T}\,:\mathbb{Z}_{n}^{\mathbb{Z}_{n}}\rightarrow\left\{ 0,1\right\} 
\]
 defined such that 
\[
F_{T}\left(f\right)=\begin{cases}
\begin{array}{cc}
1 & \text{ if }G_{f}\,\text{is a SUDEC}\\
0 & \text{otherwise}
\end{array} & \text{for all}\:f\in\mathbb{Z}_{n}^{\mathbb{Z}_{n}}\end{cases}
\]
admits an efficient PDP relaxation with exponent parameter $2$.
\end{thm}

\begin{proof}
The proof follows from Tutte's skew symmetric matrix construction
\cite{Tutte1947}. Using the optimal implicit description for the
determinant in Eq. (\ref{Determinant}), the desired relaxed PDP is
\[
\left(\left(\frac{\partial^{n}}{\underset{i\in\mathbb{Z}_{n}}{\prod}\partial a_{i,f\left(i\right)}}\right)\left(\det\left(\mathbf{A}-\mathbf{A}^{\top}\right)\text{mod}\left\{ \begin{array}{c}
a_{k,i}\,a_{k,j},\\
\text{s.t.}\\
\begin{array}{c}
0\le i,\,j,\,k<n\end{array}
\end{array}\right\} \right)\right)^{2}
\]
\end{proof}
Examples discussed thus far above were either PDPs of PDP relaxations.
We describe here an optimal PDE relaxation which tests for a fixed
$g\in\mathbb{Z}_{n}^{\mathbb{Z}_{n}}$, wether or not the input function
$f$ lies in the left $g$--coset of $\mathbb{Z}_{n}^{\mathbb{Z}_{n}}$.
\begin{thm}
For an arbitrary $g\in\mathbb{Z}_{n}^{\mathbb{Z}_{n}}$ let 
\[
T_{g}=\left\{ f\in\mathbb{Z}_{n}^{\mathbb{Z}_{n}}\,:\,\exists\,h\in\mathbb{Z}_{n}^{\mathbb{Z}_{n}},\:\text{ s.t. }\:f=gh\right\} 
\]
and then the Boolean function
\[
F_{T_{g}}\left(f\right)=\begin{cases}
\begin{array}{cc}
1 & \text{ if }f\in R_{g}\\
0 & \text{otherwise}
\end{array} & \text{for all}\:f\in\mathbb{Z}_{n}^{\mathbb{Z}_{n}}\end{cases}
\]
admits an optimal PDE relaxation with exponent parameter 1.
\end{thm}

\begin{proof}
The proof of the claim follows from the identity 
\[
\prod_{i\in\mathbb{Z}_{n}}\left(\sum_{j\in\mathbb{Z}_{n}}a_{i,g\left(j\right)}\right)=\sum_{f\in T_{g}}\left|\left\{ h\in\mathbb{Z}_{n}^{\mathbb{Z}_{n}}:f=gh\right\} \right|\prod_{i\in\mathbb{Z}_{n}}\mathbf{A}\left[i,f\left(i\right)\right]
\]
The desired PDE relaxation is thus given by 
\[
F_{T_{g}}\left(f\right)=\left(\frac{\partial^{n}}{\underset{i\in\mathbb{Z}_{n}}{\prod}\partial a_{i,f\left(i\right)}}\right)\prod_{u\in\mathbb{Z}_{n}}\sum_{v\in\mathbb{Z}_{n}}a_{u,g\left(v\right)}
\]
\end{proof}
\begin{thm}
Let T denotes the largest subset of $\mathbb{Z}_{n}^{\mathbb{Z}_{n}}$
whose graphs are connected (i.e. unicyclic)
\[
F_{T}\,:\mathbb{Z}_{n}^{\mathbb{Z}_{n}}\rightarrow\left\{ 0,1\right\} 
\]
 defined such that 
\[
F_{T}\left(f\right)=\begin{cases}
\begin{array}{cc}
1 & \text{ if }G_{f}\,\text{ is unicyclic}\\
0 & \text{otherwise}
\end{array} & \text{for all}\:f\in\mathbb{Z}_{n}^{\mathbb{Z}_{n}}\end{cases}
\]
admits an efficient relaxed PDP with exponent parameter $1$.
\end{thm}

\begin{proof}
The proof of the claim follows from Tutte's Directed Matrix Theorem
\cite{ZEILBERGER198561,Tutte1948}, from which we get
\[
Q_{T}\left(\mathbf{A}\right)=\left(\sum_{0\le i,j<n}a_{i,j}\right)\det\left\{ \left(\text{diag}\left(\left(\mathbf{A}+\mathbf{A}^{\top}\right)\mathbf{1}_{n\times1}\right)-\left(\mathbf{A}+\mathbf{A}^{\top}\right)\right)\left[:n-1,:n-1\right]\right\} 
\]
\[
\equiv\left(\sum_{f\in T}\left|f^{\left(n-1\right)}\left(\mathbb{Z}_{n}\right)\right|\prod_{i\in\mathbb{Z}_{n}}a_{i,f\left(i\right)}\right)\mod\left\{ \begin{array}{c}
a_{k,i}\,a_{k,j}\\
\text{s.t.}\\
\begin{array}{c}
0\le i,j,k<n\end{array}
\end{array}\right\} .
\]
Using the optimal implicit description for the determinant in Eq.
(\ref{Determinant}), the desired relaxed PDP is
\[
\left(\frac{\partial^{n}}{\underset{i\in\mathbb{Z}_{n}}{\prod}\partial a_{i,f\left(i\right)}}\right)\left(Q_{T}\left(\mathbf{A}\right)\text{mod}\left\{ \begin{array}{c}
a_{k,i}\,a_{k,j}\\
\text{s.t.}\\
\begin{array}{c}
0\le i,j,k<n\end{array}
\end{array}\right\} \right)
\]
\end{proof}
There are of course natural examples of Boolean functions over $\mathbb{Z}_{n}^{\mathbb{Z}_{n}}$
which expectedly admit no efficient PDPs relaxation. 
\begin{conjecture}
The Boolean function 
\[
F_{\text{comp}}\,:\,\mathbb{Z}_{n}^{\mathbb{Z}_{n}}\rightarrow\mathbb{N}
\]
defined such that 
\[
F_{\text{comp}}\left(f\right)=\begin{cases}
\begin{array}{cc}
1 & \text{ if there exist }g\in\mathbb{Z}_{n}^{\mathbb{Z}_{n}}\:\text{such that }f=g\circ g\\
\\
0 & \text{otherwise}
\end{array}\end{cases}
\]
admits no efficient PDP relaxations.
\end{conjecture}

We may broaden slightly the computational model to include arithmetic
circuits whose gates are restricted to operations
\[
\left\{ \text{add}\left(\centerdot,\centerdot\right),\,\text{mul}\left(\centerdot,\centerdot\right),\,\text{exp}\left(\centerdot,\centerdot\right),\,\frac{\partial\centerdot}{\partial\centerdot},\,\text{log}_{\centerdot}\left(\centerdot\right),\,\text{mod}\left(\centerdot,\centerdot\right)\right\} .
\]
The gates above respectively correspond to addition, multiplication,
exponentiation, partial differentiation, logarithm and modular gates.
For simplicity each gate has fan-in equal to two. A partial differentiation
gate outputs the partial derivative of its left input with respect
to its right input. Whereas the output of addition, multiplication
gates are single-valued, the output of other gates nay be multivalued.
For instance, a logarithm gate outputs the multivalued logarithm of
the logarithm of its right input taken with respect to the logarithmic
basis specified by its left input. Modular gates output the remainder
the Euclidean division. We conclude this section by describing small
circuits in this broader computational model for construction akin
to PDPs expressing matrix inversion.
\begin{thm}
In the proposed model of computation there are constructions akin
to PDPs for expressing matrix inversion modulo Cauchy's algebraic
relations and specified via a Chow--rank $1$ polynomial
\end{thm}

\begin{proof}
The proof follows from the optimal expression of the determinant described
in Eq (\ref{Determinant}) and the well known identity
\[
\mathbf{A}^{-1}=\nabla_{\mathbf{A}^{\top}}\ln\left(\det\mathbf{A}\right).
\]
where 
\[
\left(\nabla_{\mathbf{A}^{\top}}F\left(\mathbf{A}\right)\right)\left[i,j\right]=\frac{\partial\,F\left(\mathbf{A}\right)}{\partial\,a_{j,i}},\ \forall\:0\le i,j<n.
\]
It follows that the desired PDP is given by 
\[
\mathbf{A}^{-1}=\nabla_{\mathbf{A}^{\top}}\ln\left\{ \left(\prod_{0\le i<n}a_{i,0}\right)\left(\prod_{0\le i<j<n}\left(a_{j,1}-a_{i,1}\right)\right)\text{mod}\left\{ \begin{array}{c}
a_{i,0}\left(a_{i,1}\right)^{j}-a_{i,j}\\
\text{s.t.}\\
\begin{array}{c}
i\in\mathbb{Z}_{n}\\
0<j<n
\end{array}
\end{array}\right\} \right\} 
\]
\end{proof}
Note that optimal PDP like constructions for inverting matrices yield
asymptotically optimal PDP like construction for multiplying matrices
via the well known reduction identity
\[
\left(\begin{array}{ccc}
\mathbf{I}_{n} & \mathbf{X} & \mathbf{0}_{n\times n}\\
\mathbf{0}_{n\times n} & \mathbf{I}_{n} & \mathbf{Y}\\
\mathbf{0}_{n\times n} & \mathbf{0}_{n\times n} & \mathbf{I}_{n}
\end{array}\right)^{-1}=\left(\begin{array}{ccc}
\mathbf{I}_{n} & -\mathbf{X} & \mathbf{X}\mathbf{Y}\\
\mathbf{0}_{n\times n} & \mathbf{I}_{n} & -\mathbf{Y}\\
\mathbf{0}_{n\times n} & \mathbf{0}_{n\times n} & \mathbf{I}_{n}
\end{array}\right).
\]

\section{Orbital bound for graph isomorphism and sub-isomorphism instances
via group actions.}

We introduce conjugacy class variants of Boolean functions $F_{\subseteq S}$,
$F_{\supseteq S}$ and $F_{=S}$ as Boolean functions defined with
respect to some given graph $G$ such that 
\[
F_{\underset{\sim}{\subset}G}\left(\mathbf{A}_{H}\right)=\begin{cases}
\begin{array}{cc}
1 & \text{ if }H\underset{\sim}{\subset}G\\
0 & \text{otherwise}
\end{array}, & F_{\underset{\sim}{\supset}G}\left(\mathbf{A}_{H}\right)\end{cases}=\begin{cases}
\begin{array}{cc}
1 & \text{ if }H\underset{\sim}{\supset}G\\
0 & \text{otherwise}
\end{array},\end{cases}
\]
\[
\text{and}
\]
\[
F_{\simeq G}\left(\mathbf{A}_{H}\right)=\begin{cases}
\begin{array}{cc}
1 & \text{ if }H\simeq G\\
0 & \text{otherwise}
\end{array},\end{cases}
\]
where $\mathbf{A}_{H}\in\left\{ 0,1\right\} ^{n\times n}$ denotes
the adjacency matrix of the $n$-vertex graph $H$. Let $\mathcal{O}_{\mathbf{Z}}$
denote the $n\times n$ orbital matrix whose entries (are monomials
in entries of a symbolic $n\times n\times\left(n!\right)$ hypermatrix
$\mathbf{Z}$) depict edge orbits induced by the action of the symmetric
group on the vertex set
\[
\mathcal{O}_{\mathbf{Z}}\left[i,j\right]=\prod_{\sigma\in\text{S}_{n}}\mathbf{Z}\left[\sigma\left(i\right),\,\sigma\left(j\right),\,\text{lex}_{\text{S}_{n}}\left(\sigma\right)\right],\quad\forall\,\left(i,j\right)\in\mathbb{Z}_{n}\times\mathbb{Z}_{n},
\]
\[
\text{where}
\]
\[
\text{lex}_{\text{S}_{n}}\left(\sigma\right)=\sum_{k\in\mathbb{Z}_{n}}\left(n-1-k\right)!\,\left|\left\{ \sigma\left(i\right)>\sigma\left(k\right):0\le i<k<n\right\} \right|.
\]
A lower bounds on the number of terms per factor in an optimal PDE/PDP
follows from the prime factorization of the number of non vanishing
terms occurring in the expanded form of multilinear polynomials used
to specify a PDE. For instance, consider the Boolean functions $F_{\simeq G}$
where $G$ is an arbitrary rigid $n$-vertex graph. Then PDPs for
$F_{\simeq G}$ are of the form
\[
F_{\simeq G}\left(\mathbf{A}_{H}\right)=\left(\left.\prod_{\left(i,j\right)\in\mathbb{Z}_{n}\times\mathbb{Z}_{n}}\left(\frac{\partial}{\partial a_{i,j}}\right)^{\mathbf{A}_{H}\left[i,j\right]}P_{\simeq G}\left(\mathbf{A}\right)\right\rfloor _{\mathbf{A}=\mathbf{0}_{n\times n}}\right)^{m},
\]
where 
\[
P_{\simeq G}\left(\mathbf{A}\right)\in\left\{ \sum_{\sigma\in\nicefrac{\text{S}_{n}}{\text{Aut}\left(G\right)}}\omega_{\sigma G\sigma^{-1}}\,\prod_{\left(i,j\right)\in\mathbb{Z}_{n}\times\mathbb{Z}_{n}}a_{ij}^{\mathbf{A}_{G}\left[\sigma\left(i\right),\sigma\left(i\right)\right]}\,:\,\left(\omega_{\sigma G\sigma^{-1}}\right)^{m}=1\right\} ,
\]
Let the prime factorization of the number of non-vanishing terms in
the expanded form of $P_{\simeq G}$ be given by
\[
\left|\nicefrac{\text{S}_{n}}{\text{Aut}\left(G\right)}\right|=n!=\prod_{p\in\mathbb{P}}p^{\alpha_{p}},
\]
where $\mathbb{P}\subset\mathbb{N}$ denotes the set of all primes.
Given that $G$ is rigid we know that
\[
\alpha_{p}=\sum_{j\ge1}\left\lfloor \frac{n}{p^{j}}\right\rfloor ,\ \forall\,p\in\mathbb{P}.
\]
 The smallest depth--3 $\sum\prod\sum$ formula expressing a multilinear
polynomial whose expanded form has $\underset{p\in\mathbb{P}}{\prod}p^{\alpha_{p}}$
non vanishing terms is of size 
\[
1\times\left(\sum_{p\in\mathbb{P}}\alpha_{p}\right)\times\left(1+n^{2}\right).
\]
This lower-bounds is seldom achievable, as seen from the fact that
$P_{\simeq G}$ typically has Chow--Rank $>1$. Using the orbital
argument we derive upper bound on the Chow--rank of polynomial used
to specify PDPs of $F_{\underset{\sim}{\subset}G}$, $F_{\underset{\sim}{\supset}G}$
and $F_{\simeq G}$ prescribed modulo binary algebraic relations.
\begin{thm}
Let $G$ be a given graph on $n$ vertices. Let PDPs for Boolean functions
$F_{\underset{\sim}{\subset}G}$ and $F_{\underset{\sim}{\supset}G}$
be given by
\[
F_{\underset{\sim}{\subset}G}\left(\mathbf{A}_{H}\right)=\left(\left.\frac{\partial^{\left|E\left(H\right)\right|}\left(Q_{\underset{\sim}{\subset}G}\left(\mathbf{A}\right)\text{mod}\left\{ \begin{array}{c}
\left(a_{i,j}\right)^{2}-a_{i,j}\\
0\le i,j<n
\end{array}\right\} \right)}{\underset{\left(i,j\right)\in E\left(H\right)}{\prod}\partial a_{i,j}}\right\rfloor _{\mathbf{A}=\mathbf{0}_{n\times n}}\right)^{m}
\]
and
\[
F_{\underset{\sim}{\supset}G}\left(\mathbf{A}_{H}\right)=\left(\left.\frac{\partial^{\left|E\left(H\right)\right|}\left(Q_{\underset{\sim}{\supset}G}\left(\mathbf{A}\right)\text{mod}\left\{ \begin{array}{c}
\left(a_{i,j}\right)^{2}-a_{i,j}\\
0\le i,j<n
\end{array}\right\} \right)}{\underset{\left(i,j\right)\in E\left(H\right)}{\prod}\partial a_{i,j}}\right\rfloor _{\mathbf{A}=\mathbf{0}_{n\times n}}\right)^{m},
\]
Let optimal Chow--decompositions over $\mathbb{C}$ of $Q_{\underset{\sim}{\subset}G}$
as well as $Q_{\underset{\sim}{\supset}G}$ be given by 
\[
Q_{\underset{\sim}{\subset}G}\left(\mathbf{A}\right)=\sum_{0\le u<\rho}\,\prod_{0\le v<d}\left(\mathbf{B}\left[u,v,0\right]+\sum_{0\le i,j<n}\mathbf{B}\left[u,v,1+n\,i+j\right]\,a_{i,j}\right),
\]
\[
Q_{\underset{\sim}{\supset}G}\left(\mathbf{A}\right)=\sum_{0\le u<\rho^{\prime}}\,\prod_{0\le v<d^{\prime}}\left(\mathbf{B}^{\prime}\left[u,v,0\right]+\sum_{0\le i,j<n}\mathbf{B}^{\prime}\left[u,v,1+n\,i+j\right]\,a_{i,j}\right).
\]
Then bounds on the sizes of hypermatrices $\mathbf{B}\in\mathbb{C}^{\rho\times d\times\left(1+n^{2}\right)}$
and $\mathbf{B}^{\prime}\in\mathbb{C}^{\rho^{\prime}\times d^{\prime}\times\left(1+n^{2}\right)}$
which underly depth--3 arithmetic formulas used to express $Q_{\underset{\sim}{\subset}S}$
and $Q_{\underset{\sim}{\supset}S}$ are such that
\[
\rho\le\left\lceil \frac{\left|\left\{ \mathbf{A}_{H}\in\nicefrac{\left\{ 0,1\right\} ^{n\times n}}{\text{Iso}}:H\underset{\sim}{\subset}G\right\} \right|+\left|\left\{ \mathbf{A}_{H}\in\nicefrac{\left\{ 0,1\right\} ^{n\times n}}{\text{Iso}}:\begin{array}{c}
H\underset{\not\sim}{\subset}G\\
\left|E\left(H\right)\right|\le\left|E\left(G\right)\right|
\end{array}\right\} \right|}{\left(1+n^{2}\right)\,d}\right\rceil 
\]
\[
\text{and}
\]
\[
\rho^{\prime}\le\left\lceil \frac{\left|\left\{ \mathbf{A}_{H}\in\nicefrac{\left\{ 0,1\right\} ^{n\times n}}{\text{Iso}}:H\underset{\sim}{\supset}G\right\} \right|+\left|\left\{ \mathbf{A}_{H}\in\nicefrac{\left\{ 0,1\right\} ^{n\times n}}{\text{Iso}}:\begin{array}{c}
H\underset{\not\sim}{\supset}G\\
\left|E\left(H\right)\right|\ge\left|E\left(G\right)\right|
\end{array}\right\} \right|}{\left(1+n^{2}\right)\,d^{\prime}}\right\rceil .
\]
\end{thm}

\begin{proof}
It suffices to work out the upper bound for the size of $\mathbf{B}\in\mathbb{C}^{\rho\times d\times\left(1+n^{2}\right)}$,
for the argument is identical for $\mathbf{B}^{\prime}\in\mathbb{C}^{\rho^{\prime}\times d^{\prime}\times\left(1+n^{2}\right)}$.
By definition, PDPs with exponent parameter $m=1$ prescribed modulo
Boolean relations are such that 
\[
\sum_{H\underset{\sim}{\subset}G}\prod_{\left(i,j\right)\in E\left(H\right)}a_{ij}\equiv\sum_{0\le u<\rho}\,\prod_{0\le v<d}\left(\mathbf{B}\left[u,v,0\right]+\sum_{0\le i,j<n}\mathbf{B}\left[u,v,1+n\,i+j\right]a_{i,j}\right)\text{ mod}\left\{ \begin{array}{c}
\left(a_{i,j}\right)^{2}-a_{i,j}\\
0\le i,j<n
\end{array}\right\} .
\]
By expanding the expression on the right-hand side and reducing it
modulo prescribed relations we get the equality
\[
\sum_{H\underset{\sim}{\subset}G}\prod_{\left(i,j\right)\in E\left(H\right)}a_{i,j}\equiv\sum_{H\subseteq\mathbb{K}_{n}}K_{H}\left(\mathbf{B}\right)\,\prod_{\left(i,j\right)\in E\left(H\right)}a_{i,j},
\]
the multivariate polynomial $K_{H}\left(\mathbf{B}\right)$ is given
by
\[
K_{H}\left(\mathbf{B}\right)=
\]
\[
\left(\left.\sum_{\left\{ d_{ij}\ge1:\left(i,j\right)\in E\left(H\right)\right\} }\prod_{\left(i,j\right)\in E\left(H\right)}\left(\frac{\partial}{\sqrt[d_{ij}]{d_{ij}!}\,\partial a_{i,j}}\right)^{d_{ij}}\sum_{0\le u<\rho}\prod_{0\le v<d}\left(\mathbf{B}\left[u,v,0\right]+\sum_{0\le i,j<n}\mathbf{B}\left[u,v,1+n\,i+j\right]\,a_{i,j}\right)\right\rfloor _{\mathbf{A}=\mathbf{0}_{n\times n}}\right)
\]
\begin{equation}
,
\end{equation}
substituting entries of $\mathbf{A}$ with the corresponding entries
of the orbital matrix $\mathcal{O}_{\mathbf{Z}}$ yields
\[
\sum_{H\underset{\sim}{\subset}G}\,\prod_{\left(i,j\right)\in E\left(H\right)}\mathcal{O}_{\mathbf{Z}}\left[i,j\right]=\sum_{H\subseteq\mathbb{K}_{n}}K_{H}\left(\mathbf{B}\right)\,\prod_{\left(i,j\right)\in E\left(H\right)}\mathcal{O}_{\mathbf{Z}}\left[i,j\right].
\]
Now we do modulo operations on both sides of this equation. Equating
corresponding coefficients on both sides of the equal sign, which
are coefficients in respective canonical representative congruence
classes
\[
\sum_{H\underset{\sim}{\subset}G}\,\prod_{\left(i,j\right)\in E\left(H\right)}\mathcal{O}_{\mathbf{Z}}\left[i,j\right]\text{ mod}\left\{ \begin{array}{c}
\underset{\left(i,j\right)\in E\left(H\right)}{\prod}\mathbf{Z}\left[i,j,\text{lex}_{\text{S}_{n}}\left(\sigma\right)\right]-\underset{\left(i,j\right)\in E\left(H\right)}{\prod}\mathbf{Y}\left[i,j,\text{lex}\left(H\right)\right]\\
H\underset{\sim}{\subset}G,\,\sigma\in\text{S}_{n}
\end{array}\right\} \equiv\sum_{H\subset G}\,\prod_{K\simeq H}\underset{\left(i,j\right)\in E\left(K\right)}{\prod}\mathbf{Y}\left[i,j,\text{lex}\left(K\right)\right],
\]
where 
\[
\text{lex}\left(H\right)=\sum_{\left(i,j\right)\in E\left(H\right)}2^{n\cdot i+j},\ \text{for all }H\subseteq\mathbb{K}_{n}
\]
and the corresponding coefficients in the canonical representative
of the congruence class
\[
\sum_{H\subseteq\mathbb{K}_{n}}K_{H}\left(\mathbf{B}\right)\,\prod_{\left(i,j\right)\in E\left(H\right)}\mathcal{O}_{\mathbf{Z}}\left[i,j\right]\text{ mod}\left\{ \begin{array}{c}
\underset{\left(i,j\right)\in E\left(H\right)}{\prod}\mathbf{Z}\left[i,j,\text{lex}_{\text{S}_{n}}\left(\sigma\right)\right]-\underset{\left(i,j\right)\in E\left(H\right)}{\prod}\mathbf{Y}\left[i,j,\text{lex}\left(H\right)\right]\\
\sigma\in\text{S}_{n},\,H\subseteq\mathbb{K}_{n}
\end{array}\right\} 
\]

\[
\equiv\sum_{H\subset G}K_{H}\left(\mathbf{B}\right)\,\prod_{K\simeq H}\underset{\left(i,j\right)\in E\left(K\right)}{\prod}\mathbf{Y}\left[i,j,\text{lex}\left(K\right)\right]+\sum_{H\underset{\not\sim}{\subset}G}K_{H}\left(\mathbf{B}\right)\,\prod_{K\simeq H}\underset{\left(i,j\right)\in E\left(K\right)}{\prod}\mathbf{Y}\left[i,j,\text{lex}\left(K\right)\right]
\]
yields a system of 
\[
\left|\left\{ \mathbf{A}_{H}\in\nicefrac{\left\{ 0,1\right\} ^{n\times n}}{\text{Iso}}:H\underset{\sim}{\subset}G\right\} \right|+\left|\left\{ \mathbf{A}_{H}\in\nicefrac{\left\{ 0,1\right\} ^{n\times n}}{\text{Iso}}:\begin{array}{c}
H\underset{\not\sim}{\subset}G\\
\left|E\left(H\right)\right|\le\left|E\left(G\right)\right|
\end{array}\right\} \right|
\]
equations in the $\rho\cdot d\cdot\left(1+n^{2}\right)$ unknown entries
for $\mathbf{B}\in\mathbb{C}^{\rho\times d\times\left(1+n^{2}\right)}$
after mering the same terms on both sides. $\left|\left\{ \mathbf{A}_{H}\in\nicefrac{\left\{ 0,1\right\} ^{n\times n}}{\text{Iso}}:H\underset{\sim}{\subset}G\right\} \right|$
stands for the number of graphs $H$ that are subgraph-isomorphic
to $G$ and the number of terms with a non-zero coefficient in the
canonical representative of the congruence class, while $\left|\left\{ \mathbf{A}_{H}\in\nicefrac{\left\{ 0,1\right\} ^{n\times n}}{\text{Iso}}:\begin{array}{c}
H\underset{\not\sim}{\subset}G\\
\left|E\left(H\right)\right|\le\left|E\left(G\right)\right|
\end{array}\right\} \right|$ stands for the number of graphs that are not subgraph-isomorphic
to $G$, which is also the number of terms with 0 coefiicients in
the canonical representative of the congruence class. Clearly $d\ge\left|E\left(G\right)\right|$
and we know that by eliminating variables via the method of resultants,
the latter system of equations necessarily admits a solution whenever
the number unknowns $\rho\cdot d\cdot\left(1+n^{2}\right)$ matches
or exceeds the number of algebraically independent constraints. We
see that setting 
\[
\rho=\left\lceil \frac{\left|\left\{ \mathbf{A}_{H}\in\nicefrac{\left\{ 0,1\right\} ^{n\times n}}{\text{Iso}}:H\underset{\sim}{\subset}G\right\} \right|+\left|\left\{ \mathbf{A}_{H}\in\nicefrac{\left\{ 0,1\right\} ^{n\times n}}{\text{Iso}}:\begin{array}{c}
H\underset{\not\sim}{\subset}G\\
\left|E\left(H\right)\right|\le\left|E\left(G\right)\right|
\end{array}\right\} \right|}{\left(1+n^{2}\right)\,d}\right\rceil ,
\]
the number of variables matches or exceeds the number of algebraically
independent constraints. It follows from the degree lower bound $d\ge\left|E\left(G\right)\right|$
that we can take 
\[
\rho=\left\lceil \frac{\left|\left\{ \mathbf{A}_{H}\in\nicefrac{\left\{ 0,1\right\} ^{n\times n}}{\text{Iso}}:H\underset{\sim}{\subset}G\right\} \right|+\left|\left\{ \mathbf{A}_{H}\in\nicefrac{\left\{ 0,1\right\} ^{n\times n}}{\text{Iso}}:\begin{array}{c}
H\underset{\not\sim}{\subset}G\\
\left|E\left(H\right)\right|\le\left|E\left(G\right)\right|
\end{array}\right\} \right|}{\left(1+n^{2}\right)\,\left|E\left(G\right)\right|}\right\rceil .
\]
\end{proof}
If we restrict the discussion to PDEs which test whether or not the
graph of the input function is isomorphism to the graph of the given
function $g\in\mathbb{Z}_{n}^{\mathbb{Z}_{n}}$, then the corresponding
Boolean function is of the form 
\[
F_{\simeq G_{g}}\left(f\right)=\begin{cases}
\begin{array}{cc}
1 & \text{ if there exist }\sigma\in\text{S}_{n}\:\text{such that }\sigma f\sigma^{\left(-1\right)}=g\\
\\
0 & \text{otherwise}
\end{array} & \text{for all }f\in\mathbb{Z}_{n}^{\mathbb{Z}_{n}}\end{cases}.
\]
PDEs of $F_{\simeq G_{g}}$ are of the form
\[
F_{\simeq G_{g}}\left(f\right)=\left(\frac{\partial^{n}\,P_{\simeq G_{g}}\left(\mathbf{A}\right)}{\underset{i\in\mathbb{Z}_{n}}{\prod}\partial a_{i,f\left(i\right)}}\right)^{m},\:\text{ for all }\:f\in\mathbb{Z}_{n}^{\mathbb{Z}_{n}},
\]
where
\[
P_{\simeq G_{g}}\left(\mathbf{A}\right)\in\left\{ \sum_{\sigma\in\nicefrac{\text{S}_{n}}{\text{Aut}G_{g}}}\omega_{\sigma}\,\prod_{i\in\mathbb{Z}_{n}}a_{i,\sigma g\sigma^{-1}\left(i\right)}\,:\,\left(\omega_{\sigma}\right)^{m}=1\right\} .
\]
The orbital argument yields PDP specified in term of a polynomial
$Q_{\simeq G_{g}}\left(\mathbf{A}\right)$ subject to 
\[
P_{\simeq G_{g}}\left(\mathbf{A}\right)\equiv\left(Q_{\simeq G_{g}}\left(\mathbf{A}\right)\mod\left\{ \begin{array}{c}
\left(a_{i,j}\right)^{2}-a_{i,j}\\
0\le i,j<n
\end{array}\right\} \right),
\]
$Q_{\simeq G_{g}}\left(\mathbf{A}\right)$ is of Chow-rank 
\[
\rho\le\left\lceil \frac{\kappa^{n}}{\left(n+n^{3}\right)\sqrt{n}}\right\rceil .
\]
for some real number $\kappa>1$.

\section{Orbital hypergraph isomorphism and sub-isomorphism PDPs.}

We describe hyperedges of $k$--uniform $n$-vertex hypergraph $H$
as a fixed subset subset of $E\left(H\right)\subseteq\mathbb{Z}_{n}^{\mathbb{Z}_{k}}$.
The monomial hyperedge list description of $H$ is
\[
\prod_{f\in E\left(H\right)}\mathbf{A}\left[f\left(0\right),\cdots,f\left(k-1\right)\right],
\]
where $\mathbf{A}$ denotes a symbolic side length $n$ hypermatrix
of order $m$ such that
\[
\mathbf{A}\left[i_{0},\cdots,i_{k-1}\right]=a_{i_{0},\cdots,i_{k-1}},\quad0\le i_{0},\cdots,i_{k-1}<n.
\]
At the limit where $k\rightarrow n$, an arbitrary hypergraph $H$
is specified by providing a fixed subset subset of $E\left(H\right)\subseteq\mathbb{Z}_{n}^{\mathbb{Z}_{n}}$.
Their orbit list generating polynomial yields an optimal PDP for the
Boolean function
\[
F_{\simeq H}\left(\mathbf{A}_{H^{\prime}}\right)=\begin{cases}
\begin{array}{cc}
1 & \text{ if }H^{\prime}\simeq H\\
0 & \text{otherwise}
\end{array},\end{cases}
\]
where $\mathbf{A}_{H}\in\left\{ 0,1\right\} ^{n\times\cdots\times n}$
denote the adjacency hypermatrix of $H$. Let $\mathcal{O}_{\mathbf{Z}}$
denote the orbital hypermatrix whose order is $n+1$ and side length
is equal to $n$. Entries of $\mathcal{O}_{\mathbf{Z}}$ depicts hyperedge
orbits induced by the action of the symmetric group on the vertex
set
\[
\mathcal{O}_{\mathbf{Z}}\left[i_{0},\cdots,i_{n-1}\right]=\prod_{\sigma\in\text{S}_{n}}\mathbf{Z}\left[\sigma\left(i_{0}\right),\cdots,\sigma\left(i_{n-1}\right),\text{lex}_{\text{S}_{n}}\left(\sigma\right)\right].
\]
Let 
\[
P_{\simeq H}\left(\mathbf{A}\right)=\prod_{f\in E\left(H\right)}a_{f\left(0\right),\cdots,f\left(k-1\right)},
\]
then the desired PDP is given by 
\[
F_{\simeq H}\left(\mathbf{A}_{H^{\prime}}\right)=\frac{\partial^{\left|E\left(H^{\prime}\right)\right|}}{\underset{f\in E\left(H^{\prime}\right)}{\prod}\partial a_{f\left(0\right),\cdots,f\left(n-1\right)}}
\]
\[
\left(P_{\simeq H}\left(\mathcal{O}_{\mathbf{Z}}\right)\text{mod}\left\{ \begin{array}{c}
\underset{g\in E\left(H\right)}{\prod}\mathbf{Z}\left[\sigma g\left(0\right),\cdots,\sigma g\left(n-1\right),\text{lex}_{\text{S}_{n}}\left(\sigma\right)\right]-\left(\begin{array}{cc}
1 & \frac{\underset{f\in E\left(H\right)}{\prod}a_{\sigma g\left(0\right),\cdots,\sigma g\left(n-1\right)}}{n!}\\
0 & 1
\end{array}\right)\\
\sigma\in\nicefrac{\text{S}_{n}}{\text{Aut}\left(H\right)}
\end{array}\right\} \right)\left[0,1\right].
\]
or alternatively 
\[
F_{\simeq H}\left(\mathbf{A}_{G}\right)=\left(\frac{\partial^{\left|\nicefrac{\text{S}_{n}}{\text{Aut}\left(G\right)}\right|\left|E\left(G\right)\right|}\,\omega_{H}\underset{\begin{array}{c}
g\in E\left(\sigma H\right)\\
\sigma\in\nicefrac{\text{S}_{n}}{\text{Aut}\left(H\right)}
\end{array}}{\prod}y_{g\left(0\right),\cdots,g\left(k-1\right),\text{lex}E\left(\sigma H\right)}}{\underset{\begin{array}{c}
f\in E\left(\sigma G\right)\\
\sigma\in\nicefrac{\text{S}_{n}}{\text{Aut}\left(G\right)}
\end{array}}{\prod}\partial y_{f\left(0\right),\cdots,f\left(n-1\right),\text{lex}E\left(\sigma G\right)}}\right)^{m},
\]
where 
\[
\text{lex}\left(E\left(R\right)\right)=\sum_{f\in E\left(R\right)}2^{\text{lex}_{\mathbb{Z}_{n}^{\mathbb{Z}_{n}}}\left(f\right)}.
\]
The first construction is a valid PDP since we know by Stirling approximation
that 
\[
n!\sim\left(\frac{n}{e}\right)^{n}\sqrt{2\pi n}
\]
is polynomial in the parameter $\left|\mathbb{Z}_{n}^{\mathbb{Z}_{n}}\right|$.
The latter construction describes an optimal PDE. Unfortunately adapting
the constructions above to sub-isomorphism instances does not result
in PDP for the set of algebraic relations needed is no longer polynomial
in the parameter $\left|\mathbb{Z}_{n}^{\mathbb{Z}_{n}}\right|$.
Fortunately PDPs inspire another approach to articulating the subtle
gap in complexity separating isomorphism instances from their sub-isomorphism
counterparts. Typically one considers specific isomorphism or sub-isomorphism
instances specified with two input hypergraphs. In such a setting
one seeks to determine whether or not the specific instance is a YES
instance or a NO instance. This restricted setting is very different
from the PDP constructions that we have described thus far. In PDP
construction that we have described we sought to construct a Boolean
functions which test for isomorphism or sub-isomorphism of a given
graph to any other graph. We see that determining whether or not the
specific instance is a YES instance or a NO instance is an easier
task.
\begin{thm}
Given $m$-uniform hypergraphs $H$ and $G$, the corresponding isomorphism
instance is a YES instance if and only if 
\[
0=\text{Discriminant}_{x}\left(x^{2}-p_{1}\,x+\frac{p_{1}^{2}-p_{2}}{2}\right)=\left(2p_{2}-p_{1}^{2}\right),
\]
where
\[
\left(\prod_{f\in E\left(H\right)}\mathcal{O}_{\mathbf{Z}}\left[f\left(0\right),\cdots,f\left(n-1\right)\right]+\prod_{g\in E\left(G\right)}\mathcal{O}_{\mathbf{Z}}\left[g\left(0\right),\cdots,g\left(n-1\right)\right]\right)
\]
\[
\text{mod}\left\{ \begin{array}{c}
\underset{h\in E\left(R\right)}{\prod}\mathbf{Z}\left[h\left(0\right),\cdots,h\left(n-1\right),\text{lex}_{\text{S}_{n}}\left(\sigma\right)\right]-\left(\sqrt[\left|\text{Aut}\left(R\right)\right|]{\underset{h\in R}{\prod}\mathbf{Y}\left[h\left(0\right),\cdots,h\left(n-1\right),\text{lex}\left(R\right)\right]}\right)^{k}\\
R\in\left(\underset{\sigma\in\nicefrac{\text{S}_{n}}{\text{Aut}\left(H\right)}}{\bigcup}\sigma H\right)\cup\left(\underset{\sigma\in\nicefrac{\text{S}_{n}}{\text{Aut}\left(G\right)}}{\bigcup}\sigma G\right)
\end{array}\right\} 
\]

\[
\equiv\left(\prod_{K\simeq H}\prod_{f\in E\left(K\right)}\mathbf{Y}\left[f\left(0\right),\cdots,f\left(n-1\right),\text{lex}\left(K\right)\right]^{k}+\prod_{K\simeq G}\prod_{f\in E\left(K\right)}\mathbf{Y}\left[f\left(0\right),\cdots,f\left(n-1\right),\text{lex}\left(K\right)\right]^{k}\right)=p_{k}
\]
\end{thm}

\begin{proof}
The key idea here is that if $H\simeq G$ then both $p_{1}$ and $p_{2}$
are monomials and $p_{2}=\frac{\left(p_{1}\right)^{2}}{2}$, while
if $H$ is not isomorphic to $G$ then $p_{1},p_{2}$ each have two
terms. For a specific isomorphism instance the claim immediately follows
from the observation that the canonical representative of 
\[
p_{k}\equiv\left(\prod_{f\in E\left(H\right)}\mathcal{O}_{\mathbf{Z}}\left[f\left(0\right),\cdots,f\left(n-1\right)\right]+\prod_{f\in E\left(G\right)}\mathcal{O}_{\mathbf{Z}}\left[f\left(0\right),\cdots,f\left(n-1\right)\right]\right)
\]
\[
\text{mod}\left\{ \begin{array}{c}
\underset{f\in E\left(R\right)}{\prod}\mathbf{Z}\left[f\left(0\right),\cdots,f\left(n-1\right),\text{lex}_{\text{S}_{n}}\left(\sigma\right)\right]-\left(\sqrt[\left|\text{Aut}\left(R\right)\right|]{\underset{h\in R}{\prod}\mathbf{Y}\left[f\left(0\right),\cdots,f\left(n-1\right),\text{lex}\left(R\right)\right]}\right)^{k}\\
R\in\left(\underset{\sigma\in\nicefrac{\text{S}_{n}}{\text{Aut}\left(H\right)}}{\bigcup}\sigma H\right)\cup\left(\underset{\sigma\in\nicefrac{\text{S}_{n}}{\text{Aut}\left(G\right)}}{\bigcup}\sigma G\right)
\end{array}\right\} 
\]
equals 
\[
p_{k}=2\prod_{\sigma\in\nicefrac{\text{S}_{n}}{\text{Aut}\left(H\right)}}\left(\underset{f\in E\left(\sigma H\right)}{\prod}\mathbf{Y}\left[f\left(0\right),\cdots,f\left(n-1\right),\,\text{lex}\left(\sigma H\right)\right]\right)^{k}
\]
for a YES instance and equals 
\[
\prod_{\sigma\in\nicefrac{\text{S}_{n}}{\text{Aut}\left(H\right)}}\left(\underset{f\in E\left(\sigma H\right)}{\prod}\mathbf{Y}\left[f\left(0\right),\cdots,f\left(n-1\right),\,\text{lex}\left(\sigma H\right)\right]\right)^{k}+\prod_{\sigma\in\nicefrac{\text{S}_{n}}{\text{Aut}\left(G\right)}}\left(\underset{g\in E\left(\sigma G\right)}{\prod}\mathbf{Y}\left[g\left(0\right),\cdots,g\left(n-1\right),\,\text{lex}\left(\sigma G\right)\right]\right)^{k}
\]
for a NO instance. The discriminant equation therefore follows from
Newton--Girard formulas.
\end{proof}
Note that the Chow--rank of the polynomial construction is at most
2 both before and after performing the reduction modulo prescribed
algebraic relations. Also note that the number of variables appearing
in the PDP can be reduced by considering an orbital matrix whose entries,
instead, depict the action of cosets of some canonically chosen set
of generators for the automorphism groups of hypergraphs $H$ and
$G$ respectively. This is best illustrated with isomorphism instances
defined over functional directed graphs. Let $f,g\in\mathbb{Z}_{n}^{\mathbb{Z}_{n}}$
and consider two distinct orbital matrices 
\[
\mathcal{O}_{\mathbf{Z}}\left[i,j\right]=\prod_{\begin{array}{c}
\sigma\in\nicefrac{\text{S}_{n}}{\text{Aut}\left(G_{f}\right)}\\
\gamma\in\text{Generator of }\text{Aut}\left(G_{f}\right)
\end{array}}\mathbf{Z}\left[\sigma\gamma\left(i\right),\sigma\gamma\left(j\right),\text{lex}_{\text{S}_{n}}\left(\sigma\gamma\right)\right].
\]
\[
\mathcal{O}_{\mathbf{Z}}^{\prime}\left[i,j\right]=\prod_{\begin{array}{c}
\sigma\in\nicefrac{\text{S}_{n}}{\text{Aut}\left(G_{g}\right)}\\
\gamma\in\text{Generator of }\text{Aut}\left(G_{g}\right)
\end{array}}\mathbf{Z}\left[\sigma\gamma\left(i\right),\sigma\gamma\left(j\right),\text{lex}_{\text{S}_{n}}\left(\sigma\gamma\right)\right].
\]
The expression of interest:
\[
p_{k}\equiv\left(\prod_{i\in\mathbb{Z}_{n}}\mathcal{O}_{\mathbf{Z}}\left[i,f\left(i\right)\right]+\prod_{j\in\mathbb{Z}_{n}}\mathcal{O}_{\mathbf{Z}}^{\prime}\left[j,g\left(j\right)\right]\right)
\]
\[
\text{mod}\left\{ \begin{array}{c}
\underset{i\in\mathbb{Z}_{n}}{\prod}\mathbf{Z}\left[i,h\left(i\right),\text{lex}_{\text{S}_{n}}\left(\sigma\right)\right]-\left(\sqrt[\left|\text{Aut}\left(G_{h}\right)\right|]{\underset{j\in\mathbb{Z}_{n}}{\prod}\mathbf{Y}\left[j,h\left(j\right),\text{lex}_{\mathbb{Z}_{n}^{\mathbb{Z}_{n}}}\left(h\right)\right]}\right)^{k}\\
\sigma\in\text{S}_{n},\ h\in\left(\underset{\begin{array}{c}
\sigma\in\nicefrac{\text{S}_{n}}{\text{Aut}\left(G_{f}\right)}\\
\gamma\in\text{Generator of }\text{Aut}\left(G_{f}\right)
\end{array}}{\bigcup}\sigma\gamma f\left(\sigma\gamma\right)^{-1}\right)\cup\left(\underset{\begin{array}{c}
\sigma\in\nicefrac{\text{S}_{n}}{\text{Aut}\left(G_{g}\right)}\\
\gamma\in\text{Generator of }\text{Aut}\left(G_{g}\right)
\end{array}}{\bigcup}\sigma\gamma g\left(\sigma\gamma\right)^{-1}\right)
\end{array}\right\} .
\]
The isomorphism instance is thus a YES instance if and only if 
\[
0=\text{Discriminant}_{x}\left(x^{2}-p_{1}x+\frac{p_{1}^{2}-p_{2}}{2}\right)=\left(2p_{2}-p_{1}^{2}\right).
\]
The number of substitutions prescribed by the search and replacement
procedure reduces in this setting to
\[
\left|\nicefrac{\text{S}_{n}}{\text{Aut}\left(G_{f}\right)}\right|\left|\text{Generator of }\text{Aut}\left(G_{f}\right)\right|+\left|\nicefrac{\text{S}_{n}}{\text{Aut}\left(G_{g}\right)}\right|\left|\text{Generator of }\text{Aut}\left(G_{g}\right)\right|.
\]
We now contrast the analysis above to sub-isomorphism instances. In
order to check sub-isomorphism, we may construct two polynomials 

\[
d\left(x\right)=\prod_{K\underset{\sim}{\subset}G}(x+\underset{f\in E(K)}{\prod}\mathbf{Y}\left[f\left(0\right),\cdots,f\left(n-1\right),\text{lex}\left(K\right)\right])
\]

\[
g\left(x\right)=\prod_{K\underset{\sim}{\subset}G\ or\ K\simeq H}(x+\underset{f\in E(K)}{\prod}\mathbf{Y}\left[f\left(0\right),\cdots,f\left(n-1\right),\text{lex}\left(K\right)\right])
\]
If $H$ is sub-isomorphic to $G$, then $g(x)$ divides $\left(d\left(x\right)\right)^{2}$.
For we see that every monomial in the entries of $\mathbf{Y}$ occuring
in factors of $g(x)$ appears at most twice in $\left(d\left(x\right)\right)^{2}$.
Otherwise, some monomial in the entries of $\mathbf{Y}$ occuring
in factors of $g\left(x\right)$ never occurs in a factor of $\left(d\left(x\right)\right)^{2}$.
Meanwhile, an orbital construction yields $d\left(x\right)$ and $g\left(x\right)$
by reducing modulo relations introduced in the previous theorem. Using
the foundamental theorem of symmetric polynomials we devise an explicit
expression for the the expanded form of $d\left(x\right)$ and $g\left(x\right)$
\begin{thm}
Given $m$-uniform hypergraphs $H$ and $G$, the corresponding sub-isomorphism
instance is a YES instance if and only if the polynomial 
\[
g\left(x\right)=\left(x^{1+2^{\left|E\left(G\right)\right|}}+\sum_{0<k\le1+2^{\left|E\left(G\right)\right|}}\left(-1\right)^{k}x^{1+2^{\left|E\left(G\right)\right|}-k}\sum_{{m_{1}+2m_{2}+\cdots+km_{k}=1+2^{\left|E\left(G\right)\right|}\atop m_{1}\ge0,\ldots,m_{k}\ge0}}\prod_{0<i\le1+2^{\left|E\left(G\right)\right|}}\frac{\left(-q_{k}\right)^{m_{i}}}{m_{i}!\,i^{m_{i}}}\right)
\]
divides the polynomial 
\[
\left(d\left(x\right)\right)^{2}=\left(x^{2^{\left|E\left(G\right)\right|}}+\sum_{0<k\le2^{\left|E\left(G\right)\right|}}\left(-1\right)^{k}x^{2^{\left|E\left(G\right)\right|}-k}\sum_{{m_{1}+2m_{2}+\cdots+km_{k}=1+2^{\left|E\left(G\right)\right|}\atop m_{1}\ge0,\ldots,m_{k}\ge0}}\prod_{0<i\le\eta}\frac{\left(-p_{k}\right)^{m_{i}}}{m_{i}!\,i^{m_{i}}}\right)^{2},
\]
where 
\[
\prod_{f\in E\left(G\right)}\left(1+\mathcal{O}_{\mathbf{Z}}\left[f\left(0\right),\cdots,f\left(n-1\right)\right]\right)
\]
\[
\text{mod}\left\{ \begin{array}{c}
\underset{f\in E\left(R\right)}{\prod}\mathbf{Z}\left[f\left(0\right),\cdots,f\left(n-1\right),\text{lex}_{\text{S}_{n}}\left(\sigma\right)\right]-\left(\sqrt[\left|\text{Aut}\left(R\right)\right|]{\underset{f\in E(R)}{\prod}\mathbf{Y}\left[f\left(0\right),\cdots,f\left(n-1\right),\text{lex}\left(R\right)\right]}\right)^{k}\\
R\subseteq\mathbb{Z}_{n}^{\mathbb{Z}_{n}}
\end{array}\right\} 
\]

\[
\equiv\sum_{K\underset{\sim}{\subset}G}\underset{f\in E(K)}{\prod}\mathbf{Y}\left[f\left(0\right),\cdots,f\left(n-1\right),\text{lex}\left(R\right)\right]^{k}=p_{k}
\]

\[
\left(\prod_{f\in E\left(H\right)}\mathcal{O}_{\mathbf{Z}}\left[f\left(0\right),\cdots,f\left(n-1\right)\right]+\prod_{f\in E\left(G\right)}\left(1+\mathcal{O}_{\mathbf{Z}}\left[f\left(0\right),\cdots,f\left(n-1\right)\right]\right)\right)
\]
\[
\text{mod}\left\{ \begin{array}{c}
\underset{f\in E\left(R\right)}{\prod}\mathbf{Z}\left[f\left(0\right),\cdots,f\left(n-1\right),\text{lex}_{\text{S}_{n}}\left(\sigma\right)\right]-\left(\sqrt[\left|\text{Aut}\left(R\right)\right|]{\underset{f\in E(R)}{\prod}\mathbf{Y}\left[f\left(0\right),\cdots,f\left(n-1\right),\text{lex}\left(R\right)\right]}\right)^{k}\\
R\subseteq\mathbb{Z}_{n}^{\mathbb{Z}_{n}}
\end{array}\right\} 
\]

\[
\equiv\left(\sum_{K\underset{\sim}{\subset}G}\underset{f\in E(K)}{\prod}\mathbf{Y}\left[f\left(0\right),\cdots,f\left(n-1\right),\text{lex}\left(R\right)\right]^{k}+\sum_{K\simeq H}\underset{f\in E(K)}{\prod}\mathbf{Y}\left[f\left(0\right),\cdots,f\left(n-1\right),\text{lex}\left(R\right)\right]^{k}\right)=q_{k}
\]
\end{thm}

\begin{proof}
The claim follows from the observation that for a YES sub-isomorphism
instance the monomial support of the canonical representative of the
congruence class 
\[
\prod_{f\in E\left(G\right)}\left(1+\mathcal{O}_{\mathbf{Z}}\left[f\left(0\right),\cdots,f\left(n-1\right)\right]\right)
\]
\[
\text{mod}\left\{ \begin{array}{c}
\underset{f\in E\left(R\right)}{\prod}\mathbf{Z}\left[f\left(0\right),\cdots,f\left(n-1\right),\text{lex}_{\text{S}_{n}}\left(\sigma\right)\right]-\sqrt[\left|\text{Aut}\left(R\right)\right|]{\underset{f\in E(R)}{\prod}\mathbf{Y}\left[f\left(0\right),\cdots,f\left(n-1\right),\text{lex}\left(R\right)\right]}\\
R\subseteq\mathbb{Z}_{n}^{\mathbb{Z}_{n}}
\end{array}\right\} ,
\]
matches the monomial support of the canonical representative of the
congruence class
\[
\prod_{f\in E\left(H\right)}\mathcal{O}_{\mathbf{Z}}\left[f\left(0\right),\cdots,f\left(n-1\right)\right]+\prod_{f\in E\left(G\right)}\left(1+\mathcal{O}_{\mathbf{Z}}\left[f\left(0\right),\cdots,f\left(n-1\right)\right]\right)
\]
\[
\text{mod}\left\{ \begin{array}{c}
\underset{f\in E\left(R\right)}{\prod}\mathbf{Z}\left[f\left(0\right),\cdots,f\left(n-1\right),\text{lex}_{\text{S}_{n}}\left(\sigma\right)\right]-\sqrt[\left|\text{Aut}\left(R\right)\right|]{\underset{f\in E(R)}{\prod}\mathbf{Y}\left[f\left(0\right),\cdots,f\left(n-1\right),\text{lex}\left(R\right)\right]}\\
R\subseteq\mathbb{Z}_{n}^{\mathbb{Z}_{n}}
\end{array}\right\} ,
\]
Furthermore, for such an instance the two polynomial differ by exactly
one of the non-vanishing integer coefficient being incremented by
one. Using the Newton-Girard formulas we derive the polynomial division
property.
\end{proof}
We see that the Chow--rank of polynomials start out having Chow--Rank
at most 2 prior to the reduction and increases to at least $2^{n}$
and at most $2^{n}+1$ after performing the reduction. The PDP construction
therefore exhibits an unconditional exponential separation between
isomorphism and sub-isomorphism instances.

\bibliographystyle{amsalpha}
\bibliography{On_Partial_Differential_Encodings_with_Application_to_Boolean_Circuits}

\newcommand{\etalchar}[1]{$^{#1}$}
\providecommand{\bysame}{\leavevmode\hbox to3em{\hrulefill}\thinspace}
\providecommand{\MR}{\relax\ifhmode\unskip\space\fi MR }
\providecommand{\MRhref}[2]{%
  \href{http://www.ams.org/mathscinet-getitem?mr=#1}{#2}
}
\providecommand{\href}[2]{#2}
\begin{thebibliography}{GKKS16}

\bibitem[Aar16]{Aar16}
Scott Aaronson, \emph{P=?np}, pp.~1--122, springer international publishing,
  cham, 2016.

\bibitem[AV08]{4690941}
M.~{Agrawal} and V.~{Vinay}, \emph{Arithmetic circuits: A chasm at depth four},
  2008 49th Annual IEEE Symposium on Foundations of Computer Science, 2008,
  pp.~67--75.

\bibitem[Boo54]{Bool54}
George Boole, \emph{An investigation of the laws of thought: On which are
  founded the mathematical theories of logic and probabilities}, Cambridge
  Library Collection - Mathematics, Cambridge University Press, 1854.

\bibitem[BP20]{br2020algorithmic}
Cornelius Brand and Kevin Pratt, \emph{An algorithmic method of partial
  derivatives}, 2020.

\bibitem[BS83]{BS83}
Walter Baur and Volker Strassen, \emph{The complexity of partial derivatives},
  Theoretical Computer Science \textbf{22} (1983), no.~3, 317 -- 330.

\bibitem[Cay89]{Cay45}
Arthur Cayley, \emph{On the theory of linear transformations}, Cambridge
  Library Collection - Mathematics, vol.~1, pp.~80--94, Cambridge University
  Press, 1889.

\bibitem[CKW11]{TCS2011}
Xi~Chen, Neeraj Kayal, and Avi Wigderson, \emph{Partial derivatives in
  arithmetic complexity and beyond}, Foundations and Trends in Theoretical
  Computer Science \textbf{6} (2011), no.~1--2, 1--138.

\bibitem[For01]{Kolmogorovcomplexity}
Lance Fortnow, \emph{Kolmogorov complexity}, pp.~73 -- 86, De Gruyter, Berlin,
  Boston, 2001.

\bibitem[GIM{\etalchar{+}}19]{Garg19}
Ankit Garg, Christian Ikenmeyer, Visu Makam, Rafael Oliveira, Michael Walter,
  and Avi Wigderson, \emph{Search problems in algebraic complexity, gct, and
  hardness of generator for invariant rings}, 2019.

\bibitem[GKKS16]{doi:10.1137/140957123}
Ankit Gupta, Pritish Kamath, Neeraj Kayal, and Ramprasad Saptharishi,
  \emph{Arithmetic circuits: A chasm at depth 3}, SIAM Journal on Computing
  \textbf{45} (2016), no.~3, 1064--1079.

\bibitem[GMQ16]{Gro16}
Joshua~A. Grochow, Ketan~D. Mulmuley, and Youming Qiao, \emph{{Boundaries of VP
  and VNP}}, 43rd International Colloquium on Automata, Languages, and
  Programming (ICALP 2016) (Dagstuhl, Germany) (Ioannis Chatzigiannakis,
  Michael Mitzenmacher, Yuval Rabani, and Davide Sangiorgi, eds.), Leibniz
  International Proceedings in Informatics (LIPIcs), vol.~55, Schloss
  Dagstuhl--Leibniz-Zentrum fuer Informatik, 2016, pp.~34:1--34:14.

\bibitem[Gro15]{JGrochow2015}
Joshua~A. Grochow, \emph{Unifying known lower bounds via geometric complexity
  theory}, computational complexity \textbf{24} (2015), no.~2, 393--475.

\bibitem[Gro20]{Gro20}
Joshua~A. Grochow, \emph{Complexity in ideals of polynomials: Questions on
  algebraic complexity of circuits and proofs}, Bull. {EATCS} \textbf{130}
  (2020).

\bibitem[Hya79]{doi:10.1137/0208010}
Laurent Hyafil, \emph{On the parallel evaluation of multivariate polynomials},
  SIAM Journal on Computing \textbf{8} (1979), no.~2, 120--123.

\bibitem[Lan17]{landsberg_2017}
J.~M. Landsberg, \emph{Geometry and complexity theory}, Cambridge Studies in
  Advanced Mathematics, Cambridge University Press, 2017.

\bibitem[NW96]{WN96}
Noam Nisan and Avi Wigderson, \emph{Lower bounds on arithmetic circuits via
  partial derivatives}, computational complexity \textbf{6} (1996), no.~3,
  217--234.

\bibitem[Pol37]{Pol37}
G.~Polya, \emph{Kombinatorische anzahlbestimmungen fur gruppen, graphen und
  chemische verbindungen}, Acta Math. \textbf{68} (1937), 145--254.

\bibitem[Pol40]{Pol40}
G.~Polya, \emph{Sur les types des propositions compos{\'e}es}, The Journal of
  Symbolic Logic \textbf{5} (1940), no.~3, 98--103.

\bibitem[Raz13]{10.1145/2535928}
Ran Raz, \emph{Tensor-rank and lower bounds for arithmetic formulas}, J. ACM
  \textbf{60} (2013), no.~6.

\bibitem[Red27]{Red27}
J.~Howard Redfield, \emph{The theory of group-reduced distributions}, American
  Journal of Mathematics \textbf{49} (1927), no.~3, 433--455.

\bibitem[{Sha}49]{Sha49}
C.~E. {Shannon}, \emph{The synthesis of two-terminal switching circuits}, The
  Bell System Technical Journal \textbf{28} (1949), no.~1, 59--98.

\bibitem[SY10]{TCS-039}
Amir Shpilka and Amir Yehudayoff, \emph{Arithmetic circuits: A survey of recent
  results and open questions}, Foundations and Trends in Theoretical Computer
  Science \textbf{5} (2010), no.~3--4, 207--388.

\bibitem[Syl52]{Sylvester1852}
James~Joseph Sylvester, \emph{On the principles of the calculus of forms},
  Cambridge and Dublin Mathematical Journal \textbf{7} (1852), 57--92.

\bibitem[Tur36]{Tur36}
Alan~M. Turing, \emph{On computable numbers, with an application to the
  {E}ntscheidungsproblem}, Proceedings of the London Mathematical Society
  \textbf{2} (1936), no.~42, 230--265.

\bibitem[Tut47]{Tutte1947}
W.~T. Tutte, \emph{The factorization of linear graphs}, Journal of the London
  Mathematical Society \textbf{s1-22} (1947), no.~2, 107--111.

\bibitem[Tut48]{Tutte1948}
\bysame, \emph{The dissection of equilateral triangles into equilateral
  triangles}, Mathematical Proceedings of the Cambridge Philosophical Society
  \textbf{44} (1948), no.~4, 463--482.

\bibitem[VS81]{10.1007/3-540-10856-4_79}
L.~G. Valiant and S.~Skyum, \emph{Fast parallel computation of polynomials
  using few processors}, Mathematical Foundations of Computer Science 1981
  (Berlin, Heidelberg) (Jozef Gruska and Michal Chytil, eds.), Springer Berlin
  Heidelberg, 1981, pp.~132--139.

\bibitem[Wig19]{Wig19}
Avi Wigderson, \emph{Mathematics and computation: A theory revolutionizing
  technology and science}, Princeton University Press, 2019.

\bibitem[Wol08]{Wol08}
Paul~R. Wolfson, \emph{George boole and the origins of invariant theory},
  Historia Mathematica \textbf{35} (2008), no.~1, 37 -- 46.

\bibitem[Zei85]{ZEILBERGER198561}
Doron Zeilberger, \emph{A combinatorial approach to matrix algebra}, Discrete
  Mathematics \textbf{56} (1985), no.~1, 61--72.

\end{thebibliography}

\end{document}